\documentclass[11pt]{article}
\usepackage[hmargin=1in,vmargin=1in]{geometry}
\usepackage{hchang}
\usepackage{thm-restate}
\usepackage{cleveref}
\usepackage{comment}

\newtheorem{theorem}{Theorem}[section]
\newtheorem{lemma}[theorem]{Lemma}
\newtheorem{claim}[theorem]{Claim}
\newtheorem{observation}[theorem]{Observation}

\newtheorem{corollary}[theorem]{Corollary}

\newtheorem{definition}[theorem]{Definition}
\newtheorem{question}[theorem]{Question}

\newcommand{\parent}{\operatorname{Pa}}
\newcommand{\len}[1]{{\norm{#1}}}
\def\eps{\e}
\newcommand{\home}{\mathrm{home}}
\newcommand{\lca}{\mathcal{LCA}}

\DeclareMathOperator\dist{\delta}
\DeclareMathOperator*{\cost}{cost}
\DeclareMathOperator*{\tw}{tw}

\def\note#1{}

\begin{document}

\begin{titlepage}
	
\title{Covering Planar Metrics (and Beyond): O(1) Trees Suffice}


\author{%
Hsien-Chih Chang%
\thanks{Department of Computer Science, Dartmouth College. Email: {\tt hsien-chih.chang@dartmouth.edu}.}  
\and 
Jonathan Conroy%
\thanks{Department of Computer Science, Dartmouth College. Email: {\tt jonathan.conroy.gr@dartmouth.edu}}  
\and 
Hung Le%
\thanks{Manning CICS, UMass Amherst. Email: {\tt hungle@cs.umass.edu}}  
\and
Lazar Milenkovic%
\thanks{Tel Aviv University}  
\and
Shay Solomon%
\thanks{Tel Aviv University}  
\and
Cuong Than%
\thanks{Manning CICS, UMass Amherst. Email: {\tt cthan@cs.umass.edu}}  
}

\date{}

\maketitle
	
\thispagestyle{empty}
	
\begin{abstract}
While research on the geometry of planar graphs has been active in the past decades, many properties of planar metrics remain mysterious.   This paper studies a fundamental aspect of the planar graph geometry: covering planar metrics by a small collection of simpler metrics.  Specifically, a \EMPH{tree cover} of a metric space $(X, \delta)$ is a collection of trees, so that every pair of points $u$ and $v$ in $X$ has a low-distortion path in at least one of the trees. 

The celebrated ``Dumbbell Theorem''~\cite{ADMSS95} states that any low-dimensional Euclidean space admits a tree cover with $O(1)$ trees and distortion $1+\varepsilon$, for any fixed $\varepsilon \in (0,1)$. This result has found numerous algorithmic applications, and has been generalized to the wider family of doubling metrics~\cite{BFN19Ramsey}.  Does the same result hold for planar metrics?  A positive answer would add another evidence to the well-observed connection between Euclidean/doubling metrics and planar metrics. 

In this work, we answer this fundamental question affirmatively. Specifically, we show that for any given fixed $\varepsilon \in (0,1)$, any planar metric can be covered by $O(1)$ trees with distortion $1+\varepsilon$.  Our result for planar metrics follows from a rather general framework:  First we reduce the problem to constructing tree covers with \emph{additive distortion}.  Then we introduce the notion of \emph{shortcut partition}, and draw connection between shortcut partition and additive tree cover.  Finally we prove the existence of shortcut partition for any planar metric, using new insights regarding the grid-like structure of planar graphs.
To demonstrate the power of our framework: 
\begin{itemize}
    \item We establish additional tree cover results beyond planar metrics; in particular, we present an $O(1)$-size tree cover with distortion $1+\varepsilon$ for \emph{bounded treewidth} metrics;
    \item We obtain several algorithmic applications in planar graphs from our tree cover. 
\end{itemize}

The grid-like structure is a technical contribution that we believe is of independent interest. We showcase its applicability beyond tree cover by constructing a simpler and better embedding of planar graphs into $O(1)$-treewidth graphs with  small additive distortion, resolving an open problem in this line of research.
\end{abstract}

\end{titlepage}


\section{Introduction}

Research on the structure of planar graphs has provided many algorithmic tools such as separators~\cite{LT79} and cycle separators~\cite{Miller86}, $r$-divisions~\cite{Federickson87}, sphere-cut decomposition~\cite{DPBF09}, abstract Voronoi diagram~\cite{Cabello18}, and strong product theorem~\cite{DJMPUW20}, to name a few.
These structural results are rooted in the simple topology of planar graphs.  
Another line of important and complementary research from an algorithmic point of view is to understand the geometry of planar graphs and, more precisely, the metric spaces induced by shortest-path distances in planar graphs. Such metric spaces are called \emph{planar metrics}. Understanding the properties of metric spaces in general---the main subject of the metric embedding literature---has led to surprising algorithmic consequences~\cite{LLR95,Bar96,FRT04}. 
One may naturally expect that the simple topology of planar graphs would help in understanding planar metrics, from which we could significantly extend our algorithmic toolkit. 
Indeed, there have been a few such successful attempts, such as padded decompositions with $O(1)$ padding parameter~\cite{KPR93}, or embedding into $\ell_2$ with $O(\sqrt{\log n})$ distortion~\cite{Rao99} which has a matching lower bound~\cite{NR02}, and several other results~\cite{GNRS04,AFGN22}. 
However, many basic questions remain open.  A notable example is the $\ell_1$ embedding conjecture: can we embed a planar metric into $\ell_1$ with $O(1)$ distortion~\cite{GNRS04}?  More generally, what is the distortion for embedding into $\ell_p$ for any $p \geq 1, p\not=2$? 
(See~\cite{GNRS04,BC05,AFGN22,Filtser20}, and references therein for a host of other related questions.) 
These suggest that understanding the geometry of planar graphs is very challenging.  On the other hand, a deep understanding of the geometry of planar graphs often leads to remarkable algorithmic results. For example, a constant approximation for the sparsest cut problem in planar graphs would follow from a positive resolution of the $\ell_1$ embedding conjecture; the best-known algorithm achieving such constant factor approximation without relying on the unproven conjecture runs in quasi-polynomial time~\cite{CGKL21}.

Embedding is one important aspect of the geometry of planar metrics, but it might not be the only telling one.  As planar graphs are defined by drawing in the Euclidean plane, can we relate planar metrics to the 2D (or more generally, low-dimensional) Euclidean metric? 
Embedding is not illuminating in this respect: there exists an (unweighted) planar metric of $n$ points that requires distortion $\Omega(n^{2/3})$ for any embedding into $\mathbb{R}^2$~\cite{BDHM07}. More generally, as simple as the unweighted star graph of $n$ vertices, any embedding into an $O(1)$-dimensional Euclidean space requires distortion $n^{\Omega(1)}$ by a volume argument. 
This motivates us to look into the \emph{covering} aspect, namely, covering metrics by simpler metrics. 
Here \emph{tree metrics} are of special interest due to their simplicity and algorithmic applicability; in addition, Euclidean/doubling metrics were known to have a covering of constant size, as defined next.

We say that an edge-weighted tree $T$ with shortest-path distance $d_T$ is a \EMPH{dominating tree} of a metric space $(X,\delta_X)$ if $X$ is a vertex subset of $T$, and for every two points $x$ and $y$ in $X$, one has $\delta_X(x,y)\leq d_T(x,y)$. 
For any given parameter $\alpha\geq 1$, we say that a collection of trees, denoted by \EMPH{$\mathcal{F}$}, is a \EMPH{$\alpha$-tree cover} of $(X,\delta_X)$ if every tree in $\mathcal{F}$ is dominating, and for every two points $x\neq y$ in $X$, there is a tree $T$ in $\mathcal{F}$ such that $d_T(x,y) \leq \alpha\cdot \delta_X(x,y)$.  The \EMPH{size} of the tree cover is the number of trees in $\mathcal{F}$. 
Parameter $\alpha$ is called the \EMPH{distortion} of the tree cover. 
When $\alpha=1$, we say that the tree cover is an \EMPH{exact tree cover}. The notion of tree covers, and its variants, were studied in the past by many researchers~\cite{AP92,AKP94,ADMSS95,GKR01,BFN19Ramsey,FL22,KLMS22}. 

More than two decades ago, Arya et al.~\cite{ADMSS95} showed that any set of $n$ points in $\mathbb{R}^2$ admits a $(1+\e)$-tree cover with a constant number of trees for any fixed $\eps \in (0,1)$.%
\footnote{In this work, we consider $\eps$ to be a fixed constant. We only spell out the precise dependency on $\eps$ in theorem~statements.}
This result indeed holds for any Euclidean space of constant dimension and can be extended to any metric of constant doubling dimension~\cite{BFN19Ramsey}. These results naturally motivate the following question:

\begin{question}\label{ques:planar-cover}
Can planar metrics be covered by a constant number of trees with a multiplicative distortion $(1+\eps)$ for any fixed $\eps \in (0,1)$?
\end{question}

 A positive answer to \Cref{ques:planar-cover} would imply that planar metrics are similar to Euclidean/doubling metrics in the tree-covering sense. 
Furthermore, the tree cover serves as a bridge to transfer algorithmic results from Euclidean/doubling metrics to planar graphs/metrics; for example, routing, spanners, emulators, distance oracles, and possibly more. 
In their pioneering work~\cite{BFN19Ramsey}, Bartal, Fandina, and Neiman constructed a $(1+\e)$-tree cover of size $O(\log^2 n)$ for planar metrics; they left \Cref{ques:planar-cover} as an open problem. 

\subsection{Our contributions}

\paragraph{Tree cover for planar graphs and related results.} Our main result is a positive answer to \Cref{ques:planar-cover}. 

\begin{theorem}\label{thm:main}
    Let $G$ be any edge-weighted undirected planar graph with $n$ vertices.  For any parameter $\eps \in (0,1)$, there is a $(1+\eps)$-tree cover $\mathcal{F}$ for the shortest path metric of $G$ using $O(\eps^{-3} \cdot \log(1/\eps))$~trees.
\end{theorem}

We note a few related known results about tree covers. 
If we allow the distortion to be a rather large constant $C$, then Bartal, Fandina, and Neiman~\cite{BFN19Ramsey} showed that it is possible to construct a $C$-tree cover with $O(1)$ size. 
Gupta,  Kumar, and Rastogi~\cite{GKR01} constructed a tree cover with distortion $3$ using $O(\log n)$ trees. 
For distortion $1$ (exact tree cover), $O(\sqrt{n})$ trees is sufficient~\cite{GKR01}; furthermore, $\Omega(\sqrt{n})$ trees are necessary for some planar graphs if each tree in the tree cover must be a spanning tree of $G$.  
The main takeaway is that either the distortion is too high for a constant number of trees, or the number of trees have to depend on the number of vertices. 
Our tree cover constructed in \Cref{thm:main} simultaneously has $1+\e$ distortion and has no dependency on the size of the graph.

Our proof of \Cref{thm:main} follows from a \EMPH{general framework} that we introduce to construct tree covers. 
Our framework applies to minor-free graphs, which is a much broader class than planar graphs. 
At a high level, the framework has three steps. 
First, we devise a reduction of a tree cover with \emph{multiplicative distortion} to the construction of tree covers with \emph{additive distortion}.  Second, inspired by the scattering partition defined by Filtser~\cite{Filtser20B}, we introduce the notion of \emph{shortcut partition}, and show that the existence of a shortcut partition suffices to get a tree cover with additive distortion and a constant number of trees. 
The final step involves constructing shortcut partitions for graphs of interest. 
Applying our framework as is, we obtain $(1+\eps)$-tree covers for planar metrics using $O_\e(1)$ trees, with dependency exponential in $1/\eps$. 
Surprisingly, we manage to identify a new structural result by leveraging the fact that planar graphs are \emph{grid-like} in a formal sense, and then proceed to construct shortcut partition with additional properties. The additional structure allows us to reduce tree-cover size to a polynomial in $1/\eps$, as stated in \Cref{thm:main}. To keep our construction simple, we do not attempt to optimize the dependency on $1/\eps$; determining the exact dependency on $1/\eps$ for the size of the tree cover is an interesting question that we do not pursue in this work.

The (weighted) planar grids are often used as canonical examples in developing structural and algorithmic results for planar graphs. In many cases, the planar grids are ``hard'': for example, the worst-case bound on the separator/treewidth of planar graphs with $n$ vertices is realized by a $\sqrt{n} \times \sqrt{n}$ planar grid---this fact plays a central role in the Robertson-Seymour graph minor theory~\cite{RS86}. On the other hand, the planar grids also have a simple regular structure that often serves as a starting point for algorithmic developments (e.g.\ \cite{GKR01,CK15,FKS19}). 
Thus, our new grid-like structure for planar graphs may be used to leverage insights developed for planar grids to solve problems on general planar graphs.
In addition to our tree cover result (\Cref{thm:main}), we showcase another application:
embedding planar graphs into bounded treewidth graphs with small \emph{additive distortion.}  

More formally, given a weighted planar graph $G$ of diameter $\Delta$, we want to construct an embedding $f: V(G)\rightarrow V(H)$ into a graph $H$ such that $\delta_G(x,y)\leq d_H(f(x), f(y)) \leq \delta_G(x,y) + \eps \Delta$ and the treewidth of $H$, denoted by $\tw(H)$, is minimized.  
The work of Fox-Epstein, Klein, and Schild~\cite{FKS19}
was the first to show that $\tw(H) = O(\eps^{-c})$ for some constant $c$. 
However, the constant $c$ they obtained is very big---a rough estimate\footnote{In page 1084 of~\cite{FKS19}, the treewidth bound is at least $\eps^{-3}\cdot h(\eps^{-11})$ with $h(x)=O(x^{5})$ determined by Proposition~7.7.} from their paper gives $c\geq 58$---and the proof is extremely complicated, with several reduction steps to what they called the \emph{cage instances}, which themselves require another level of technicality to handle. Followup work~\cite{CFKL20,LS22} provided simpler constructions with  a linear dependency on $1/\eps$, at the cost of an extra $O((\log\log n)^2)$ factor in the treewidth. 
It remains an open problem how to construct an embedding that has the best of both: a \emph{simpler  construction} such that the treewidth has a reasonable dependency on $1/\eps$ but no dependency on $n$. 
We exploit the aforementioned grid-like structure to resolve this problem.

\begin{theorem}\label{thm:FKS-improved}
    Let $G$ be any given edge-weighted planar graph with $n$ vertices and diameter $\Delta$.  For any given parameter $\eps \in (0,1)$, we can construct in polynomial time an embedding of $G$ into a graph $H$ such that the additive distortion is $\eps \Delta$ and $\tw(H) = O(1/\eps^4)$.
\end{theorem}

\paragraph{Beyond planar graphs.}
We also obtain several other results as corollaries of the framework and its technical construction. 
In particular, we show that bounded treewidth graphs admit a tree cover of constant size as well.

\begin{theorem}\label{thm:treewidth}
    Let $G$ be any graph of treewidth $t$ with $n$ vertices. 
    For any given parameter $\eps \in (0,1)$, there is a $(1+\eps)$-tree cover $\mathcal{F}$ for the shortest path metric of $G$ using $2^{(t/\eps)^{O(t)}}$ trees.
\end{theorem}

\noindent \Cref{thm:treewidth} improves upon the tree cover construction by Gupta,  Kumar, and Rastogi~\cite{GKR01} who obtained a tree cover of size $O(\log n)$ for constant treewidth $t$ and constant $\eps$.  Again, our tree cover has no dependency on the number of vertices in $G$.

\bigskip
We also obtain the first non-trivial result for \emph{exact} tree covers in \emph{unweighted} minor-free graphs $G$ with small diameter.
Graphs of small diameter have been central in distributed computing. Specifically, the structure of unweighted \emph{planar graphs} of constant diameter has recently attracted attention from distributed computing community~\cite{GH15,GH16,HIZ16,LP19}.  
We obtain an \emph{exact} tree cover of constant size for such graphs.  Furthermore, our tree cover is \EMPH{spanning} (in the sense of metric embedding literature~\cite{AKPW95,AN12}); that is, every tree in the tree cover is a subgraph of $G$. Having a spanning tree cover is important: For example, it is useful in distributed computing, as messages can only be sent along the edges of the input graph. We believe our result for the exact tree cover is of independent interest.

\begin{theorem}\label{thm:planar-spanning}
    Let $G$ be any unweighted $K_r$-minor-free graph (for any constant $k$) with $n$ vertices and diameter $\Delta$.  
    There is an exact spanning cover $\mathcal{F}$ for the shortest path metric of $G$ using $2^{O(\Delta)}$ trees.
\end{theorem}
Gupta, Kumar, and Rastogi~\cite{GKR01} showed that there exists an $n$-vertex planar graph such that any spanning tree cover of the graph must have size $\Omega(\sqrt{n})$.
Our \Cref{thm:planar-spanning} circumvents their lower bound when the diameter of the graph is small.

\paragraph{Applications.}
One application of our tree cover result (\Cref{thm:main}) is to the design of $(1+\eps)$-approximate distance oracle for planar graphs. 
The goal is to construct a data structure of small space $S(n)$ for a given planar graph, so that each distance query can be answered quickly in time $Q(n)$ and the returned distance is within $1+\eps$ factor of the queried distance. 
Ideally, we want $S(n) = O(n)$ and $Q(n) = O(1)$. 
For doubling metrics, such a data structure was known for a long time~\cite{HM06}, but for planar metrics, attempts to obtain the same result were not successful for more than two decades~\cite{Thorup04,Klein02,WulffNilsen16} until the recent work by Le and Wulff-Nilsen~\cite{LW21}. 
Our \Cref{thm:main} provides a simple reduction to the same problem \emph{in trees}: given a distance query, query the distance in each tree and then return the minimum. This illustrates the power of our tree cover theorem in mirroring results in Euclidean/doubling metrics to planar metrics. 

The distance oracle constructed by our tree cover theorem has other advantages over that of Le and Wulff-Nilsen. 
First, their algorithm is very complicated, with many steps using the full power of the RAM model to pack $O(\log n)$ bits of data in $O(1)$ words to guarantee $O(n)$ space. 
In many ways, bit packing
can be viewed as ``abusing'' the RAM model. 
Second, in weaker models, such as the \emph{pointer machine model}---a popular and natural model in data structures---where bit packing is not allowed, their construction does not give anything better (and sometimes worse) than the older constructions. 
The best oracle in the pointer machine model was by Wulff-Nilsen~\cite{WulffNilsen16}, with $O(n \poly(\log \log n))$ space and $O(\poly(\log \log n))$ query time.  

We instead reduce to querying distances on trees, which we further reduce to the lowest common ancestor (LCA) problem.  LCA admits a data structure with $O(n)$ space and $O(1)$ time in the RAM model~\cite{HT84,BF00}, and $O(n)$ space and $O(\log \log n)$ query time in the pointer machine model~\cite{VanLeeuwen76}. (Harel and Tarjan~\cite{HT84} showed a lower bound of $\Omega(\log \log n)$ for querying LCA in the pointer machine model.) As a result, we not only recover the result by Le and Wulff-Nilsen~\cite{LW21}, but also obtain the best known distance oracle in the pointer machine model with $O(n)$ space and $O(\log \log n)$ query time.

\begin{theorem} \label{thm:app-distance-oracle} 
Suppose that any $n$-vertex tree admits a data structure for querying lowest common ancestors with space $S_{LCA}(n)$ and query time $Q_{LCA}(n)$. Then given any parameter $\eps \in (0,1)$, and any edge-weighted undirected planar graphs with $n$ vertices, we can design a $(1+\eps)$-approximate distance oracle with space $O(S_{LCA}(O(n))\cdot\tau(\e))$ and query time $O(Q_{LCA}(O(n))\cdot\tau(\e))$, where $\tau(\e) = \eps^{-3}\log(1/\eps)$.  
Consequently, we obtain:
\begin{itemize}\itemsep=0pt
    \item In the word RAM model with word size $\Omega(\log n)$, our oracle has space $O(n\cdot\tau(\e))$ and query time~$O(\tau(\e))$.
    \item In the pointer machine model, our oracle has space $O(n\cdot\tau(\e))$ and query time $O(\log \log n\cdot\tau(\e))$.
\end{itemize}
\end{theorem}

In \Cref{S:applications}, we also discuss other applications of our tree cover theorems. We construct the first $(1+\eps)$-emulator of planar graphs with linear size. We also obtain improved bounds for several problems in constructing low-hop emulators and routing studied in prior work~\cite{GKR01,CKT22,KLMS22}.

\subsection{Techniques}

Previous constructions of exact or $(1+\eps)$-multiplicative tree covers   rely on \emph{separators} of planar graphs.  
In particular, Gupta, Kumar, and Rastogi~\cite{GKR01} constructed an exact tree cover of size $O(\sqrt{n})$ by first finding a separator of size $O(\sqrt{n})$, creating $O(\sqrt{n})$ shortest path trees each rooted at a vertex in the separator, and recursing on the rest of the graph.  
The $O(\log^2 n)$-size construction of Bartal, Fandina, and Neiman~\cite{BFN19Ramsey} follows the same line, but uses \emph{shortest-path separators} instead. 
More precisely, they find a balanced separator consisting of $O(1)$ many shortest paths. 
Then for each shortest path, they construct $O(\log n)$ trees by randomly ``attaching'' remaining vertices (via new edges) to the $O(1/\e)$ portals (\`{a} la Thorup~\cite{Thorup04}) along the shortest path, and recurse. 
As the recursion depth is $O(\log n)$, they obtain a $(1+\e)$-tree cover of size $O(\log^2 n)$. 
It is possible to remove the $O(\log n)$ factor in the random attachment step by a more careful analysis, but $\Omega(\log n)$ is the barrier to the number of trees for recursive constructions using balanced separators. (Indeed, many other constructions in planar graphs suffer from the same $\log n$ barriers~\cite{Klein02,Thorup04,EKM13,CGH16,CFKL20}, some of which were overcome by very different techniques~\cite{FKS19,LW21,LS22,CKT22}.)

Here we devise a new technical framework to overcome the $\log n$ barrier in the construction of tree covers. 
The first step in the framework is a reduction to a tree cover with additive distortion. 
Given $\beta > 0$, we say that a set of trees $\mathcal{F}$ is a tree cover with \EMPH{additive distortion} $+\beta$ if every tree in $\mathcal{F}$ is dominating and for any $x,y \in V(G)$, there exists a tree $T\in \mathcal{F}$ such that $d_T(x,y)\leq \delta_G(x,y) + \beta$. 
We say that the tree cover $\mathcal{F}$ is \EMPH{$\Delta$-bounded} if the diameter of every tree in $\mathcal{F}$ is at most $\Delta$. 
In \Cref{S:add-to-mul}, we show that the reduction to tree cover with additive distortion $+\e\Delta$ only incurs tiny loss of $O(\log(1/\eps))$ factor on the size of the tree cover.

\begin{restatable}[Reduction to additive tree covers]{lemma}{MultToAdd}
\label{lm:reduction} 
Let $(X,\delta_X)$ be a $K_r$-minor-free metric
(for any constant $r$)
with $n$ points. For any parameter $\e \in (0,1)$, suppose that any $K_r$-minor-free submetric induced by a subset $Y\subseteq X$ with diameter $\Delta$ has an $O(\Delta)$-bounded tree cover $\mathcal{F}$ of size $\tau(\eps)$ of additive distortion $+\eps \Delta$. Then $(X,\delta_X)$ has a tree cover of size $O(\tau(O(\eps))\cdot \log(1/\eps))$ with multiplicative distortion $1+\eps$. 
\end{restatable}

To prove \Cref{lm:reduction}, we introduce the notion of a \EMPH{hierarchical pairwise partition family} (HPPF), and show that the reduction follows from an HPPF. 
An HPPF is a collection of hierarchies of partitions, where each partition at level $i$ of a hierarchy in the family is a partition of the planar metric into clusters of diameter roughly $\Theta(1/\eps^i)$. The HPPF has a special property called the \emph{pairwise property}: for every pair $(x,y)\in X$, there is a partition at some level in a hierarchy in the family such that both $x$ and $y$ are contained in some cluster $C$ of the partition, and the diameter of $C$ is roughly $\Theta(\delta_X(x,y))$. 
We show that HPPF can be constructed from a hierarchical partition family studied in  previous works~\cite{KLMN04,BFN19Ramsey}. 

\medskip
Next we focus on constructing a tree cover for planar graphs of diameter $\Delta$  with additive distortion $+\eps \Delta$.  
Our goal is to construct a clustering where each cluster has small diameter $O(\eps \Delta)$. 
After contracting all these clusters, we obtain a cluster graph $\check{G}$, which we would like to treat as an unweighted graph. 
Furthermore, the clusters we constructed will ensure that $\check{G}$ has a small unweighted diameter, independent to the size of the original graph. 
We formalize the properties we need through the notion of an \EMPH{$(\e,h)$-shortcut partition}: every cluster in the partition has diameter at most $\e\Delta$ (where $\Delta$ denotes the diameter of $G$), and (loosely speaking) the hop diameter of $\check{G}$ is at most $h$; see \Cref{def:shortcut-partition} for a more formal definition. 
Now for this special case when cluster graph $\check{G}$ has constant diameter $h = O_\e(1)$, we devise an inductive construction that runs in $h$ rounds, where in the $i$-th round we only preserve paths in $\check{G}$ with distances up to $i$. 
The inductive construction not only gives us an exact tree cover of constant size (as claimed in \Cref{thm:planar-spanning}), but also has other attractive properties. 
One property, which then plays a crucial role in our construction of a tree cover with additive distortion, is what we call the \EMPH{root preservation property} (see Section~\ref{S:exact-cover}):
every tree in the tree cover can be decomposed into a forest, where each tree in the forest is a BFS tree, and the distance between any two vertices is preserved by a path going through the root of a tree in some forest.
Our construction only makes use of the fact that any minor of the input graph has bounded \emph{degeneracy}; hence the result can be extended to any minor-free graphs. 
Furthermore, the tree cover we get is a \emph{spanning} tree cover, which means that the trees are subgraphs of the input graph. We emphasize again that our construction is the first that does not use balanced separators.

Here comes another issue: if we construct (an exact) tree cover of $\check{G}$, we must somehow turn it into a tree cover for $G$. 
A simple idea is to take a tree, say $\check{T}$, in a tree cover of  $\check{G}$,
and \emph{expand} each (contracted) vertex $\check{c}$ in $\check{T}$ with the corresponding cluster $C$ in $G$: 
attach every vertex $v\in C$ to $\check{c}$ by an edge of length roughly $\eps \Delta$ (the diameter of $C$). 
However, the shortest path from $u$ to $v$ in $G$ could intersect multiple clusters, and simply expanding clusters as described above would incur a very large additive distortion (remember that we can only tolerate up to $+\eps \Delta$ distortion). This is where the root preservation property comes to the rescue:
we can replace every rooted tree with a (Steiner) star centered at the root, where the star edges are weighted \emph{according to distance on $G$}. 
The root preservation property
implies the only relevant paths in the tree are those that pass through the root---and transforming trees into stars preserves these distances up to $+O(\e\Delta)$ distortion.
As a result, we are able to show a black-box reduction from $(\e,h)$-shortcut partition to tree covers with additive distortion. This reduction works for any minor-free graph.

\begin{restatable}{theorem}{PartitionToCover} \label{thm:partition-to-cover}
Let $G$ be an (undirected) weighted minor-free graph with diameter $\Delta$. Suppose that  for any $\e > 0$ there is an $(\e,f(\eps))$-shortcut partition for $G$ for some function $f(\e)$ depending only on $\e$, Then $G$ admits a tree cover of size $2^{O(f(\e))}$ with additive distortion $+O(\e\Delta)$.
\end{restatable}

The construction of tree cover now is reduced to constructing an $(\eps ,f(\eps))$-shortcut partition. 
For graphs with treewidth $t$, we provide a 
construction with $f(\eps) = (t/\eps)^{O(t)}$; see \Cref{S:weak-scattering-treewidth} for details. 
This, together with \Cref{lm:reduction} and \Cref{thm:partition-to-cover}, implies our $O(1)$-size tree cover for bounded-treewidth graphs (\Cref{thm:treewidth}).
For planar graphs, constructing an $(\eps,f(\eps))$-shortcut partition is much more difficult. 
Our key insight here is that 
planar graphs are \emph{grid-like}.

\paragraph*{Informal discussion of grid-like structure of planar graphs.}  
A planar grid graph can be decomposed into an ordered collection of columns, where each column is a collection of vertices. We identify two important properties of grid graphs:
\begin{itemize}
    \item Each column is a shortest path.
    \item Every edge goes between two vertices either in the same column or in consecutive columns. (In particular, this implies that every path from a column $C_1$ to a column $C_2$ passes through every column in between $C_1$ and $C_2$.)
\end{itemize}
We would like to say that every planar graph admits a partition
into clusters of diameter $\e\Delta$, such that if we contract all clusters into supernodes, the contracted graph satisfies the above properties (for example, $\check{G}_1$ in Figure~\ref{fig:gridlike}). However, this may not always be the case.

\begin{figure}[h!]
\centering
\includegraphics[width=0.7\textwidth]{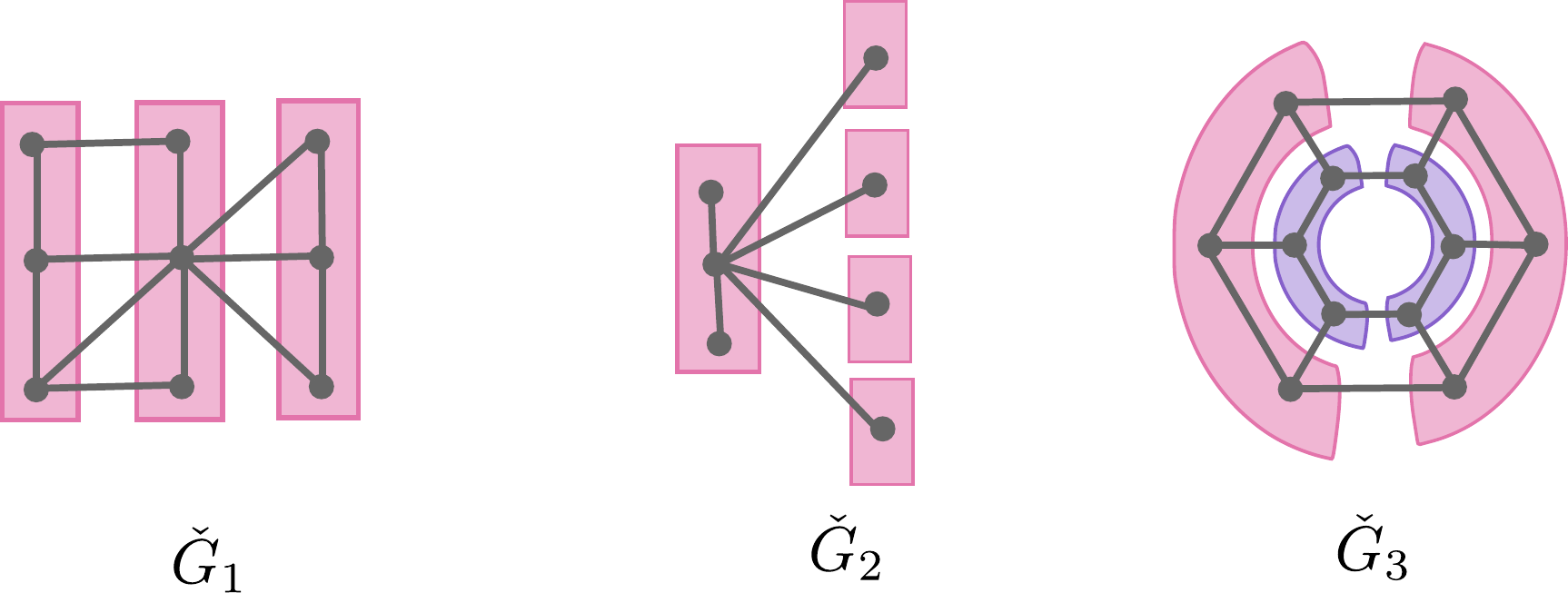}
\caption{Three planar graphs, with ``columns'' highlighted in pink. Graph $\check{G}_1$ looks like a grid; graph $\check{G}_2$ requires tree-like column ordering; graph $\check{G}_3$ requires column nesting. The nested columns in $\check{G}_3$ are highlighted in purple.}
\label{fig:gridlike}
\end{figure}

\medskip \noindent \textbf{Issue 1 (Trees):} If the contracted graph is a star, it does not look like a grid (for example, $\check{G}_2$ in Figure~\ref{fig:gridlike}). There is no way to partition the vertices into columns and assign an ordering such that edges only occur between consecutive columns. As such, we relax the guarantee. Instead of assigning a total order to the columns, we assign a tree structure to the columns such that edges occur only between columns that are adjacent in the tree.

\medskip \noindent \textbf{Issue 2 (Nesting):} If the contracted graph is a circular grid graph, it
also does not look like a grid (for example, $\check{G}_3$ in Figure~\ref{fig:gridlike}). We cannot partition the vertices into columns and assign a valid ordering to the columns. To deal with this, we allow some subgraphs to ``sit in between'' two adjacent columns. These subgraphs (recursively) satisfy the grid-like structure with nesting. We guarantee that there are few layers of nesting. 

\bigskip 
\noindent We show that these are essentially the only ways in which planar graphs can violate the grid-like~property.
To deal with nesting, we decompose the graph into a tree structure called \EMPH{gridtree hierarchy}, such that each node of the hierarchy is associated with a grid-like structure called \EMPH{gridtree}. The hierarchy has depth $O(1/\e)$; i.e., there are few layers of nesting. Each gridtree is a tree, such that each node is associated with a subgraph of $G$ we call a \emph{column}.
Each column contains vertices of distance at most $\eps \Delta$ from a shortest path within;
we say the width of the hierarchy is $\eps \Delta$.
The gridtrees in the hierarchy are reminiscent of the recursive structures built for sparse covers on planar graphs in the work of Busch, LaFortune, and Tirthapura~\cite{BLT14}. 
We show that, given a gridtree hierarchy with width $\eps^2 \Delta$, we can construct an $(\eps, O(1/\eps^2))$-shortcut partition of $G$.
Again with \Cref{lm:reduction} and \Cref{thm:partition-to-cover} we get a tree cover of size $2^{O(1/\eps^2)}$ out of the box using our framework.  
Perhaps surprisingly, we can improve upon the exponential size bound by working with the gridtree directly, using techniques inspired by our earlier reduction to tree cover in low-diameter graphs.
Roughly speaking, shortcut partitions enable us to reduce to the special case of bounded-diameter planar graphs; the extra properties of the gridtree enable us to reduce (at least in spirit) to the special case of bounded-diameter planar grids.
Using the gridtree hierarchy, we construct a tree cover with additive distortion $+\eps \Delta$ and size $O(1/\eps^3)$. Thus, by \Cref{lm:reduction}, we obtain a tree cover of size $O(\eps^{-3}\cdot \log(1/\eps))$ as claimed in \Cref{thm:main}.  

\paragraph{Embedding into bounded treewidth graphs with additive distortion.} 

To demonstrate the power of our new grid-like structure, we prove that any planar graph $G$ of diameter $\Delta$ can be embedded into a bounded treewidth graph with additive distortion $+O(\e\Delta)$.
We construct a shortcut partition $\mathfrak{P}$ and tree cover 
$\mathcal{F}$ (which actually is a collection of \emph{spanning forests} with the root preservation property) 
for $G$ with the following extra property:
\begin{itemize}
    \item \textnormal{[Disjointness.]}
    No two trees in any forest $F$ in $\mathcal{F}$ contain vertices from the same cluster in~$\mathfrak{P}$.
\end{itemize}

Our embedding has three steps. 
First, we contract each cluster of $\mathfrak{P}$ into a supernode and obtain $\check{G}$. 
The low-hop property of shortcut partition $\mathfrak{P}$ implies that $\tw(\check{G}) = O(1/\eps)$. 
Second, for each vertex $\check{c}$ in $\check{G}$, let $C$ be the corresponding cluster in $\mathfrak{P}$. 
We attach (a copy of) each vertex $v$ in $C$ to $\check{c}$ via a single edge. 
At this point, treewidth of $\check{G}$ remains $O(1/\eps)$, but the distortion could be large. 
Third, we augment $\check{G}$ by adding more edges: For each rooted tree $T$ in each forest $F$ in $\mathcal{F}$, we add an edge from (the copy of) the root of $T$, say $r$, in $\check{G}$ to (the copy of) every other vertex of $T$, say $v$, by an edge of weight $\delta_G(r,v)$. 
Let $\hat{G}$ be the resulting graph. 
The root preservation property of tree cover $\mathcal{F}$ implies that the distortion is $+O(\eps\Delta)$, as required by \Cref{thm:FKS-improved}. 
However, it is unclear why $\hat{G}$ has small treewidth. 
Here we use the disjointness property, which intuitively implies that the augmentation only happens at ``disjoint local areas'' of $\hat{G}$. 
We formalize this intuition via a key lemma (\Cref{lem:extend-forest}) showing that $\tw(\hat{G}) = O(|\mathcal{F}|\cdot \tw(\check{G})) = O(\eps^{-3} \cdot \eps^{-1}) = O(\eps^{-4})$. \Cref{thm:FKS-improved} now follows.

\subsection{Related work}

A closely related notion of tree cover is a \emph{Ramsey tree cover}. 
In a Ramsey $t$-tree cover $\mathcal{F}$, each vertex $v$ is associated with a tree $T_{\home(v)}\in \mathcal{F}$, such that the distance from $v$ in  $T_{\home(v)}$ to every other vertex is an approximation of the original distance up to a factor of $t$. 
Thus, Ramsey tree covers have stronger guarantees than ordinary tree covers. 
Both tree covers and Ramsey tree covers for general metrics were studied in the past~\cite{BLMN05,TZ01b,BFN19Ramsey}. 
In general metrics, bounds obtained for tree covers and Ramsey tree covers are almost the same. 
It is possible to construct a Ramsey tree cover with tradeoff between size $k$ and distortion $\tilde{O}(n^{1/k})$~\cite{BFN19Ramsey}, or of size $\tilde{O}(n^{1/t})$ and distortion $O(t)$~\cite{MN06}; these bounds are almost optimal~\cite{BFN19Ramsey}.

However, Ramsey tree covers with a constant number of trees and constant distortion do not exist in planar and doubling metrics. In particular, for metrics which are \emph{both planar and doubling} and for any distortion $\alpha \geq 1$, Bartal \etal~\cite{BFN19Ramsey}  showed that the number of trees in the Ramsey tree cover must be $n^{\Omega(1/(\alpha \log(\alpha)))}$. This sharply contrasts with our result for (non-Ramsey) tree covers in \Cref{thm:main}.


\section{Tree cover for graphs with good shortcut partition}

A \EMPH{tree cover $\mathcal{F}$} of an edge-weighted planar graph $G$ is a collection of trees, so that every pair of vertices $u$ and $v$ in $G$ has a \emph{low-distortion}
path in at least one of the trees in $\mathcal{F}$.
Specifically, 
\begin{itemize}
\item tree cover $\mathcal{F}$ has \EMPH{$\alpha$-multiplicative distortion} if there is a tree $T$ in $\mathcal{F}$ satisfying
\[
\dist_G(u,v) \le \dist_T(u,v) \le \alpha \cdot \dist_G(u,v).
\]

\item tree cover $\mathcal{F}$ has \EMPH{$\beta$-additive distortion} if there is a tree $T$ in $\mathcal{F}$ satisfying
\[
\dist_G(u,v) \le \dist_T(u,v) \le \dist_G(u,v) + \beta.
\]
\end{itemize}
Sometimes it is easier to describe the construction in terms of \EMPH{forest covers}: that is, instead of a collection of trees, we allow a collection of \emph{forests} to be in the cover of $G$. Let $\Delta$ be the diameter of $G$. Recall that a tree cover $\mathcal{T}$ is \EMPH{$\Delta$-bounded} if every tree of $\mathcal{T}$ has diameter at most $\Delta$; we say that a forest cover $\mathcal{F}$ is $\Delta$-bounded if every tree in every forest of $\mathcal{F}$ is $\Delta$-bounded. We remark that one can easily construct an $O(\Delta)$-bounded tree cover from an $O(\Delta)$-bounded forest cover: Simply connect the tree components within each forest into a single tree in a star-like way, assigning weight $\Delta$ to the newly added star edges.
As the diameter of the graph is $\Delta$, each newly constructed tree is a dominating tree. 

Our main result of this section is to show that in order to construct a tree cover of constant size with $(1+\e)$-multiplicative distortion for arbitrary weighted planar graph $G$ with diameter $\Delta$, it is sufficient to do the following:
(1) Reduce the problem to constructing tree covers of constant size with $\e\Delta$-additive distortion.
(2) Find a shortcut partition for $G$ into \emph{clusters} such that every cluster has diameter $\e\Delta$, and contract all clusters into supernodes to form the \emph{cluster graph} $\check{G}$, such that there is a shortest path between any two nodes that has at most $O_\e(1)$ edges.
(3) Construct a constant-size tree cover for the cluster graph $\check{G}$, which has bounded hop-diameter.

After introducing some terminologies (Sections~\ref{SS:prelim} and~\ref{S:exact-cover}),
we first prove that the above reduction strategy works in Section~\ref{SS:reduction} (\Cref{thm:large-diam-cover}), then in Section~\ref{SS:expand-forest}  we prove the existence of tree cover for planar graph with bounded hop-diameter (\Cref{thm:bdd-diameter-tree-cover}).  (Surprisingly, the tree cover constructed preserves the distance \emph{exactly} without distortion.)
In Sections~\ref{S:grid-scattering} and~\ref{S:gridtree-planar} we prove the existence of shortcut partitions for planar graphs.
We construct tree cover for planar graphs whose size is polynomially dependent on $1/\e$ in Section~\ref{S:treecover-gridlike}.
Next, in Section~\ref{S:weak-scattering-treewidth} we construct shortcut partition for bounded-treewidth graphs.
In Section~\ref{S:treewidth-embedding} we prove that any planar graph embeds into a bounded-treewidth graph with additive distortion.
Finally, in Section~\ref{S:add-to-mul} we present the full details of the reduction from multiplicative to additive distortion for tree covers.
We conclude the paper with some applications (Section~\ref{S:applications}).

\subsection{Shortcut partitions}
\label{SS:prelim}

Throughout this section and the rest of the paper, let $\e$ and $\Delta$ be fixed parameters, with $\Delta > 1$
and $0 < \e < 1$.  (Later on in the paper we might recurse on some subgraph $H$ of $G$. In such case we will consistently use $\Delta$ for the diameter of $G$; in particular, the diameter of $H$ might increase and be bigger than $\Delta$.)
Let $G$ be an (undirected) weighted planar graph with diameter $\Delta$.
When a graph $H$ is a subgraph of $G$, we write \EMPH{$H \subseteq G$}.
For any subgraph $H \subseteq G$, let \EMPH{$G[H]$} denote the subgraph of $G$ induced by the vertices of $H$.
A \EMPH{cluster} $C$ is a subset of vertices in $G$ such that the induced subgraph $G[C]$ is \emph{connected}. 
A \EMPH{clustering} of a planar graph $G$ is a partition of the vertices of $G$ into clusters $\mathcal{C} \coloneqq \set{C_1, \ldots, C_m}$.
Let \EMPH{$\check{G}$} be the \EMPH{cluster graph} of $G$ with respect to $\mathcal{C}$, where each cluster in $\mathcal{C}$ is contracted to a \EMPH{supernode}.
We always treat $\check{G}$ as an unweighted graph; to emphasize this, we use the terms \EMPH{hop-length}, \EMPH{hop-distance}, and \EMPH{hop-diameter} as opposed to length, distance, and diameter when referring to $\check{G}$. An \EMPH{$h$-hop path} is a path with $h$~edges.

Our general goal is to construct a clustering for $G$ such that the \emph{diameter} of each cluster is small. 
One can measure the diameter of a cluster in two ways:
\begin{itemize}
\item a cluster $C$ has \EMPH{strong diameter} $D$ if $\dist_{G[C]}(u,v) \leq D$ for any two vertices $u$ and $v$ in $C$;
\item a cluster $C$ has \EMPH{weak diameter} $D$ if $\dist_{G}(u,v) \leq D$ for any two vertices $u$ and $v$ in $C$.
\end{itemize}
Notice that cluster $C$ has strong diameter $D$ implies it has weak diameter $D$ as well; the main difference is whether a shortest path between $u$ and $v$ is within the cluster $C$ itself, or within the whole graph $G$.

\begin{definition} \label{def:shortcut-partition}
An \EMPH{$(\e,h)$-shortcut partition}%
\footnote{Arnold Filtser~\cite{Filtser20B} introduced the notion of \emph{scattering partition}.  In a scattering partition it is required that \emph{every shortest path} of length $\alpha\cdot \e\Delta$ intersects at most $O(\alpha)$ clusters, which is stronger than shortcut partition.  Scattering partition is conjectured to exist for any minor-free graphs~\cite[Conjecture~1]{Filtser20B}.}
is a clustering $\mathcal{C} = \set{C_1, \ldots, C_m}$ of $G$ such that:
\begin{itemize}
    \item \textnormal{[Diameter.]}
    the \emph{strong} diameter of each cluster $C_i$ is at most $\e\Delta$, where $\Delta$ is the diameter of $G$; 
    
    \item \textnormal{[Low-hop.]}
    for any vertices $u$ and $v$ in $G$, there is an approximate shortest-path $\pi$ between $u$ and $v$ in $G$ with length at most $(1+\e) \cdot \dist_G(u,v)$, and there is a path $\check{\pi}$ in the cluster graph $\check{G}$ between the clusters containing $u$ and $v$ such that (1) $\check{\pi}$ has hop-length at most $h$, and (2) $\check{\pi}$ only contains clusters that have nontrivial intersection with $\pi$.

\end{itemize}
\end{definition}

\noindent Notice that the low-hop property does \emph{not} guarantee that path $\pi$ between $u$ and $v$ intersects at most $h$ clusters. Rather, it guarantees that there is some $h$-hop path between the two clusters containing $u$ and $v$ respectively in the cluster graph;
see Figure~\ref{fig:low_hop}.
We remark that the cluster graph obtained by contracting every cluster in an $(\e, h)$-shortcut partition has hop-diameter at most $h$.

\begin{figure}[htbp]
    \centering
    \includegraphics[width=0.4\textwidth]{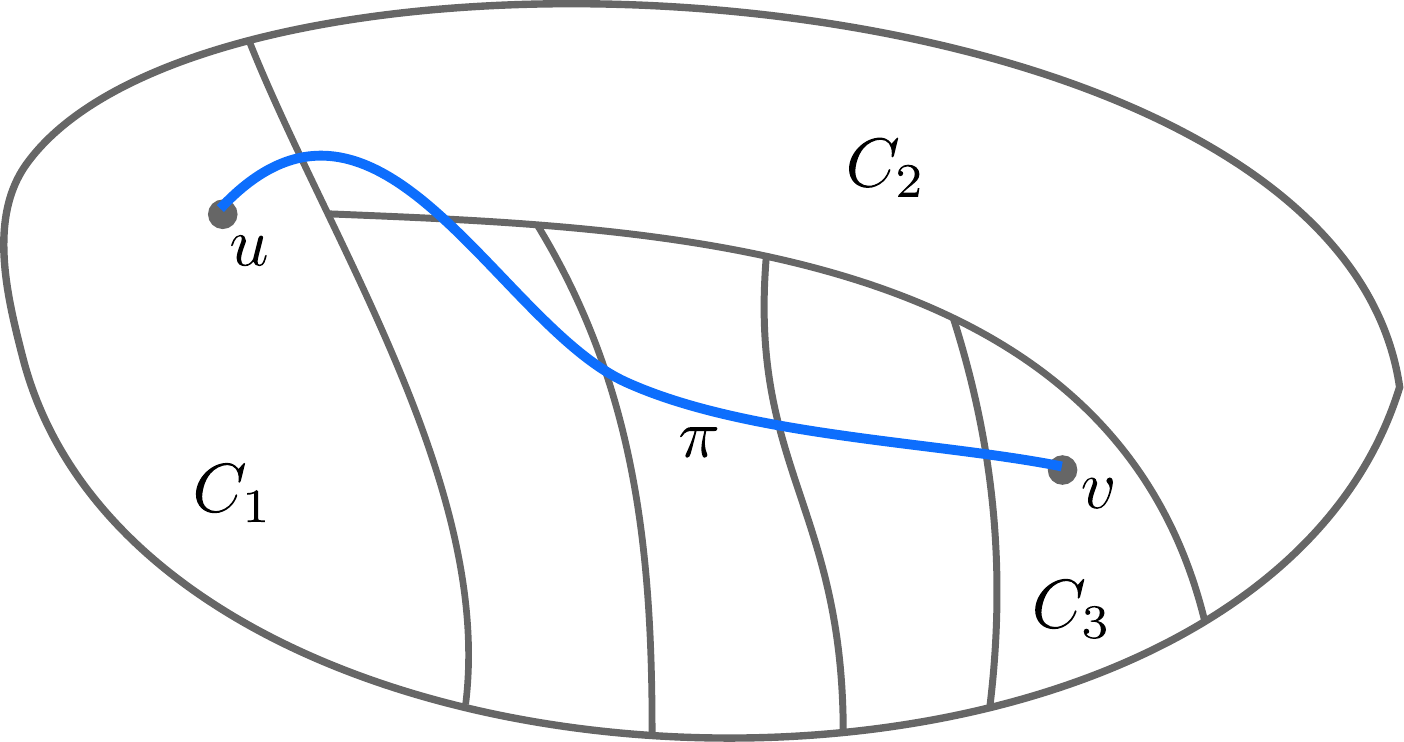}
    \caption{A graph partitioned into clusters. There is a path $\pi$ between vertices $u$ and $v$. Path $\pi$ intersects 6 clusters. In the cluster graph, there is a 2-hop path $\check{\pi} = (C_1, C_2, C_3)$ that only contains clusters intersecting $\pi$.}
    \label{fig:low_hop}
\end{figure}

\noindent The following theorem summarizes our framework of constructing (additive) tree covers for graphs with shortcut partitions; it is a restatement of \Cref{thm:partition-to-cover} with more details. The proof can be found in Section~\ref{SS:reduction} after we introduce the necessary terminologies.

    \begin{theorem}\label{thm:large-diam-cover}
    Let $G$ be an (undirected) weighted minor-free graph with diameter $\Delta$. 
    Suppose that 
    (1) for any $\e > 0$ there is an $(\e,f(\e))$-shortcut partition for $G$ for some function $f(\e)$ depending only on $\e$, and
    (2) the cluster graph $\check{G}$ with respect to the shortcut partition
    has a forest cover of $2^{O(f(\e))}$ size, satisfying the \emph{root preservation property} defined in \Cref{S:exact-cover}.
    Then $G$ admits an $O(\Delta)$-bounded forest cover with additive distortion $+O(\e\Delta)$, with $2^{O(f(\e))}$ forests.
\end{theorem}


\subsection{Exact spanning tree cover for small diameter minor-free graphs}
\label{S:exact-cover}

Let $G$ be an \emph{unweighted}, undirected $K_r$-minor-free graph with diameter $\Delta$ for some constant $r$. 
For every $k = 1, \ldots, \Delta$, we will construct a set of forests \EMPH{$\mathcal{F}_{k}$} (each of which is a \emph{subgraph} of $G$) such that if $\dist_G(u,v) \leq k$, then there is some forest $F \in \mathcal{F}_k$ such that $\dist_F(u,v) = \dist_G(u,v)$. 
In other words, the shortest path between $u$ and $v$ in $F$ is also a shortest path between $u$ and $v$ in $G$.
As a base case, we use the fact that the edges of a $K_r$-minor-free graph $G$ can be partitioned into $O_r(1)$ star forests (the \emph{arboricity} and \emph{star arboricity} of $K_r$-minor-free graphs are both $O(r\sqrt{\log r})$~\cite{Kostochka1982}; this set of forests is $\mathcal{F}_1$. 
Roughly speaking, we construct $\mathcal{F}_{k+1}$ from $\mathcal{F}_{k}$ by ``expanding'' all trees in $\mathcal F_k$ to include one more edge. 
As a first attempt at implementing this idea, we might try replacing every tree $T \in \mathcal{F}_k$ with $T \cup N_G(T)$, where $N_G(T)$ denotes the set of all edges incident to $T$ in $G$. 
Expanding all trees in this way would clearly satisfy the distance-preserving guarantee for $k+1$, but the expanded trees might not be vertex-disjoint (i.e., expanding a forest this way might produce a graph that is not a forest).   

Perhaps surprisingly, we demonstrate that the suggested approach isn't far off, and the expansion can indeed be simulated with a set of $O(1)$ forests (see Lemma~\ref{lem:expand-forest}). 
This idea gives a forest cover for $O(1)$-diameter minor-free graphs. 
In fact, the forest cover will have a very specific structure: each tree in each forest will be a BFS tree in the original graph $G$, and the distance between any two vertices in $G$ is preserved in the forest cover by a path going through the root of some tree.
Notice that this implies that the forest cover is \emph{spanning}: in other words, it uses no Steiner points or edges.  
The spanning property is helpful for later applications.

\paragraph{Root expansions.}
Let $G$ be a planar graph with vertices $V(G)$ and edges $E(G)$, and $H$ be a subgraph of $G$.
Let \EMPH{$\dist_H(u,v)$} denote the distance in $H$ between vertices $u$ and $v$ in $V(H)$. 
If $T \subseteq G$ is a tree rooted at $r$, we say that $T$ is a \EMPH{BFS tree} if it is a BFS tree on $G[T]$ (that is, if $\dist_T(r, v) = \dist_{G[T]}(r, v)$ for every vertex $v$ of $T$.  In general, $\dist_{G[T]}(r,v)$ may still be larger than $\dist_G(r,v)$.)
We say that a forest $F \subseteq G$ is a \EMPH{BFS forest} if every tree in $F$ is a BFS tree.
    Let $T \subseteq G$ be a BFS tree rooted at $r$, and let $\pi$ be a path in $G$. 
    The BFS tree $T$ \EMPH{preserves} $\pi$ if (i) every vertex in $\pi$ is in $V(T)$, and (ii) $\pi$ passes through $r$.

\begin{observation}
\label{Obs:path-in-tree}
If tree $T$ preserves a shortest path $\pi$ in $G$ between vertices $u$ and $v$, then $\dist_T(u,v) = \dist_G(u,v)$, even though $\pi$ might not be contained in $T$.
\end{observation}

We say that a BFS forest $F$ preserves $\pi$ if $F$ contains a BFS tree that preserves $\pi$; a set $\mathcal{F}$ of BFS forests preserves $\pi$ if $\mathcal{F}$ contains a BFS forest that preserves $\pi$.
If $G$ is a graph with diameter $\Delta$ and $\mathcal{F}$ is a forest cover for $G$, we say that $\mathcal{F}$ has the \EMPH{root preservation property} if $\mathcal{F}$ is a set of BFS forests of $G$ that preserves every path in $G$ of length at most $\Delta$.

We now formalize the key idea of ``expanding'' a tree, as described at the beginning of the section.
For any path $\pi$ in $G$, we define the \EMPH{prefix} of $\pi$ to be the path containing all but the last edge of $\pi$.

\begin{definition}\label{def:root-exp}
    Let $T \subseteq G$ be a BFS tree with root $r$. 
    A \EMPH{root expansion} of $T$ is a set of BFS trees $\mathcal{T}$ such that for every path $\pi$ in $G$ whose prefix is preserved by $T$, the set $\mathcal{T}$ preserves $\pi$.
\end{definition}

A \EMPH{root expansion} of a BFS forest $F$ is a set of BFS forests $\mathcal{F}$, each consisting of vertex-disjoint BFS trees, such that the set of all trees in $\mathcal{F}$ forms a root expansion for each tree in $F$.
The \EMPH{size} of the expansion is the number of forests in the set.
We emphasize that the trees in the same BFS forest have to be vertex-disjoint from each other,
while trees in different BFS forests may not be vertex-disjoint.

\begin{restatable*}{lemma}{expandforest}%
\label{lem:expand-forest}
    For every BFS forest $F \subseteq G$,
    there is a root expansion of $F$\! of size $O(1)$.
\end{restatable*}

Postponing the proof of Lemma~\ref{lem:expand-forest} to \Cref{SS:expand-forest}, we first prove our main result of this section.

\begin{theorem}
\label{thm:bdd-diameter-tree-cover}
    Let $G$ be an unweighted planar graph with diameter $\Delta$. Then there is a set of BFS forests~$\mathcal{F}$ of size $2^{O(\Delta)}$, such that $\mathcal{F}$ preserves every path in $G$ of length at most $\Delta$. Consequently,
    for every $u,v \in V(G)$, there is some forest $F \in \mathcal{F}$ where $\dist_F(u,v) = \dist_G(u,v)$.
\end{theorem}

\begin{proof}    
    As $G$ is $K_r$-minor-free, the edges of $G$ can be covered by $O_r(1)$ forests by Nash-Williams theorem~\cite{Nash64}. 
    These forests can be converted into $O_r(1)$ star forests. 
    Each star is a BFS tree in $G$. 
    Let $\mathcal{F}_1$ be such a set of BFS forests. 
    For each $k \in 2, \ldots, \Delta $, let $\mathcal{F}_k$ be the set containing all BFS forests from the root expansions of each BFS forest in $\mathcal{F}_{k-1}$ provided by \Cref{lem:expand-forest}. Return $\mathcal{F}_\Delta$.

    \smallskip \noindent
    By induction, each $\mathcal{F}_k$ has size ${O(1)}^k$; setting $k = \Delta$, we find that $\mathcal{F}_\Delta$ has size $2^{O(\Delta)}$. 
    Again by induction, we have that
    for every path $\pi$ of length at most $k$, there is some forest in $\mathcal{F}_k$ that preserves $\pi$, by definition of the root expansion.  (The base case when $k=1$ is immediate since each length-1 path belongs to one of the stars of $\mathcal{F}_1$.)
    As $G$ has diameter $\Delta$, all shortest paths in $G$ are preserved by $\mathcal{F}_\Delta$.
\end{proof}

\subsection{Proving Theorem~\ref{thm:large-diam-cover}} 
\label{SS:reduction}

\noindent Before we proceed to prove Lemma~\ref{lem:expand-forest}, first we prove the correctness of our reduction to bounded hop-diameter case.

\begin{proof}[of Theorem~\ref{thm:large-diam-cover}]
    By assumption (1) of Theorem~\ref{thm:large-diam-cover}, there is an $(\e,f(\e))$-shortcut partition for $G$.
    As each cluster is connected, the cluster graph $\check{G}$ obtained by contracting each cluster is still minor-free. 
    For each subset $\check{S}$ of $\check{G}$, there is a corresponding set of vertices in $V(G)$ associated with $\check{S}$, which we will naturally denoted as $S$.

    \vspace{-6pt}
    \paragraph{Construction.} 
    Treat $\check{G}$ as an unweighted graph.
    By assumption (2) of Theorem~\ref{thm:large-diam-cover}, there is 
    a forest cover $\check{\mathcal{F}}$ for~$\check{G}$ of size $2^{O(f(\e))}$
    that preserves all paths in $\check{G}$ with hop-length $O(f(\e))$.
    
    For each tree $\check{T}$ in each forest $\check{F}$ in $\check{\mathcal{F}}$ rooted at some supernode $C$, 
    perform the following \emph{transformation}:
    Let $r$ be an arbitrary vertex in~$C$; construct a \EMPH{star $S_r$} centered at $r$ connected to every vertex in $T$, where the weight of the edge\footnote{Note that the star $S_r$ uses Steiner edges.  Unlike in \Cref{thm:bdd-diameter-tree-cover}, the tree cover we construct for \Cref{thm:large-diam-cover} is not a spanning tree cover.} from $r$ to $v$ is set to be $\dist_G(r, v)$. 
    Applying this transformation to every tree in a forest on $V(\check{G})$ produces a forest on $V(G)$: the fact that clusters in the shortcut partition are vertex-disjoint implies that two trees $\check{T_1}$ and $\check{T_2}$ in forest $\check{F}$ are vertex-disjoint if and only if the two corresponding subsets $T_1$ and $T_2$ of $G$ are vertex-disjoint.
    Return $\mathcal{F}$, the set of forests produced.

    \vspace{-6pt}
    \paragraph{Distortion guarantee.} 
    We claim that for every pair $u$ and $v$ in $G$, there is a path in some forest $F \in \mathcal{F}$ such that $\dist_F(u,v) \leq \dist_G(u,v) + O(\e \Delta)$. 
    Let $u$ and $v$ be two vertices in $G$.
    By assumption (1), there is a $f(\e)$-hop path \EMPH{$\check{\pi}$}  in the cluster graph induced by the subset of clusters that has nontrivial intersections with some shortest path $\pi(u,v)$ between $u$ and $v$ in $G$.
    By construction, $\check{\mathcal{F}}$ contains some BFS tree $\check{T}$ in a forest $\check{F}$ that preserves $\check{\pi}$. 
    Let $C$ be the cluster that is the root of $\check{T}$. 
    Then $\mathcal{F}$ contains some star $S_r$ that is rooted at $r \in C$ connecting to every vertex in $T$. 
    As $\check{\pi}$ is preserved by $\check{T}$, $T$ contains both $u$ and $v$ because $\check{\pi}$ starts at the cluster containing $u$ and ends at the one containing $v$. 
    Thus, $\dist_{S_r}(u,v)$ is the length of a shortest path between $u$ and $v$ in $G$ that passes through $r$; i.e.\ $\dist_{S_r}(u,v) = \dist_G(r,u) + \dist_G(r,v)$.
    Further, because $C$ is guaranteed to intersect $\pi(u,v)$ by the property of shortcut partition, there is a short path in $G$ between $u$ and $v$ through $r$---%
    walk from $u$ along $\pi(u,v)$ until reaching some vertex $c \in C$, walk from $c$ to $r$ and then back to $c$ (in $G$ but not necessarily in $C_r$), and then finish traveling along $\pi(u,v)$. 
    As $C$ has diameter $\e \Delta$, this path has length $\dist_G(u,v) + O(\e \Delta)$.
    Thus, the star ${S_r}$ in $\mathcal{F}$ satisfies $\dist_{S_r}(u,v) \le \dist_G(u,v) + O(\e \Delta)$. We remark that each star has radius $\Delta$, so our forest cover is $O(\Delta)$-bounded.
\end{proof}

\subsection{Expanding a forest}
\label{SS:expand-forest}

\expandforest{}
\begin{proof}
    We give an algorithm for constructing a root expansion. 
    Starting with the graph $G$, treat each tree in $F$ as a cluster, and create the cluster graph {$\check{G}$}. 
    Find a vertex-coloring of $\check{G}$ with
    $O_r(1)$ colors (the \emph{chromatic number} of a $K_r$-minor-free graph is upper-bounded by its \emph{degeneracy}, which is asymptotically equivalent to arboricity~\cite{Nash64,Diestel}),
    and let \EMPH{$F_c$} denote the set of trees in $F$ that are colored with color $c$.
    Starting with $G$, contract each tree of $F_c$ into a supernode, and call the resulting graph \EMPH{$\check{G}_c$}. 
    One important property is that 
    the other endpoint of any edge of $\check{G}_c$ incident to a supernode must be an \EMPH{ordinary} vertex not in $F_c$.  (See Figure~\ref{fig:supernodes}.)

\begin{figure}[h!]
    \centering \includegraphics[width=0.7\textwidth]{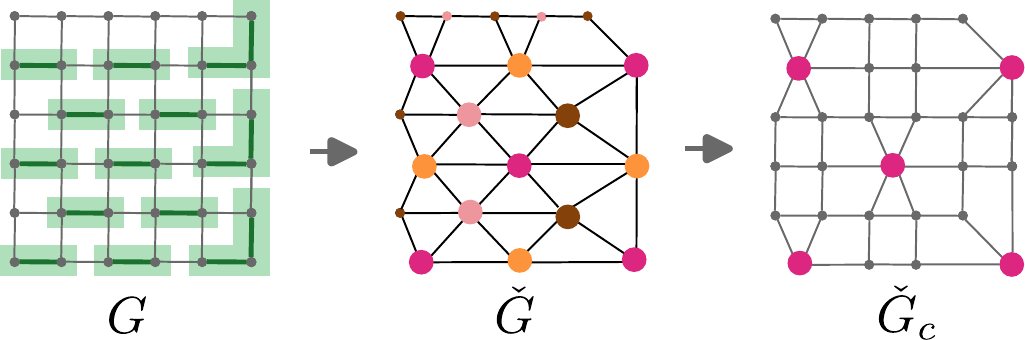}
    \caption{Graph $G$ with forest $F$ colored in green; contracted graph $\check{G}$ with an $O(1)$ vertex-coloring; and contracted graph $\check{G}_c$.}
    \label{fig:supernodes}
\end{figure}
    
    Using the minor-freeness of $\check{G}_c$, find an edge decomposition of $\check{G}_c$ by $O_r(1)$ star forests \EMPH{$\set{\check{F}_{c,1}, \ldots, \check{F}_{c, O(1)}}$}.
    For each star $\check{S}$ in $\check{F}_{c,t}$, define \EMPH{\textsc{Uncompress}$(\check{S})$} as follows:
    \begin{quote}
        If the center of $\check{S}$ is a supernode (corresponding to some BFS tree in $F_c$), let $r \in V(G)$ be the root of that tree. Otherwise, let $r \in V(G)$ be the center of the star $\check{S}$.

        Let \EMPH{$S$} denote the set of vertices in $G$ that belong to $\check{S}$ (both as an ordinary vertex in $G$, or as a vertex in some supernode of $\check{F}_{c,t}$.)
        Return the BFS tree on $G[S]$ rooted at $r$. 
    \end{quote}
    For every star forest $\check{F}_{c,t}$, let \EMPH{$F_{c,t}$} be the union of $\textsc{Uncompress}(\check{S})$ over every star $\check{S}$ in $\check{G}_{c,t}$.  By construction $F_{c,t}$ is a subset of edges in $G$.
    Because the stars in $\check{F}_{c,t}$ are vertex-disjoint in $\check{G}$ and the vertex sets in $G$ corresponding to any two nodes in $\check{G}$ are disjoint,
    $F_{c,t}$ is a disjoint union of BFS trees in $G$, which is a forest.
    Return the set of $O_r(1)$ BFS forests $\mathcal{F} \coloneqq \set{F} \cup \bigcup_{c,t} \set{F_{c,t}}$. 

\begin{figure}[h!]
    \centering \includegraphics[width=0.45\textwidth]{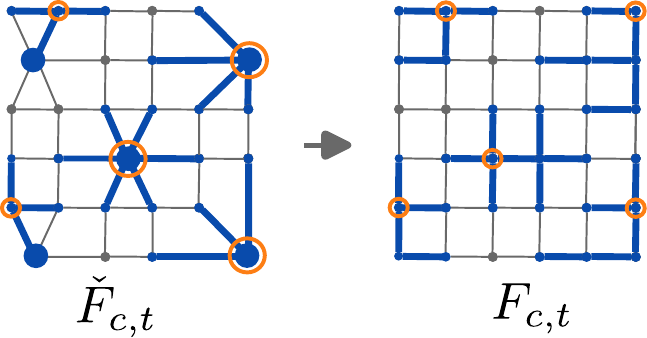}
    \caption{Star forest $\check{F}_{c,t}$ and the uncompressed forest $F_{c,t}$. The root of each tree is circled in orange.}
    \label{fig:uncompress}
\end{figure}

    \paragraph{Preservation guarantee.} 
    Let $\pi$ be a path in $G$ between vertices $u$ and $v$, whose prefix is preserved by $F$. Let $T$ be a tree in $F$ that preserves the prefix of~$\pi$.
    \begin{itemize}
        \item \textbf{Case 1: $v \in V(T)$.} In this case, $T$ preserves the entirety of $\pi$, not just the prefix. As $F \in \mathcal{F}$, the path $\pi$ is preserved by $\mathcal{F}$.
        
        \item \textbf{Case 2: $v \not \in V(T)$.} Let $c$ be the color of the forest $F_c$ where $T$ is in. As $v \not \in V(T)$, there is an edge $e_\pi \in E(\check{G}_c)$ that corresponds to the last edge of $\pi$. 
        The edge $e_\pi$ is covered by a star $\check{S}$ in some $\check{G}_{c,t}$, connecting a vertex in $T$ with $v$.
        There are two subcases. 
        \begin{itemize}
            \item \textbf{Case 2a: $\mathbf{T}$ is the center of $\mathbf{\check{S}}$.} As the prefix of $\pi$ passes through the root $r$ of $T$, clearly $\pi$ passes through $r$. Further, the vertices of $\pi$ are all in $S$, as $e_\pi$ is in $E(\check{S})$. Thus, the BFS tree \textsc{Uncompress}$(\check{S})$ preserves $\pi$.
            
            \item \textbf{Case 2b: $\mathbf{T}$ is not the center of $\mathbf{\check{S}}$.} Here the root of the star $\check{S}$ is $v$, and $T$ is a leaf. 
            Clearly $\pi$ passes through $v$, and the vertices of $\pi$ are all in $S$. Thus, \textsc{Uncompress}$(\check{S})$ preserves~$\pi$. \qed
        \end{itemize}
    \end{itemize}
\end{proof}


\section{A grid-like clustering for planar graphs}
\label{S:grid-scattering}

We show that planar graphs can be partitioned into clusters, such that the clusters interact with each other in a manner similar to a grid graph. Using this structure, we show that every planar graph admits an $(\e, O(1/\e^2))$-shortcut partition.
In this section, we define the grid-like structure and show that it leads naturally to a shortcut partition. In \Cref{S:gridtree-planar}, we show that every planar graph admits a grid-like clustering.  In \Cref{S:treecover-gridlike}, we use this structure to directly construct tree covers for planar graphs. 

\subsection{Definition of gridtree}

Let $G$ be a graph and $H$ be a connected induced subgraph of $G$, and let $w$ be a fixed parameter.
Throughout the rest of the paper, when $G$ is a planar graph, we assume that $G$ has a fixed drawing in the plane; all references to the external face of (a subgraph of) $G$ refer to this fixed drawing.  All subgraphs of $G$ inherit the drawing of $G$.

\begin{definition}
\label{def:column}
    Let $H$ be a connected graph with a disjoint vertex partition into subsets, some of which are \EMPH{columns} and some of which are \EMPH{leftover sets}. 
    A \EMPH{width-$w$ gridtree $\mathcal{T}$} (for short, a \EMPH{$w$-gridtree}) is a tree in which there is a one-to-one correspondence between columns and nodes of $\mathcal{T}$, and between leftover sets and edges of $\mathcal{T}$. 
    The gridtree $\mathcal{T}$ satisfies the following properties:
    \begin{itemize}
        \item 
        \label{def:column-adj}
        \textnormal{[Column adjacency.]}    
        Let $(u,v)$ be an edge of $H$.  
        Either (1) the endpoints $u$ and $v$ belong to the same subset (either a column or a leftover set) in the partition, or (2) $u$ and $v$ belong to columns that are adjacent in $\mathcal{T}$, or (3) $u$ and $v$ belong to a column and a leftover set that are incident in $\mathcal{T}$. 

        \item 
        \label{def:column-width}
        \textnormal{[Column width.]} 
        Let $\eta$ be a column in $\mathcal{T}$. 
        Notice that every column (other than $\eta$) and every edge in $\mathcal{T}$ is either above or below
        $\eta$ in $\mathcal{T}$. 
        If some column or edge containing a vertex $v$ is above (resp.\ below) $\eta$, we say that $v$ is \EMPH{above} (resp.\ \EMPH{below}) $\eta$. 
        If $a$ is a vertex in $\eta$ that is adjacent (in $H$) to a vertex above $\eta$, and $b$ is a vertex below $\eta$, and $P$ is a path in $H$ between $a$ and $b$, then we say $P$ \EMPH{passes through} $\eta$. Every path that passes through $\eta$ has length at least $w$.
        
        \item 
        \label{def:column-shortcut}
        \textnormal{[Column shortcut.]} 
        Let $\eta$ be a column in $\mathcal{T}$, 
        and let \EMPH{$H_\eta$} denote the subgraph induced by all columns below $\eta$, together with all leftover sets below $\eta$ or incident to $\eta$.
        There is a shortest path \EMPH{$\pi_\eta$} in $H_\eta$ such that 
        all vertices in $\eta$ are within distance $2w$ to $\pi_\eta$ with respect to the induced subgraph $H[\eta]$.
    \end{itemize}
\end{definition}

\begin{figure}[h!]
    \centering \includegraphics[width=0.7\textwidth]{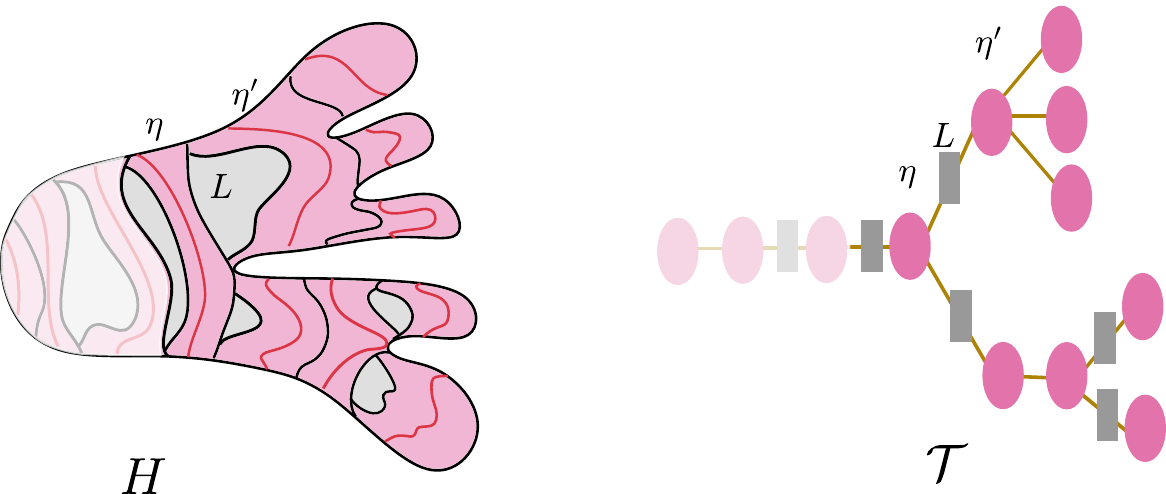}
    \caption{Graph $H$ partitioned into columns and leftover sets, and a gridtree $\mathcal{T}$. Columns $\eta$ and $\eta'$ are adjacent, and the leftover set $L$ lies between them. Column $\eta$ is above $L$ and $\eta$. Each column contains a path $\pi$ (from column shortcut property), marked in red. The subgraph $H_\eta$ is marked in a darker color.}
    \label{fig:gridtree}
\end{figure}

\begin{lemma}
\label{L:gridtree}
    For any $w > 0$, any planar graph $H$ has an $w$-gridtree $\mathcal{T}$ in which every vertex that is at most $w$ away from an external vertex (in a given planar drawing of $H$) belongs to a column of $\mathcal{T}$.
\end{lemma}

\begin{definition}\label{def:gridtree-hierarchy}
    A \EMPH{width-$w$ gridtree hierarchy $\mathcal{H}$} of $G$ is a tree in which each node is a pair $\mu = (H^\mu, \mathcal{T}^\mu)$, where $H^\mu$ is a connected subgraph of $G$, and $\mathcal{T}^\mu$ is a $w$-gridtree for $H^\mu$. The root of $\mathcal{H}$ is associated with the entire graph $G$. 
    The parent node of $\mu$ is denoted by $\parent(\mu)$.
    The hierarchy $\mathcal{H}$ satisfies the following:
    \begin{itemize}
        \item \textnormal{[Layer nesting.]} The children of every node $(H^\mu, \mathcal{T}^\mu)$ are in one-to-one correspondence with the components of subgraphs induced by the leftover sets of $\mathcal{T}^\mu$, together with their~gridtrees.
        
        \item \textnormal{[Layer width.]} For every node $(H^\mu, \mathcal{T}^\mu)$, we say that a vertex $v$ in $H^\mu$ is an \EMPH{outer vertex}
        of $H^\mu$ if $v$ is adjacent to some vertex in a column of $\mathcal{T}^{\parent(\mu)}$. (The root node of $\mathcal{H}$ has no outer vertices.)
        For every outer vertex $v$ in $H^\mu$, every vertex $u$ in $H^\mu$ with $\dist_{H^\mu}(u,v) \leq w$ belongs to some column of $\mathcal{T}^\mu$.
        In other words, every vertex that is at most distance $w$ away from any outer vertex is covered by columns of $\mathcal{T}^\mu$.
    \end{itemize}    
\end{definition}
We remark that, when we construct a gridtree hierarchy for planar graphs in \Cref{S:gridtree-planar}, the outer vertices of $H^\mu$ will be the vertices on the external face of $H^\mu$. 

\begin{figure}[h!]
    \centering \includegraphics[width=0.6\textwidth]{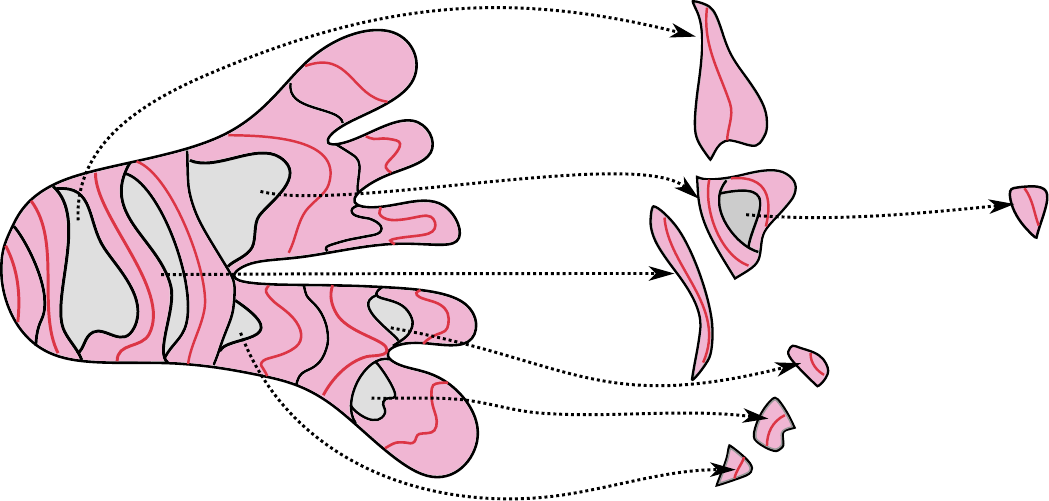}
    \caption{A gridtree hierarchy with depth 3.}
    \label{fig:hierarchy}
\end{figure}

Define \EMPH{level} of a column $\eta$ in gridtree $\mathcal{T}$ to be the distance between $\eta$ and the root of $\mathcal{T}$.
Define \EMPH{layer} of a pair $(H^\mu,\mathcal{T}^\mu)$ in gridtree hierarchy $\mathcal{H}$ to be the distance between $\mu$ and the root of $\mathcal{H}$.
A gridtree hierarchy $\mathcal{H}$ has \EMPH{depth} $d$ if the number of layers in $\mathcal{H}$ is at most $d$. 
We will reiterate that there are two tree structures at play here: 
the \emph{gridtrees $\mathcal{T}^\mu$} associated with each instance $(H^\mu,\mathcal{T}^\mu)$ of the gridtree construction, and the \emph{gridtree hierarchy $\mathcal{H}$} which represents the recursive nature of the construction where every component of the leftover sets of $\mathcal{T}^\mu$ becomes a child of $\mu$.
We will be consistent when it comes to the term \emph{levels} and \emph{layers}, where the former is within a specific grid tree $\mathcal{T}^\mu$, and the latter is for the whole hierarchy $\mathcal{H}$.
Notice that all columns and leftover sets in a gridtree are pairwise vertex-disjoint, and thus so are the subgraphs $H^\mu$ associated with nodes in the same layer of $\mathcal{H}$.

\begin{lemma}
\label{L:gridtree-planar}
    Let $G$ be a planar graph with diameter $\Delta$, and let $\e$ be a parameter in $(0,1)$. Then $G$ has an $\e\Delta$-gridtree hierarchy~$\mathcal{H}$.
\end{lemma}

We will prove Lemma~\ref{L:gridtree} and Lemma~\ref{L:gridtree-planar} in Section~\ref{S:gridtree-planar}.

\subsection{Gridtree hierarchy gives shortcut partition}
\label{S:gridtree-scatterting}

\noindent For the rest of the section, we prove that every
gridtree hierarchy for a graph $G$ gives rise to a shortcut partition for $G$. 
In particular, the hierarchy constructed by Lemma~\ref{L:gridtree-planar} 
produces a shortcut partition.
This suggests that planar graphs admit a stronger form of shortcut partition; using this partition, we will prove (cf. Section~\ref{S:treecover-gridlike}) that planar graphs have $O_\e(1)$-size tree cover where the constant depends polynomially on $1/\e$.%

Recall that a \EMPH{clustering} of a planar graph $G$ is a partition of the vertices of $G$ into \EMPH{clusters} $\mathcal{C} \coloneqq \set{C_1, \ldots, C_m}$.
Let \EMPH{$\check{G}$} be the \EMPH{cluster graph} of $G$ with each cluster in $\mathcal{C}$ contracted to a \EMPH{supernode}.
Denote \EMPH{$\len{P}$} to be the length of the path $P$; in notation, if $P$ starts at vertex $s$ and ends at $t$, then $\len{P} \coloneqq \dist_P(s,t)$.

\begin{definition}
    The \EMPH{cost} of a path $P$ in $G$ with respect to clustering $\mathcal{C}$, denoted as \EMPH{$\cost_{\mathcal{C}}(P)$}, is equal to the minimum hop-length over all paths $\check{P}$ in $\check{G}$ where (1) the endpoints of $\check{P}$ are the clusters containing $u$ and $v$, and (2) $\check{P}$ only touches supernodes that correspond to clusters that $P$ passes through.
    For any $t$ between $0$ and $1$ and for any vertex pair $(u,v)$, we define the \EMPH{cost with $(1+t)$ distortion} with respect to clustering $\mathcal{C}$, denoted \EMPH{$\cost_{t, \mathcal{C}}{(u,v)}$}, to be the minimum $\cost_{\mathcal{C}}(P)$ across every approximate shortest path $P$ between $u$ and $v$ whose length is at most $(1 + t) \cdot \delta_G(u,v)$.
\end{definition}

When $\mathcal{C}$ is clear from context, we omit it from the subscript and simply write $\cost(P)$ and $\cost_t(u,v)$.  We now introduce a slight generalization of shortcut partitions.

\begin{definition}
    Let $G$ be a graph with diameter $\Delta$. An \EMPH{$(\e, h)$-shortcut partition with $(1+t)$ distortion} for $G$ is a clustering $\mathcal{C} = \set{C_1, \ldots, C_m}$ of $G$ such that:
    \begin{itemize}
    \item \textnormal{[Diameter.]} The strong diameter of each cluster $C_i$ is at most $\e \Delta$;
    \item \textnormal{[Low-hop.]} For any vertices $u$ and $v$ in $G$, we have $\cost_{t, \mathcal{C}}{(u,v)} \le h$.
    \end{itemize}
\end{definition}

Notice that an $(\e, h)$-shortcut partition with $(1+\e)$ distortion (as defined above) is an $(\e, h)$-shortcut partition as defined in Section~\ref{SS:prelim}.
Given a planar graph $G$ with diameter $\Delta$ along with a $(t\e\Delta)$-gridtree hierarchy (where $t$ and $\e$ are between $0$ and $1$), we will find an $(O(\e), O \left(\frac{1}{t\e} \right))$-shortcut partition with $(1 + O(t))$ distortion for $G$.
In fact, the partition constructed will satisfy an extra property:
\begin{itemize}
    \item \textnormal{[Cluster ordering.]}
    For every node $(H, \mathcal{T})$ in the hierarchy $\mathcal{H}$ and for every column $\eta$ in $\mathcal{T}$, there is an ordering on the \EMPH{cluster centers} $c_1, \ldots, c_m$ in $\eta$, one from each cluster containing vertices of $\eta$, such that for every pair $i, j \in \set{1, \ldots, m}$, we have $\dist_{H_\eta}(c_i, c_j) \geq \abs{i - j}\cdot \e\Delta$.
\end{itemize}

\subsubsection{Clustering a column}
\label{SS:cluster-column}

\begin{figure}[ht!]
\centering\small
\begin{algorithm}
\textul{$\textsc{ClusterColumn}(\eta)$:} \+
\\  \Comment{select a set of cluster centers}
\\  $\pi' \gets \pi_\eta$
\\  $i \gets 1$
\\  while $\pi'$ is nonempty: \+
\\      $c_i \gets$ first vertex on $\pi'$
\\      $\pi' \gets$ longest suffix of $\pi'$ that is $\ge \e\Delta$ shorter than $\pi'$
\\      initialize cluster $C_i \gets \set{c_i}$ 
\\      increment $i$ \-
\\
\\  \Comment{assign vertices in $\pi_\eta$ using closest center}
\\  for each $v$ in $V(\pi_\eta)$: \+
\\      $i^* \gets \argmin_{i \in [m]} \dist_\eta(c_i, v)$, breaking ties by choosing the smallest $i$ 
\\      assign $v$ to cluster $C_{i^*}$ \-
\\
\\  \Comment{assign vertices in $\eta$ using closest vertex in $\pi_\eta$}
\\  for each $v$ in $\eta$: \+
\\      $v^* \gets \argmin_{v \in V(\pi_\eta)} \dist_{H[\eta]}(v^*, v)$, breaking ties using some fixed ordering of $V(\pi_\eta)$
\\      assign $v$ to the cluster containing $v^*$ \-
\\  return the clusters $\set{C_1, \ldots, C_m}$
\end{algorithm}
\end{figure}

First we describe how to create clusters for a single column in any $(t\e\Delta)$-gridtree in the hierarchy $\mathcal{H}$.
To simplify the presentation, we assume $t \in (0, 1/8]$; indeed, if $t \in (1/8, 1)$ then we can scale down $t$ by a factor of $8$ without affecting our results by more than a constant factor.
Let $(H, \mathcal{T})$ be a node in $\mathcal{H}$.
Let $\eta$ be a column in $\mathcal{T}$ and let $\pi_\eta$ be the shortest path in $H_\eta$ guaranteed by the column shortcut property (Definition~\ref{def:column-shortcut}) for gridtree. 
In particular, every vertex in $\eta$ is within distance $2t \e \Delta$ to $\pi$. 
Then we can cluster $\eta$ according to the procedure \EMPH{$\textsc{ClusterColumn}(\eta)$}. It is immediate from the construction of the clustering that the cluster centers satisfy the cluster ordering property.
Notice that while the column $\eta$ has width $t \e \Delta$, the diameter of each cluster constructed is $O(\e\Delta)$, not $O(t \e\Delta)$.

\begin{lemma}
\label{lem:cluster-path}
    The algorithm $\textsc{ClusterColumn}(\eta)$ returns a clustering of $\eta$ such that (i) each cluster has strong diameter at most $O(\e \Delta)$, and (ii) for any pair of points $u$ and $v$ in $\eta$, we have \(
        \cost_{8t}(u,v) \leq \frac{\dist_{H_\eta}(u,v)}{\e \Delta} + 4
    \).
\end{lemma}

\begin{proof}
    \textbf{(i)} By assumption, every vertex $v$ in $\eta$ is within distance $2t \e \Delta$ of some vertex $p_v$ on $\pi_\eta$ such that $v$ is assigned to the same cluster as $p_v$.
    By choice of cluster centers $\set{c_i}$, vertex $p_v$ is within distance $\e \Delta$ (not $t \e \Delta$) of some $c_i$ and is assigned to same cluster as $c_i$.
    Thus, every cluster has radius at most $(1+2t)\e\Delta < 2\e \Delta$ and diameter at most $4 \e \Delta$. 
    Every point on a shortest path between $v$ and $p_v$ belongs to the same cluster, because the construction breaks ties by a fixed ordering of $V(\pi_\eta)$. Similarly, every point on a shortest path between $p_v$ and $c_i$ belongs to the same cluster.
    We conclude that every clusters is connected, and the $4 \e \Delta$ bound applies to strong diameter.

    \medskip \noindent \textbf{(ii)}
    Let $P$ be a shortest path in $H_\eta$  with endpoints $u$ and $v$ in $\eta$. Let $C_u$ and $C_v$ denote the clusters containing $u$ and $v$, respectively. If  $C_u = C_v$ or if $C_u$ and $C_v$ are adjacent in the ordering, then $\cost_{8t}(u,v) \leq 1$ and we are done.
    Otherwise, we note that $P$ may walk outside of the column $\eta$ into other columns or leftover sets. However,
    by the column shortcut property (\Cref{def:column-shortcut}),
    we can find an approximate shortest path \EMPH{$P'$}
    that is contained within $\eta$:
    From vertex $u$, walk (along a shortest path in $C_u$) to the point $p_u$ in $\pi_\eta \cap C_u$ that is closest to $u$,
    then walk along $\pi_\eta$ to the point $p_v$ in $\pi_\eta \cap C_v$ that is closest to $v$,
    and then walk (along a shortest path in $C_v$) to $v$.
    Denote the subpath of $P'$ on $\pi_\eta$ from $p_u$ to $p_v$ as \EMPH{$P'[\pi_\eta]$}; see Figure~\ref{fig:cluster_col}.

    \begin{figure}[h]
    \centering
    \includegraphics[width=0.4\textwidth]{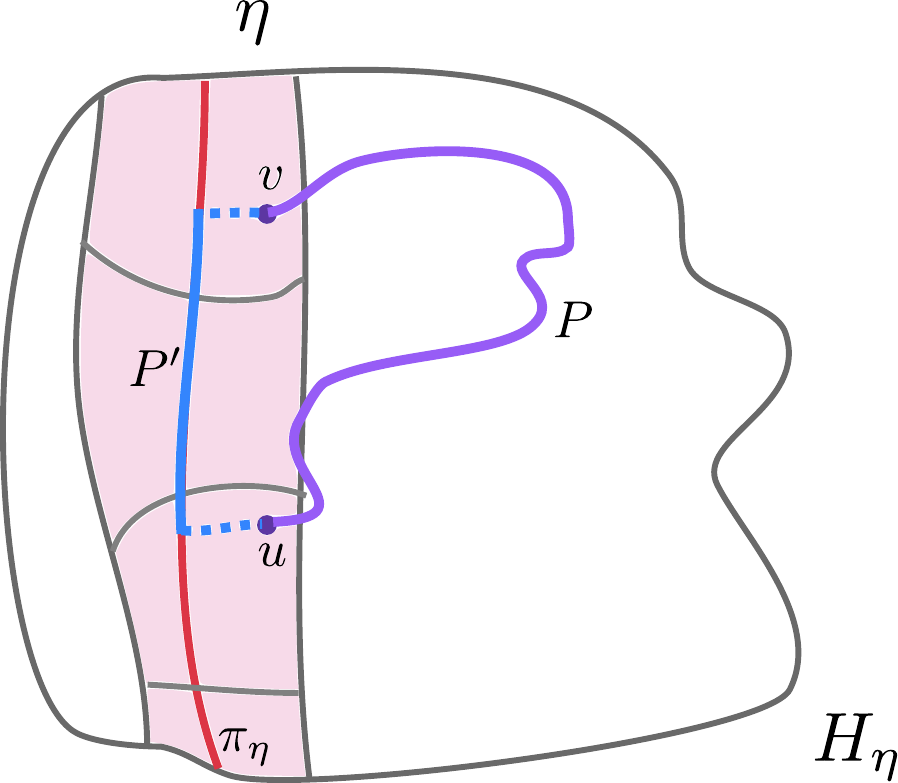}
    \caption{A partition of the column $\eta$ into clusters, the graph $H_\eta$, and vertices $u$ and $v$ in $\eta$. The path $P$ (in purple) is a shortest path between $u$ and $v$ in $H_\eta$. The path $P'$ (in blue) described in the proof of \Cref{lem:cluster-path}; the subpath $P'[\pi_\eta]$ is solid, and the subpaths $P' \setminus P'[\pi_\eta]$ are dotted. }
    \label{fig:cluster_col}
    \end{figure}

    \medskip \noindent \textbf{Cost of $P'$.} First observe that, because of the way vertices not in $\pi_\eta$ were assigned, we can walk from $u$ to $p_u$ (and from $p_v$ to $v$) while remaining within a single cluster. This implies that $\cost(P') = \cost(P'[\pi_\eta])$.
    Because of the cluster ordering, path $P'[\pi_\eta]$ fully walks past every cluster that it intersects (other than $C_u$ and $C_v$); that is, $\cost(P'[\pi_\eta]) \leq \frac{\len{P'[\pi_\eta]}}{\e \Delta} + 2$. Now observe that, because $P'[\pi_\eta]$ is a shortest path in $H_\eta$, we have $\len{P'[\pi_\eta]} \leq \len{P} + 2t \e \Delta$. We conclude that $\cost(P') \leq \frac{\len{P}}{\e \Delta} + 4$.

    \medskip \noindent \textbf{Length of $P'$.}
    By assumption $u$ and $v$ are not in adjacent clusters, so $p_u$ and $p_v$ are also not in adjacent clusters.  
    This implies that $\len{P'[\pi_\eta]} \geq \e\Delta$.
    As every vertex is within distance $2t\e\Delta$ of $\pi$, we have $\len{P} \ge \len{P'[\pi_\eta]} - 4t \e\Delta$.
    Using the fact that $t \le 1/8$, we have 
    $\len{P} \ge \e\Delta/2$,
    and we conclude that $\len{P'} \leq \len{P'[\pi_\eta]} + 2 t\e\Delta  \leq \len{P} + 4 t\e \Delta \leq (1+8t)\cdot \len{P}$.

    \medskip \noindent Combining the two claims, we get that $\cost_{8t}(u,v) \leq \cost(P') \leq \frac{\len{P}}{\e\Delta} + 4$.
\end{proof}

\subsubsection{Clustering a gridtree hierarchy}
\label{SS:cluster}

\begin{figure}[ht!]
\centering\small
\begin{algorithm}
\textul{$\mathsc{ClusterHierarchy}(G, \mathcal{H})$:} \+
\\  \Comment{Cluster each column $\eta$}
\\  $\mathcal{C} \gets \varnothing$
\\  for each node $(H^\mu, \mathcal{T}^\mu)$ in $\mathcal{H}$: \+
\\      for each column $\eta$ of $\mathcal{T}^\mu$: \+
\\          $\mathcal{C} \gets \mathcal{C} \cup \textsc{ClusterColumn}(\eta)$ \-\-
\\  return $\mathcal{C}$
\end{algorithm}
\end{figure}

Given a planar graph $G$ (with fixed parameters $\e$, $t$, and $\Delta$) and a $(t\e\Delta)$-gridtree hierarchy $\mathcal{H}$ for~$G$, we use the recursive algorithm \EMPH{$\textsc{ClusterHierarchy}(G,\mathcal{H})$} to create a shortcut partition of $G$.
Let \EMPH{$\mathcal{C}$} be the resulting clustering.
For every node $(H^\mu,\mathcal{T}^\mu)$ in $\mathcal{H}$,
we assign clusters to each column $\eta$ of $\mathcal{T}^\mu$ by calling $\textsc{ClusterColumn}(\eta)$.
Let \EMPH{$\mathcal{C}^{\mu}$} denote the set of clusters associated with columns of $\mathcal{T}^\mu$.  
We have $\EMPH{$\mathcal{C}$} = \bigcup_{\mu} \mathcal{C}^{\mu}$.
We want to show that $\mathcal{C}$ is indeed a shortcut partition by proving each of the properties. 
The diameter and cluster ordering properties follow directly from the correctness of $\textsc{ClusterColumn}$.

To prove the low-hop property that $\cost_{8t}{(u,v)} \le f(\e)$ for any $u$ and $v$, 
we will look at an arbitrary shortest path $P$ between $u$ and $v$: 
Find the highest node $(H^\mu, \mathcal{T}^\mu)$ in the hierarchy $\mathcal{H}$ such that $H^\mu$ contains a vertex in $P$.
Chop up $P$ into parts, alternating between parts that are covered by the clusters of the columns in $\mathcal{T}^\mu$ and parts that are in the leftover sets.
Inductively the parts in the leftover sets have low costs; we just have to prove that the parts covered by $\mathcal{T}^\mu$ passes through at most $O \left(\frac{1}{t\e} \right)$ many columns in the worse case, using the column width property.
We define the notation \EMPH{$P[u,v]$} to be the subpath of $P$ that starts from vertex $u$ and ends at vertex $v$ for any path $P$.

\begin{lemma}\label{lem:planar-cost-inductive}
    Let $(H, \mathcal{T})$ be node in the hierarchy $\mathcal{H}$ that is not the root. 
    Let $u$ and $v$ be two vertices in $H$ such that $u$ is an outer vertex, and $H$ contains a shortest path $P$ (with respect to $G$) between $u$ and $v$.
    Then $\cost_{8t}(u,v) \leq 85\cdot \frac{\len P}{t \e \Delta} + 16$.
\end{lemma}

\begin{proof}
    We proceed by induction on layers of nodes of $\mathcal{H}$. 
    For the base case, we prove the claim when $(H, \mathcal{T})$ is a leaf in the deepest layer. In the inductive case, we assume the claim is true for all nodes in deeper layers. We begin with the inductive case. 

    Let $\eta$ be the lowest column in $\mathcal{T}^\mu$ such that $H_\eta$ contains $P$. 
    Notice that $\eta$ contains some vertex $p$ of~$P$ (as otherwise, column adjacency property would imply that there is some child $\eta'$ of $\eta$ where $H_{\eta'}$ contains $P$). 
    We split $P$ into two paths: \EMPH{$P^{(u)}$}, the subpath starting at $p$ and ending at $u$; and \EMPH{$P^{(v)}$}, the subpath starting at $p$ and ending at $v$. (Notice that both $P^{(u)}$ and $P^{(v)}$ contain $p$ and are therefore not vertex-disjoint; this is the only time in this section where we split a path into subpaths that are not vertex-disjoint.)
    
    We chop $P^{(u)}$ into $\ell_u$ vertex-disjoint subpaths
    $P^{(u)} = P_1 \circ Q_1 \circ \ldots \circ P_{\ell_u} \circ Q_{\ell_u}$, where each $P_i$ (except for $P_1$
    and $P_{\ell_u}$) \emph{passes through} a column (as defined in the column width property), and each $Q_i$ (possibly empty) is contained within a single leftover set.
    We do this as follows:
    Initialize $P^* \gets P^{(u)}$ and $i \gets 1$.
    Define $\eta_i$ to be the column containing the first vertex of $P^*$.
    Define $P_i$ to be the maximal prefix of $P^*$ that ends at some vertex in $\eta_i$.
    Remove prefix $P_i$ from $P^*$, and define $Q_i$ to be the maximal prefix of $P^*$ that touches no column. (Notice that $Q_i$ may be empty.)
    Remove prefix $Q_i$ from $P^*$, and increment $i$. (Notice that by maximality of $Q_i$, the first vertex on $P^*$ is in some column.)
    The process terminates when $P^*$ is empty.
    Using an identical process, we chop $P^{(v)}$ into $\ell_v$ vertex-disjoint subpaths $P^{(v)} = P_{\ell_u + 1} \circ Q_{\ell_u + 1} \circ \ldots \circ P_{\ell_u + \ell_v} \circ Q_{\ell_u + \ell_v}$; see Figure~\ref{fig:chop}.

\begin{figure}[h!]
    \centering \includegraphics[width=0.5\textwidth]{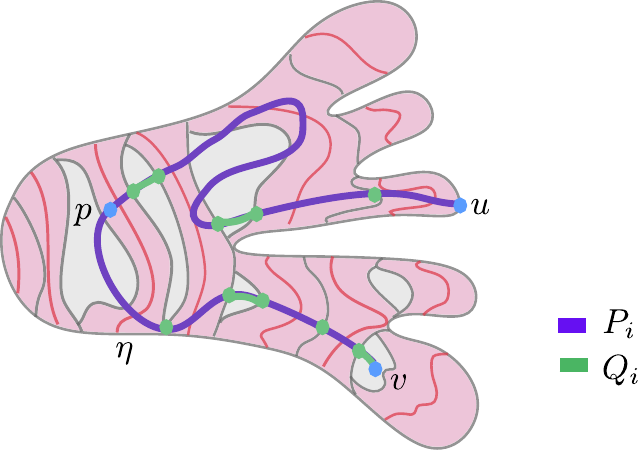}
    \caption{The path $P$ from $u$ to $v$, chopped into subpaths $P_i$ and $Q_i$. The column $\eta$ is the lowest node in the gridtree such that $H_\eta$ contains $P$. Vertex $p$ is in $\eta \cap P$.}
    \label{fig:chop}
\end{figure}

    We now argue that there are few paths, i.e.\ $\ell_u + \ell_v = O \Paren{\frac{\len{P}}{t \e\Delta}}$.
    For every $i$, let $(p_i, p_i')$ be the endpoints of $P_i$, and let $(q_i, q_i')$ be the endpoints of $Q_i$ if $Q_i$ is nonempty.
    For every $P_i$ that is a subpath of $P^{(u)}$, let \EMPH{$P_i^+$} denote the path $P_i$ concatenated with the next vertex in $P^{(u)}$.
    That is, $P_i^+ \coloneqq P^{(u)}[p_i, q_i]$ if $q_i$ exists, and $P_i^+ \coloneqq P^{(u)}[p_i, p_{i+1}]$ if $q_i$ does not exist and $i \neq \ell_u$, and $P_i^+ \coloneqq P_i$ if $q_i$ does not exist and $i = \ell_u$. 
    We similarly define $P_i^+$ for subpaths $P_i$ of $P^{(v)}$, and define \EMPH{$Q_i^+$} for subpaths $Q_i$ of $P^{(u)}$ and $P^{(v)}$.
    Notice that $\len{P} = \sum_i \Paren{\len{P_i^+} + \len{Q_i^+}}$, and that $\len{P} \ge \sum_i \Paren{\len{P_i^+} + \len{Q_i}}$.
    
    By choice of $\eta$ and $p$, the path $P_1^+$ is contained within $H_\eta = H_{\eta_1}$ (and a similar claim holds for $P_{\ell_u + 1}^+$).
    Column adjacency property implies that  path $P_1^+$ ends at a vertex below $\eta_1$ in $\mathcal{T}$ (unless $\ell_{u}=1$). 
    An easy inductive argument generalizes this statement: For all $i$, path $P_i^+$ is contained in $H_{\eta_i}$; if $i \not \in \set{\ell_u, \ell_u + \ell_v}$, then the path $P_i^+$ ends at a vertex below $\eta_i$ in $\mathcal{T}$; and if $i \not \in \set{1, \ell_u + 1}$, then the subpath $P_i^+$ begins at a vertex above $\eta$. This implies that, if $i \not \in \set{1, \ell_u, \ell_u + 1, \ell_u + \ell_v}$, the subpath $P_i^+$ \emph{passes through} column $\eta_i$.
    By the column width property of the $(t \e \Delta)$-gridtree $\mathcal{T}$, every such $P_i^+$ has length at least $t \e \Delta$.
    We are now ready to bound $\cost_{8t}(u,v)$. We can write $\cost_{8t}(s,t)$ in terms of the costs of (the endpoints of) the subpaths $P_i$ and $Q_i$:
    \[
    \cost_{8t}(u,v) \leq \sum_i \cost_{8t}(p_i,p_i') + \sum_{i : Q_i \neq \varnothing} \cost_{8t}(q_i,q_i')
    \]
    We bound the cost of $P_i$s and $Q_i$s separately.

    \begin{itemize}
        \item Path $P_i$ is a shortest path in $H_\eta$,
        so \Cref{lem:cluster-path} implies that
        \(
            \cost(p_i, p_i') \leq \frac{\len{\smash{P_i}}}{\e \Delta} + 4.
        \)
        As $t \le 1$ and $\len{P_i} \le \len{P_i^+}$, we have
        \(
            \cost(p_i, p_i') \leq \frac{\len{\smash{P_i^+}}}{t \e \Delta} + 4.
        \)
        If 
        \(
        i \not \in \set{1, \ell_u, \ell_u+1, \ell_v + \ell_v}, 
        \)
        we have $\len{P_i^+} \geq t \e \Delta$ and so 
        \(
        \cost(p_i, p_i') \leq 5 \cdot \frac{\len{P_i^+}}{t \e \Delta}.
        \)
        Thus, we have
        \begin{equation}
        \label{eq:cost-p}
        \sum_i \cost_{8t}(p_i, p_i') \leq 16 + 5 \cdot \sum_i \frac{\len{P_i^+}}{t\e\Delta}.    
        \end{equation}

        \item Path $Q_i$ lies in some component of a leftover set, which corresponds to a child $\mu$ in the hierarchy~$\mathcal{H}$. The starting vertex $q_i$ is an outer vertex of $H^\mu$ (as it is adjacent to a vertex in a column of $\mathcal{T}$), so induction hypothesis implies
        \(
        \cost_{8t}(q_i, q_i') \le 85\cdot\frac{\len {Q_i}}{t \e \Delta} + 16. 
        \)
        If $i \not \in \set{1, \ell_u, \ell_u+1, \ell_u + \ell_v}$, subpath $Q_i$ is accompanied by an occurrence of $P_i^+$ where $\len{P_i} \geq t \e \Delta$. For these $i$, we have 
        \(
        \cost_{8t}(q_i, q_i')
        \leq 85 \cdot \frac{\len{\smash{Q_i}}}{t \e \Delta} + 16 \cdot \frac{\len{\smash{P_i^+}}}{t \e \Delta}.
        \)
        Thus, we have

        \begin{equation}
        \label{eq:cost-q}
        \sum_{i : Q_i \neq \varnothing} \cost_{8t}(q_i, q_i')
        \leq 64 + \sum_{i : Q_i \neq \varnothing} \Paren{85 \cdot \frac{\len{\smash{Q_i}}}{t \e \Delta} + 16 \cdot \frac{\len{\smash{P_i^+}}}{t\e \Delta}}.    
        \end{equation}

    \end{itemize}

    \noindent We now combine the costs bounds for the $P_i$s and $Q_i$s. We consider two cases.
    \begin{itemize}
        \item If there are no nonempty $Q_i$, then we have
        \[
        \cost_{8t}(u,v) = \sum_i \cost_{8t}(p_i, p_i) \leq 6\cdot \frac{\len{P}}{t \e \Delta} + 16
        \]
        and the claim is satisfied.
        \item If there is some nonempty $Q_i$, then the layer width property implies that $\sum_i \len{P_i^+} \geq t \e \Delta$, as every vertex within $t \e \Delta$ distance of $u$ is assigned to some column. We have
        \begin{align*}
        \cost_{8t}(u,v) 
        &= 
        \sum_i \cost_{8t}(p_i, p_i') + \sum_{i : Q_i \neq \varnothing} \cost_{8t}(q_i, q_i') \\
        &\le 21 \cdot \sum_i \frac{\len {P_i^+}}{t \e \Delta} + 85 \cdot \sum_i \frac{\len {Q_i}}{t \e \Delta} + 80 &\text{by \Cref{eq:cost-p,eq:cost-q}}\\
        &\le 21 \cdot \frac{\sum_i \len {P_i^+}}{t \e \Delta} + 85 \cdot \frac{\len{P} - \sum_i \len{P_i^+}}{t \e \Delta} + 80\\
        &\le 85 \cdot \frac{\len{P}}{t \e\Delta} + 16. 
        \end{align*}
        The last inequality holds because $\frac{\sum_i \len{P_i^+}}{t\e\Delta} \ge 1$. 
    \end{itemize}
    
\noindent This completes the inductive case. The base case, when $(H, \mathcal{T})$ is a leaf at the deepest layer of $\mathcal{H}$, is identical except that all $Q_i$ are empty and so there is no need to appeal to the inductive hypothesis.
\end{proof}

\begin{lemma}
\label{lem:planar-cost}
    Let $u$ and $v$ be two vertices in $G$ (where $u$ is \emph{not} necessarily an outer vertex), and let $P$ be a shortest path (with respect to $G$) between $u$ and $v$. Then $\cost_{8t}(u,v) \leq 85 \cdot \frac{\len{P}}{t \e \Delta} + 80$.
\end{lemma}
\begin{proof}
    Let $(H, \mathcal{T})$ be the lowest node in the hierarchy $\mathcal{H}$ such that $H$ fully contains $P$. There is some vertex along $P$ in a column of $\mathcal{T}$. This means that, following \Cref{lem:planar-cost-inductive}, we can chop $P$ into subpaths $P_1 \circ Q_1 \circ \ldots \circ P_{\ell_1 + \ell_2} \circ Q_{\ell_1 + \ell_2}$. Letting $(p_i, p_i')$ and $(q_i, q_i')$ denote the endpoints of $P_i$ and $Q_i$, respectively, we have
    \[
    \cost_{8t}(u,v) \le \sum_i \cost_{8t}(p_i, p_i') + \sum_{i:Q_i \neq \varnothing} \cost_{8t}(q_i, q_i').
    \]
    To bound $\cost_{8t}(p_i, p_i')$, we may use identical analysis to that in the proof of \Cref{lem:planar-cost-inductive} to show that \Cref{eq:cost-p} still holds. To bound $\cost_{8t}(q_i, q_i')$, we follow a similar proof to that of \Cref{eq:cost-q}. There is one minor change. To prove that $\cost_{8t}(Q_i) \le 85 \cdot \frac{\len{Q_i}}{t\e\Delta} + 16$, the earlier proof of \Cref{eq:cost-q} appealed to an inductive hypothesis. We no longer have this inductive hypothesis; instead, we can directly apply the statement of \Cref{lem:planar-cost-inductive}. The rest of the proof follows verbatim, and so \Cref{eq:cost-q} holds in this setting. Combining these bounds, we conclude:
    \[
    \cost_{8t}(u,v)
    \le 21 \cdot \sum_i \frac{\len {P_i^+}}{t \e \Delta} + 85 \cdot \sum_i \frac{\len {Q_i}}{t \e \Delta} + 80
    \le 85 \cdot \frac{\len{P}}{t \e\Delta} + 80.
    \]
        
\aftermath
\end{proof}

\noindent \Cref{lem:cluster-path,lem:planar-cost} imply the following theorem.

\begin{theorem}
    \label{thm:gridtree-implies-shortcut}
    Let $\e$ be a number in $(0,1)$, let $t$ be a number in $(0,1/8]$, and let $G$ be a graph with diameter $\Delta$ that has a $(t \e \Delta)$-gridtree hierarchy~$\mathcal{H}$. 
    Then the set $\mathcal{C}$ of clusters produced by $\textsc{ClusterHierarchy}(G)$ 
    satisfies the following: (i) each cluster in $\mathcal{C}$ has strong diameter $O(\e \Delta)$, and (ii) for any $u,v \in V(G)$, we have $\cost_{8t, \mathcal{C}}(u,v) = O \left(\frac{1}{t \e} \right) \cdot \frac{\dist_G(u,v)}{\Delta} + O(1)$.
\end{theorem}

\noindent In particular, planar graphs
have $(t\e\Delta)$-gridtree hierarchies for any $t, \e \in (0,1)$ (by \Cref{L:gridtree-planar}), and so they have $(O(\e), O \left(\frac{1}{t\e} \right))$-shortcut partitions with $(1 + O(t))$ distortion. By choosing $t = O(\e)$ or choosing $t = O(1)$, we conclude:

\begin{corollary}
Any planar graph has an $(\e, O(\e^{-2}))$-shortcut partition with $(1+O(\e))$ distortion, satisfying the diameter, low-hop, and cluster ordering properties.
\end{corollary}

\begin{corollary}
\label{cor:shortcut-constant-distortion}
Any planar graph has an $(\e, O(\e^{-1}))$-shortcut partition with $t$ distortion for any constant $t > 1$, satisfying the diameter, low-hop, and cluster ordering properties.
\end{corollary}

\section{Constructing gridtree hierarchy for planar graphs}
\label{S:gridtree-planar}
 
In this section we prove the existence of an $\e\Delta$-gridtree hierarchy $\mathcal{H}$ for any planar graph $G$ (Lemma~\ref{L:gridtree-planar}).
Throughout this section, we assume that $G$ has a fixed drawing in the plane; all references to the external face of (a subgraph of) $G$ refer to this fixed drawing.  All subgraphs of $G$ inherit the drawing of $G$.

To construct a gridtree hierarchy for $G$, we proceed in rounds. 
In each round we will construct a gridtree $\mathcal{T}$ for a subgraph $H$ of $G$ using the algorithm $\textsc{Gridtree}(H)$, iteratively partitioning the vertices in $H$ into columns and leftover sets, such that every vertex within distance $\e \Delta$ from the exterior face of $H$ is assigned to a column.
We then recursively construct a gridtree hierarchy on each leftover component.
Within a single round, we adapt an algorithm by Busch, LaFortune, and Tirthapura for constructing a sparse cover of a planar graph~\cite{BLT14}. Their algorithm crucially relies on recursively selecting a set of ``far apart'' paths. We follow the same algorithm for selecting paths, and define the columns to be $O(\e \Delta)$-neighborhoods around each path.

\subsection{Constructing a single gridtree}

We begin by defining the algorithm \EMPH{$\textsc{SelectPaths}_H(H', \pi)$}, which appeared implicitly in~\cite{BLT14}. It takes as input a subgraph $H'$ of $H$ (which itself is a subgraph of $G$) and a path $\pi$ that is a shortest path on $H'$ between external vertices of $H$. (We allow $\pi$ to be a path of length $0$ --- that is, $\pi$ may be a single external vertex of $H$.)
\begin{figure}[ht!]
\centering\small
\begin{algorithm}
\textul{$\mathsc{SelectPaths}_H(H', \pi):$} \+
\\  \Comment{Cut away $\e \Delta$-neighborhood of $\pi$}
\\  $N \gets $ $\e\Delta$-neighborhood of $\pi$ in $H'$
\\  $H'_1, \ldots, H'_m \gets$ components of $H' \setminus N$ containing at least an external vertex of $H$
\\
\\  \Comment{Recurse}
\\  for each $H'_i \gets H'_1, \ldots, H'_m$: \+
\\      $Y_i \gets$ set of external vertices of $H'_i$ connected to $N$ by some external edge
\\      \Comment{by~\cite{BLT14}, $1 \leq |Y_i| \leq 2$}
\\      if $|Y_i| = 1$: \+
\\          $\pi_i \gets Y_i$ \-
\\      else: \+
\\          $\pi_i \gets $ shortest path in $H'_i$ connecting both vertices in $Y_i$ \- 
\\      $\textsc{SelectPaths}_H(H'_i, \pi_i)$ \-
\end{algorithm}
\end{figure}

We initialize the process by
selecting an arbitrary external vertex $\pi_0$ of~$H$, and call $\textsc{SelectPaths}_H(H, \pi_0)$. 
There is a tree \EMPH{$\mathcal{T}^*$} naturally associated with the recursion of $\textsc{SelectPaths}_H(H, \pi_0)$. 
Each node $\eta$ of $\mathcal{T}^*$ contains the following:
\begin{itemize}
    \item An induced subgraph $\EMPH{$H^*_\eta$} \subseteq H$, containing at least one external vertex of $G$. We say that a vertex $v$ \EMPH{appears} in the subgraph of $\eta$ if $v \in V(H^*_\eta)$.
    \item A path \EMPH{$\pi_\eta$} which is a shortest path on $H^*_\eta$.
    \item A neighborhood \EMPH{$N_\eta$} defined to be the set of vertices $v \in V(H^*_\eta)$ where $\dist_{H^*_\eta}(v, \pi_\eta) \leq \e \Delta$
    \item A ``leftover'' subgraph \EMPH{$L_\eta$}, defined to be the union of all connected components of $H^*_\eta \setminus N_\eta$ that do \emph{not} contain any external vertices of $H$.
\end{itemize}

The children of $\eta$ are associated with components of $V(H^*_\eta) \setminus (V(L_\eta) \cup N_\eta)$. 
This recursion tree was discussed in~\cite{BLT14}, where the authors implicitly prove the following lemma (cf. \cite[Lemma 9]{BLT14}):

\begin{lemma}
\label{lem:safe-recursion}
    Let $\eta$ be a node of the recursion tree $\mathcal{T}^*$. 
    Let $u$ be a vertex in $H^*_\eta$ that is adjacent in $H$ to some vertex in $V(H^*_{\parent(\eta)})$, where $\parent(\eta)$ is the parent of $\eta$ in $\mathcal{T}^*$. 
    Let $v \in V(H^*_\eta) \setminus (V(L_\eta) \cup N_\eta)$ be a vertex that appears in the subgraph of some child of $\eta$.
    Then every path $P$ in $H^*_\eta$ between $u$ and $v$ intersects $\pi_\eta$.
\end{lemma}

We can transform the recursion tree $\mathcal{T}^*$ into a gridtree \EMPH{$\mathcal{T}$} for $H$. The columns (resp.\ leftover sets) of $\mathcal{T}$ are in one-to-one correspondence with the nodes (resp.\ edges) of $\mathcal{T}^*$: the column associated with node $\eta \in V(\mathcal{T}^*)$ is $N_\eta$, and the leftover set associated with the edge $(\eta, \eta') \in E(\mathcal{T}^*)$ is $L_{\eta'}$. 
In the proofs below, we will slightly abuse the notation and use $\eta$ to refer simultaneously to the column of $\mathcal{T}$ and to the corresponding node of $\mathcal{T}^*$; for example, the column $\eta$ denotes the same set as $N_\eta$. Notice that the subgraph $H_\eta$ defined in the column shortcut property (\Cref{def:column-shortcut}) is the same as $H^*_\eta$.
\begin{lemma}
\label{lem:thin-gridtree}
    Tree $\mathcal{T}$ is an $\e\Delta$-gridtree for $H$, and every external vertex of $H$ is assigned to some~column.
\end{lemma}
\begin{proof}
We prove that $\mathcal{T}$ satisfies the three required properties.

\medskip
\noindent \textbf{[Column adjacency.]} Let $(u,v)$ be an edge in $H$.
Let $\eta_u$ and $\eta_v$ be the nodes of $\mathcal{T}^*$ such that $u \in N_{\eta_u} \cup L_{\eta_u}$ and $v \in N_{\eta_v} \cup L_{\eta_v}$.
If $\eta_u = \eta_v$, then the edge $(u,v)$ satisfies the column adjacency property.
Otherwise, notice that $\eta_u$ and $\eta_v$ are in an ancestor-descendent relationship, as 
$\textsc{SelectPaths}$ recurses on (maximal) connected components in each iteration.
Assume without loss of generality that $\eta_u$ is an ancestor of $\eta_v$. For the sake of simplifying notation, we write $\eta \coloneqq \eta_u$.
Let $\eta'$ denote the child of $\eta$ such that $v$ appears in the subgraph of $\eta'$.
We claim that $v$ does not appear in the subgraph of any child of $\eta'$.
Indeed, if it did, then \Cref{lem:safe-recursion} would imply that the edge $(u,v)$ is incident to $\pi_{\eta'}$; as $u \not \in V(H^*_{\eta'})$, this implies that $v \in \pi_{\eta'} \subseteq N_{\eta'}$ and thus does not appear in the subgraph of any child of $\eta'$. 
We conclude that $v \in L_{\eta'} \cup N_{\eta'}$, and so $(u,v)$ satisfies the column adjacency property.

\medskip
\noindent \textbf{[Column width.]} Let $\eta$ be a column in $\mathcal{T}$.
    Let $P$ be a path that passes through $\eta$. 
    Let $P'$ denote the longest suffix 
    of $P$ that is fully contained in $H_\eta$.
    Because $V(H_\eta)$ includes all vertices below $\eta$ in $\mathcal{T}$, and $P$ starts at a vertex that is adjacent to a vertex above $\eta$, the path $P'$ starts at a vertex that is adjacent to a vertex above $\eta$.
    Column adjacency property implies that $P'$ starts at a vertex that adjacent to some vertex in $V(H^*_{\parent(\eta)})$.
    By \Cref{lem:safe-recursion}, the path $P'$ intersects $\pi_\eta$.
    As every vertex in $V(H^*_\eta) \setminus (V(L_\eta) \cup N_\eta)$ has distance (with respect to $H^*_\eta$) at least $\e\Delta$ from $\pi_\eta$, this implies that $P$ has length at least $\e \Delta$.

\medskip
\noindent \textbf{[Column shortcut.]}
    Let $\eta$ be a column in $\mathcal{T}$, defined to be the vertex set $N_\eta$ from the associated node of $\mathcal{T}^*$. 
    By construction, every vertex in $N_\eta$ is within distance $\e \Delta$ (with respect to $H[N_\eta]$) of the path $\pi_\eta$, and $\pi_\eta$ is a shortest path with respect to $H_\eta$.
\end{proof}

Notice that the definition of an $\e\Delta$-gridtree (column shortcut property) only requires that each vertex $v$ in a column $\eta$ is within distance $2 \cdot \e\Delta$ of the path $\pi_\eta$; however, the proof above gives a better guarantee on the distance.
\begin{observation}
\label{obs:thin-shortcut}
Let $v$ be a vertex in a column $\eta$ in the $\e\Delta$-gridtree $\mathcal{T}$, and let $\pi_\eta$ be the path in $\eta$ guaranteed by the column shortcut property.
Then $v$ is within distance $\e\Delta$ (with respect to $\eta$) of $\pi_\eta$.
\end{observation}

The gridtree $\mathcal{T}$ does \emph{not} guarantee that every vertex within $\e\Delta$ distance of the external face is assigned to a column. It only guarantees that every vertex on the external face itself is assigned to a column. 
 To fix this, we simply assign every vertex within $\e\Delta$ distance of the external face to the closest column in $\mathcal{T}$.
 We describe this approach in the \EMPH{$\textsc{Gridtree}(H)$} algorithm, which takes as input a planar graph $H$.

\begin{figure}[ht!]
\centering\small
\begin{algorithm}
\textul{$\mathsc{Gridtree}(H)$:} \+
\\  $\mathcal{T} \gets$ gridtree associated with $\textsc{SelectPaths}(G)$
\\
\\  \Comment{Select vertices to add to columns}
\\  $X \gets \varnothing$
\\  for each column $\eta$ of $\mathcal{T}$: \+
\\      $X \gets X \cup \set{v \in V(H) \mid \dist_H(v, \pi_\eta) \leq 2 \e \Delta}$ \-
\\
\\  \Comment{Assign vertices to closest $\eta$}
\\  initialize $\eta^+ \gets \eta$ for each column $\eta \in V(\mathcal{\mathcal{T}})$
\\  for each vertex $v$ in $X$: \+
\\      $\eta_v \gets \argmin_{\eta \in V(\mathcal{T})} \dist(v, \eta)$
\\      (breaking ties based on a fixed ordering of $V(\mathcal{T})$)
\\      $\eta_v^+ \gets \eta_v^+ \cup \set{v}$ \-
\\
\\ \Comment{Define the gridtree}
\\ $\mathcal{T}^+ \gets \mathcal{T}$
\\ assign each node $\eta \in V(\mathcal{T}^+)$ to the column $\eta^+$
\\ assign each edge $(\eta, \eta') \in E(\mathcal{T}^+)$ to the leftover set $(\eta, \eta') \in E(\mathcal{T})$, excluding the vertices in $X$
\\ return $\mathcal{T}^+$
\end{algorithm}
\end{figure}

\begin{claim}\label{claim:expand-external}
    Let $v$ be a vertex in $V(H)$. If $v$ is in a column $\eta$ in $\mathcal{T}$, then $v$ is in the associated column $\eta^+$ in $\mathcal{T}^+$. If $v$ is in some leftover set $(\eta_1, \eta_2)$ in $\mathcal{T}$, then in $\mathcal{T}^+$ the vertex $v$ is either in the column $\eta_1^+$, the column $\eta_2^+$, or the leftover set $(\eta_1^+, \eta_2^+)$.
\end{claim}
\begin{proof}
    If $v$ is in a column of $\mathcal{T}$, the claim is immediate from the construction. Suppose $v$ is in a leftover set $(\eta_1, \eta_2)$ of $\mathcal{T}$. If $v$ is not in the set $X$ during the call to $\textsc{Gridtree}(H)$, then $v$ stays in the associated leftover set $(\eta_1^+, \eta_2^+)$ in $\mathcal{T}^+$. Otherwise, $v$ is assigned to column $\eta^+$ in $\mathcal{T}^+$, where $\eta$ is the the closest column in $\mathcal{T}$ to $v$. Let $P$ be the shortest path from $v$ to $\eta$. By the column adjacency property of $\mathcal{T}$, the first vertex on $P$ that leaves the leftover set $(\eta_1, \eta_2)$ is either in $\eta_1$ or $\eta_2$. Thus, the closest column to $v$ is either $\eta_1$ or $\eta_2$.
\end{proof}

The following lemma is a restatement of \Cref{L:gridtree}.

\begin{lemma}
    The tree $\mathcal{T}^+$ returned by $\textsc{Gridtree}(H)$ is a gridtree for $H$, in which each vertex within distance $\e \Delta$ of an external vertex of $H$ is assigned to a column.
\end{lemma}
\begin{proof}
As every external vertex in $H$ is assigned to a column in $\mathcal{T}$, and (by \Cref{obs:thin-shortcut}) every such vertex is within distance $\e\Delta$ of a path $\pi_\eta$, triangle inequality implies that every vertex $v$ within distance $\e\Delta$ of some external vertex is within distance $2 \e \Delta$ of some $\pi_\eta$. 
Thus, $v$ is assigned to a column in $\mathcal{T}^+$. We now prove that $\mathcal{T}^+$ is a gridtree.

\medskip
\noindent \textbf{[Column adjacency.]} Let $(u,v) \in E(H)$.
If both $u$ and $v$ were in columns of $\mathcal{T}$, then \Cref{lem:thin-gridtree} implies that $u$ and $v$ are in adjacent columns of $\mathcal{T}$.
\Cref{claim:expand-external} implies that $u$ and $v$ are also in adjacency columns of $\mathcal{T}^+$.
Otherwise, suppose that $u$ is in some leftover set associated with the edge $(\eta_1, \eta_2)$ of $\mathcal{T}$.
By \Cref{lem:thin-gridtree}, the vertex $v$ is either associated with the leftover set $(\eta_1, \eta_2)$ in $\mathcal{T}$, or with the columns $\eta_1$ or $\eta_2$ in $\mathcal{T}$.
By \Cref{claim:expand-external}, the vertices $u$ and $v$ are associated with $(\eta_1^+, \eta_2^+)$ or $\eta_1^+$ or $\eta_2^+$ in $\mathcal{T}^+$.
In each of these cases, the edge $(u,v)$ satisfied the column adjacency property in $\mathcal{T}$.

\medskip
\noindent \textbf{[Column width.]} Let $\eta^+$ be a column in $\mathcal{T}^+$, corresponding to column $\eta$ in $\mathcal{T}$. Let $P$ be a path that passes through $\eta^+$.
Notice that every vertex above (resp. below) $\eta^+$ in $\mathcal{T}^+$ is above (resp. below) $\eta$ in $\mathcal{T}$: this follows from \Cref{claim:expand-external}.
Thus, $P$ passes through $\eta$ in $\mathcal{T}$, and so the column width property of $\mathcal{T}$ implies that it has length at least $\e\Delta$.

\medskip
\noindent \textbf{[Column shortcut.]} Let $\eta^+$ be a column in $\mathcal{T}^+$, corresponding to column $\eta$ in $\mathcal{T}$.
Every vertex $v$ in $\eta^+$ is within $2\e \Delta$ distance of $\pi_{\eta}$. Because $\textsc{Gridtree}(H)$ breaks ties based on a fixed ordering of $V(\mathcal{T})$ when assigning vertices to columns, every vertex in the shortest path between $v$ and $\pi_\eta$ is in $\eta^+$. We now claim $\pi_\eta$ is a shortest path with respect to $H_{\eta^+}$. Indeed, \Cref{claim:expand-external} implies that columns and leftover sets are below $\eta$ in $\mathcal{T}$ if and only if they are below $\eta^+$ in $\mathcal{T}$, from which we can conclude that $H_\eta$ is a superset of $H_{\eta^+}$. By column shortcut property for $\mathcal{T}$, path $\pi_\eta$ is a shortest path with respect to $H_{\eta}$, and thus is a shortest path with respect to $H_{\eta^+}$.
\end{proof}

\subsection{Constructing a gridtree hierarchy}
We summarize our multi-round strategy with the recursive algorithm $\textsc{GridtreeHierarchy}(G)$.

\begin{figure}[h!]
\label{F:gridtree-hierarachy}
    \centering\small
    \begin{algorithm}
    \textul{{$\textsc{GridtreeHierarchy}(G)$:}} \+
    \\  $\mathcal{T} \gets \textsc{Gridtree}(G)$
    \\  set $(G,\mathcal{T})$ to be the root of the hierarchy $\mathcal{H}$
    \\  for each connected component $H$ in a leftover set of $\mathcal{T}$: \+
    \\      $\mathcal{H}_H \gets \textsc{GridtreeHierarchy}(H)$ 
    \\      attach the root of $\mathcal{H}_H$ as a child of the node $(G,\mathcal{T})$ in $\mathcal{H}$ \-
    \\  return $\mathcal{H}$
    \end{algorithm}
\end{figure}

\begin{proof}[of Lemma~\ref{L:gridtree-planar}]
We prove that $\textsc{GridtreeHierarchy}(G)$ returns a gridtree hierarchy for $G$.
The layer nesting property is immediate from the construction. We now prove the layer width property.
Let $\mu = (H^\mu, \mathcal{T}^\mu)$ be some node in the hierarchy, and let $\mu' = (H^{\mu'}, \mathcal{T}^{\mu'})$ be the parent of $\mu$.
If a vertex $v$ in $H$ is adjacent to some vertex $u$ in $H^{\mu'}$, then we claim  that $v$ is an external vertex of $H^\mu$.
Indeed, the vertices in columns of $\mathcal{T}^{\mu'}$ induce a connected component that includes an external vertex of $H^{\mu'}$. 
When an external vertex of a planar graph is removed, its neighbors becomes external vertices of the resulting graph.
The graph $H^\mu$ is a connected component of the graph obtained by removing all columns of $\mathcal{T}^{\mu'}$. We conclude that $v$ is an external vertex of $H^\mu$.
\qed

\end{proof}


\section{Constructing tree cover from grid-like clustering}
\label{S:treecover-gridlike}

Before describing the construction, we prove two lemmas about gridtree hierarchies, which are easy consequences of the layer width and column width properties.

\begin{lemma}
\label{L:layer-depth}
    Any $\e \Delta$-gridtree hierarchy $\mathcal{H}$ for a graph $G$ with diameter $\Delta$ has depth $O(1/\e)$.
\end{lemma}
\begin{proof}
We prove by induction that for any
node $\mu$ of $\mathcal{H}$, every vertex $v$ that has distance at most $\alpha\cdot\e\Delta$ to an outer vertex of $\mu$ belongs to some column of a gridtree at most $\alpha$ layers deeper than $\mu$.
To prove the lemma, we substitute $\alpha = 1/\e$; as $G$ has diameter $\Delta$, every vertex has distance at most $\Delta$ from an outer vertex on layer 2 (recall that the root node at layer 1 has no outer vertices). We conclude that gridtree hierarchy $\mathcal{H}$ with depth $O(1/\e)$ much have every vertex assigned to some column of a gridtree in $\mathcal{H}$.

Consider an arbitrary vertex $v$ in $\mu$ and a shortest path $\pi$ from $v$ to the exterior of $\mu$, whose length is at most $\alpha\cdot\e\Delta$.
By layer width property (guaranteed by Lemma~\ref{L:gridtree}), every vertex in the prefix of $\pi$ that has distance at most $\e\Delta$ to an outer vertex of $\mu$ belongs to some column of the gridtree $\mathcal{T}^\mu$.
If $v$ belongs to some column of $\mathcal{T}^\mu$ as well then we are done.
Otherwise, the suffix of $\pi$ that lies completely in a leftover component $L$ containing $v$ has length at most $(\alpha-1)\cdot\e\Delta$.  
By induction, 
$v$ belongs to some column of a gridtree $\alpha-1$ layers deeper than $L$.
This shows that $v$ belongs to a column of some gridtree at most $\alpha$ layers from $\mu$.
\end{proof}

\begin{claim}\label{claim:far-layers}
    Let $\eta$ and $\eta'$ be two columns in a gridtree $\mathcal{T}$ of a graph $H$. Suppose that $\eta$ is an ancestor of $\eta'$, and that the path between $\eta$ and $\eta'$ in $\mathcal{T}$ consists of $m$ columns $\eta_1, \ldots, \eta_m$ strictly between $\eta$ and $\eta'$. Then any path $P$ in $H$ between a vertex $u$ in $\eta$ and a vertex $v$ in $H_{\eta'}$ has length at least $m \cdot \e \Delta$.
\end{claim}
\begin{proof}
    For ease of notation, we refer to $\eta$ as $\eta_0$ and refer to $\eta'$ as $\eta_{m+1}$. By the column adjacency property, $P$ touches at least one vertex in each of $\eta_1, \ldots, \eta_{m+1}$. For every $i \in \set{1, \ldots, m+1}$, let $p_i$ be the first vertex in $P$ that is in column $\eta_i$, as $P$ travels from $u$ to $v$. The path $P$ can be written as a concatenation of subpaths $P[u, p_1] \circ P[p_1, p_2] \circ \ldots \circ P[p_{m+1}, v]$.
    Column adjacency property implies that every $p_i$ is adjacent (in $H$) to a vertex above $\eta_i$ in $\mathcal{T}$, and $p_{i+1}$ is a vertex below $\eta_i$ in $\mathcal{T}$. From the column width property, we conclude that each of the $m$ subpaths $P[p_i, p_{i+1}]$ has length at least $\e\Delta$.
\end{proof}

\subsection{Construction}

Let $G$ be a graph with an $\e\Delta$-gridtree hierarchy $\mathcal{H}$. 
Recall from \Cref{S:gridtree-scatterting} that $G$ has an $(O(\e), O(\e^{-1}))$-shortcut partition $\mathcal{C}$ with $O(1)$ distortion, associated with $\mathcal{H}$, 
where the clusters satisfy the following properties: 

\begin{itemize}
    \item \textnormal{[Diameter.]}
    There exists some constant \EMPH{$\gamma$} $\ge1$ such that $\delta_C(u,v) \le \gamma \cdot \e\Delta$ for any $u$ and $v$ in the same cluster $C$.
    
    \item \textnormal{[Low-hop.]}
    $\cost_{O(1), \mathcal{C}}{(u,v)} \le O(\e^{-1})$ for any $u$ and $v$ in $G$.%

    \item \textnormal{[Cluster ordering.]}
    For every node $(H, \mathcal{T})$ in the hierarchy $\mathcal{H}$ and for every column $\eta$ in $\mathcal{T}$, there is an ordered set of \EMPH{cluster centers} $c_1, \ldots, c_m \in \eta$, one from each cluster containing vertices of $\eta$, such that for every pair $i, j \in \set{1, \ldots, m}$, we have $d_{H_\eta}(c_i, c_j) \geq \abs{i - j}\cdot \e\Delta$.
\end{itemize}

We use partition $\mathcal{C}$ to build a forest cover for $G$. To construct the forest cover, we only use the cluster ordering property. In a later section (cf. Section~\ref{S:treewidth-embedding}), we show that $G$ can be embedded into a bounded-treewidth graph; this proof will use the fact that a forest cover for $G$ can be built using a partition with the diameter and low-hop properties. In this section, we use the cluster centers associated with $\mathcal{C}$ in an algorithm \EMPH{$\textsc{CoverGridtree}(H, \mathcal{T})$} (see Figure~\ref{F:cover-gridtree}), which takes as input a node $(H, \mathcal{T})$ from the hierarchy $\mathcal{H}$, where $H$ is a subgraph and $\mathcal{T}$ is a gridtree on $H$.
We show that this procedure returns a set of $O(\e^{-2})$ forests (cf. \Cref{lem:gridtree-forest}), and that those forests that preserve distances for vertices in columns of $\mathcal{T}$ (implicit in \Cref{lem:hierarchy-distortion}).

\begin{figure}[h!]
\centering\small
\begin{algorithm}
\textul{$\textsc{CoverGridtree}(H, \mathcal{T})$}: \+
\\  \Comment{Create forests for each column}
\\  for each column $\eta$ in $\mathcal{T}$: \+
\\      $c_1, \ldots, c_m \gets$ cluster centers for $\eta$ from the shortcut partition $\mathcal{C}$
\\      for each $c_i \gets c_1, \ldots, c_m$: \+
\\          $B_i \gets$ ball of radius $(\gamma + 1)\Delta$ with respect to $H_\eta$, centered at $c_i$
\\          $T^\eta_i \gets$ SSSP tree (w.r.t.\ $H_\eta$) of $B_i$ rooted at $c_i$ \-
\\      for $j \gets 0, \ldots, (3\gamma + 2)\e^{-1}-1$: \+
\\          $F^{\eta}_j \gets$ forest containing trees $T^\eta_{j + k\cdot (3\gamma + 2)\e^{-1}}$ for every integer $k$
\\          \Comment{Note: If there are fewer than $j$ portals on $\pi_\eta$, then $F^\eta_j = \varnothing$}\-\-
\\
\\  \Comment{Combine disjoint forests from columns that are far apart}
\\  for each $\ell \gets 0, \ldots, (\gamma + 2) \e^{-1}-1$:\+
\\      for each $j \gets 0, \ldots (3\gamma + 2)\e^{-1}-1$: \+
\\          $F^{\ell}_j \gets \bigcup_\eta F^\eta_j$ over column $\eta$ whose level equals to $\ell \pmod{(\gamma + 2)\e^{-1}}$ in $\mathcal{T}$\-\-
\\
\\  return $\{F^\ell_j\}$ where $\ell \in [(\gamma + 2)\e^{-1}]$ and $j \in [(3\gamma + 2)\e^{-1}]$
\end{algorithm}
\caption{Algorithm $\textsc{CoverGridtree}(H, \mathcal{T})$ constructs a forest cover for vertices in columns of $\mathcal{T}$.}
\label{F:cover-gridtree}
\end{figure}

\begin{lemma}\label{lem:gridtree-forest}
    The procedure $\textsc{CoverGridtree}(H, \mathcal{T})$ returns a set of $\kappa$ spanning forests on $H$, for some $\kappa = O(\e^{-2})$.
\end{lemma}
\begin{proof}
    From the procedure $\textsc{CoverGridtree}$,  \EMPH{$F_j^\eta$} is the union of spanning trees in the set $\set{T_{j+k\cdot (3\gamma + 2)\e^{-1}}^\eta : \text{integer $k$}}$.
    We first show that, for every column $\eta$ in $\mathcal{T}$, the graph $F_j^\eta$ is a spanning forest of $H$.  
    Let $T_{i_1}$ and $T_{i_2}$ be two trees from the set, with $T_{i_1} \neq T_{i_2}$. 
    We want to prove $T_{i_1}$ and $T_{i_2}$ are disjoint.
    Let $c_{i_1}$ and $c_{i_2}$ denote the cluster centers that are roots of $T_{i_1}$ and $T_{i_2}$, respectively.
    If $T_{i_1}$ and $T_{i_2}$ were not disjoint, then there would be some vertex $v$ in $T_{i_1} \cap T_{i_2}$ such that $\dist_{H_\eta}(c_{i_1}, v) \leq (\gamma + 1)\Delta$ and $\dist_{H_\eta}(c_{i_2}, v) \leq (\gamma + 1)\Delta$.
    But the choice of cluster centers guarantees that $c_{i_1}$ and $c_{i_2}$ are at distance at least $\abs{i_1 - i_2}\cdot \e\Delta > (3\gamma + 2) \Delta$ with respect to $H_\eta$. Thus, there are no vertices in both $T_{i_1}$ and $T_{i_2}$.

    Forest \EMPH{$F^\ell_j$} is the union of each spanning forest $F^\eta_j$ over the nodes $\eta$ whose level \!$\pmod {(\gamma + 2)\e^{-1}}$ in $\mathcal{T}$ is equal to $\ell$, defined in the procedure $\textsc{CoverGridtree}$.
    We now show that, for each level $\ell$ in the gridtree, the graph $F^\ell_j$ is a spanning forest. 
    Let $\eta_1$ and $\eta_2$ be two such nodes. We want to prove that $F^{\eta_1}_j$ and $F^{\eta_2}_j$ are disjoint. 
    If $\eta_1$ and $\eta_2$ are not in an ancestor-descendent relationship, then $V(H_{\eta_1})$ and $V(H_{\eta_2})$ are disjoint and so the claim holds.
    Otherwise, suppose without loss of generality that $\eta_1$ is an ancestor of $\eta_2$.
    There are at least $(\gamma + 2)\e^{-1} - 1$ columns separating $\eta_1$ and $\eta_2$ in the gridtree $\mathcal{T}$. By \Cref{claim:far-layers}, every vertex in $\eta_1$ is at least $((\gamma + 2)\e^{-1} - 1) \cdot \e\Delta > (\gamma + 1) \Delta$ distance away from the closest vertex in $H_{\eta_2}$.
    As every vertex in $F^{\eta_1}_j$ is within distance $(\gamma + 1)\Delta$ of some vertex in $\eta_1$, the forests $F^{\eta_1}_j$ and $F^{\eta_2}_j$ are disjoint. 
    Thus, each of the $(3\gamma+2)(\gamma+2)\e^{-2} = O(\e^{-2})$ graphs $F^\ell_j$ returned by $\textsc{CoverGridtree}(\mathcal{T})$ is a forest.
\end{proof}

As a brief aside, we can in fact prove something slightly stronger. We say that a two (spanning) trees are \EMPH{cluster-disjoint} with respect to $\mathcal{C}$ if no cluster in $\mathcal{C}$ intersects nontrivially with the two trees. The following lemma will not be used in the proof of $O_\e(1)$-size tree cover, but it will be helpful in Section~\ref{S:treewidth-embedding}, when we embed any planar graph into a bounded-treewidth graph.
\begin{lemma}
    Every forest returned by $\textsc{CoverGridtree}(H, \mathcal{T})$ contains trees that are pairwise cluster-disjoint with respect to $\mathcal{C}$.
\end{lemma}
\begin{proof}
    The proof follows that of \Cref{lem:gridtree-forest} almost verbatim. We use one additional fact: By the construction of clusters based on the \textsc{ClusterGridtree} algorithm from Section~\ref{SS:cluster},
    every vertex in a subgraph $H_\eta$ belongs to a cluster that is fully contained within $H_\eta$.

    We first prove that two trees $T_{i_1}$ and $T_{i_2}$ in a forest $F^\eta_j$ are cluster-disjoint. 
    This follows from the fact that each cluster has strong diameter $\gamma \e\Delta < \gamma \Delta$: if there was some cluster $C \subseteq H_\eta$ containing both a vertex $v_1$ in $T_{i_1}$ and a vertex $v_2$ in $T_{i_2}$, then triangle inequality would imply that $\dist_{H_\eta}(c_{i_1}, c_{i_2}) \leq \dist_{H_\eta}(c_{i_1}, v_1) + \dist_C(v_1, v_2) + \dist_{H_\eta}(v_2, c_{i_2}) \le (2(\gamma + 1) +\gamma \e)\Delta < (3\gamma + 2) \Delta$, a contradiction.

    We then prove that two forests $F^{\eta_1}_j$ and $F^{\eta_2}_j$ in a forest $F^{i}_j$ are cluster-disjoint. 
    If there was some cluster containing both a vertex in $F^{\eta_1}_j$ and $F^{\eta_2}_j$, then that cluster must be entirely in $H_{\eta_2}$ --- but we proved above that $F^{\eta_1}_j$ contains no vertices in $H_{\eta_2}$.
\end{proof}

We now construct a tree cover for $G$ by repeatedly apply $\textsc{CoverGridtree}$ to the gridtrees in each layer of the gridtree hierarchy $\mathcal{H}$. 
The procedure \EMPH{$\textsc{CoverHierarchy}(G, \mathcal{H})$} takes as input a depth-$d$ gridtree hierarchy $\mathcal{H}$ of the graph $G$ and returns a forest cover of size $O(d\cdot\e^{-2})$.

\begin{figure}[ht!]
\centering \small
\begin{algorithm}
\textul{$\textsc{CoverHierarchy}(G, \mathcal{H})$}\+
\\  $d \gets $ depth of gridtree hierarchy $\mathcal{H}$
\\  for each layer $\lambda \gets 1, \ldots, d$ of $\mathcal{H}$:\+
\\      for each node $(H^\mu, \mathcal{T}^\mu)$ in layer $\lambda$ of $\mathcal{H}$: \+
\\          \Comment{$\kappa$ is the constant from \Cref{lem:gridtree-forest}}
\\          $F^{\mu}_1, \ldots F^{\mu}_{\kappa} \gets \textsc{CoverGridtree}(H^\mu, \mathcal{T}^\mu)$
\-
\\      for each index $j \gets 1, \ldots, \kappa$: \+
\\          $F^{\lambda}_j \gets \bigcup_{\text{$\mu$ at layer $\lambda$}} F^{\mu}_j$ \-\-
\\ return the set $\set{F^{\lambda}_j}$ for $\lambda \in [d]$ and $j \in [\kappa]$
\end{algorithm}
\end{figure}

\begin{lemma}
    \label{lem:hierarchy-forest}
    Let $G$ be a graph with a gridtree hierarchy $\mathcal{H}$. Then the algorithm $\textsc{CoverHierarchy}(G, \mathcal{H})$ returns a set of $O(d \cdot \e^{-2})$ spanning forests on $G$. Further, the trees in any forest are pairwise cluster-disjoint with respect to $\mathcal{C}$.
\end{lemma}
\begin{proof}
    For every layer $\lambda$ and index $j$, the graph \EMPH{$F^{\lambda}_j$} is the union of forests $F^{\mu}_j$
    over all nodes $(H^\mu, \mathcal{T}^\mu)$ in layer $\lambda$ of the gridtree hierarchy $\mathcal{H}$. 
    By \Cref{lem:gridtree-forest}, each $F^{\mu}_j$ is a spanning forest of $H^\mu$. 
    For every $H^\mu$ in a single layer $\lambda$, the sets $V(H^\mu)$ are disjoint, so $F^{\lambda}_j$ is the disjoint union of forests. 
    Further, every vertex in $H^\mu$ belongs to a cluster entirely within $H^\mu$, so the trees in $F^{\lambda}_j$ are pairwise cluster-disjoint. There are $d \cdot \kappa = O(d \cdot \e^{-2})$ such forests $F^{\lambda}_j$.
\end{proof}

\begin{lemma}\label{lem:hierarchy-distortion}
    Let $G$ be a graph with diameter $\Delta$ and a gridtree hierarchy $\mathcal{H}$. Then for every pair of vertices $(u,v)$ in $G$, there is some forest $F$ returned by the $\textsc{CoverHierarchy}(G,\mathcal{H})$ such that $\dist_F(u,v) \leq \dist_G(u,v) + O(\e\Delta)$.
\end{lemma}

\begin{proof}
    Let $P$ be a shortest path between $u$ and $v$ in $G$. 
    Let $(H, \mathcal{T})$ be the lowest node of $\mathcal{H}$ such that $H$ contains all vertices of $P$. Note that this implies $\dist_{H}(u,v) = \dist_G(u,v)$.

    Let $\eta$ be the highest column in $\mathcal{T}$ that contains a vertex in $P$. Let $p$ be a vertex in $P \cap \eta$. 
    There is some cluster center $c_i$ in $\eta$ with $\dist_{H_\eta}(c_i, p) \leq \gamma \e\Delta$. We now claim that $H_\eta$ contains all vertices of $P$: Indeed, the column adjacency property of $\mathcal{T}$ implies that vertices of $H_\eta$ are only incident other vertices in $H_\eta$ and to vertices in $\parent(\eta)$, and the choice of $\eta$ implies that $P$ contains no vertices in $\parent(\eta)$.
    Thus, we have $\dist_{H_\eta}(u, c_i) \leq \dist_G(u, p) + \gamma\e\Delta < (\gamma+1)\Delta$ and $\dist_{H_\eta}(c_i, v) \leq \dist_G(p, v) + \gamma \e\Delta < (\gamma+1) \Delta$. 
    The algorithm $\textsc{CoverGridtree}(\mathcal{T})$ returns a forest $F$ containing an SSSP tree connecting root $c_i$ to every vertex in $H_\eta$ within distance $(\gamma+1)\Delta$ of $c_i$. This forest satisfies
     $\dist_{F}(u,v) \leq \dist_{G}(u,v) + 2\gamma\e\Delta$.
\end{proof}

\begin{observation}
    Every tree in a forest returned by $\textsc{CoverGridtree}(H, \mathcal{T})$ is an SSSP tree with radius $(\gamma+1)\Delta$. Thus, every tree has diameter $O(\Delta)$.
\end{observation}

\noindent Together with the fact that $\e\Delta$-gridtree hierarchy $\mathcal{H}$ of $G$ has depth $O(\e^{-1})$ from \Cref{L:layer-depth}, we conclude:

\begin{theorem}
\label{thm:planar-cover}
Every planar graph $G$ with diameter $\Delta$ has an $O(\Delta)$-bounded spanning forest cover of size $O(\e^{-3})$, with additive distortion $+O(\e\Delta)$.
\end{theorem}

\noindent Combining with the multiplicative-to-additive distortion reduction (\Cref{lm:reduction}), this proves~\Cref{thm:main}.



\section{A shortcut partition for bounded treewidth graphs}
\label{S:weak-scattering-treewidth}

In this section we present an algorithm for creating an $(\e, O_\e(1))$-shortcut partition for graphs with $O(1)$-treewidth. Our algorithm is inspired by the technique developed by Friedrich et al.~\cite{FIKMZ23} for solving the mulicut problem in bounded treewidth graphs. 

\subsection{Algorithm description}

Let $\e \in (0,1)$ be a fixed constant. Our algorithm takes as input a graph $G$ with diameter $\Delta$, and its width-$k$ tree decomposition $T$. 
It returns a set of clusters $\mathcal{C}'$ that form an $(\e, O_\e(1))$-shortcut partition. 

In the preprocessing phase, we modify $G$ and $T$ so that the resulting graph has the same treewidth up to a small constant, but the resulting tree decomposition becomes more suitable for creating clusters; we abuse the notation and refer to the resulting graph and tree decomposition as $G$ and $T$. 
First, for each bag in $T$ we connect all its vertices by a clique in $G$.
Note that each vertex $v$ in $G$ may appear in multiple bags of $T$;
in every such bag, we create a fresh copy of $v$. 
In addition, for every pair of bags $X$ and $Y$ when both contain copies of $v$ and $X$ is a parent of $Y$, we add a 0-weight edge between the corresponding copies of $v$ in $G$.  This concludes the preprocessing phase. 

Next, we build a set of ``preliminary clusters'' $\mathcal{C}$, which might not be vertex-disjoint. This is done in $k+1$ rounds, numbered from $k+1$ to $1$ for technical convenience. In the $i$th round, we invoke a recursive procedure \EMPH{\textsc{ClusteringRound}} that creates clusters. A recursive call of this procedure operates on the root bag $R$ of a subtree \EMPH{$T'$} (of the entire tree $T$) and grows balls of radius $\eps\Delta$ centered at each of the unclustered vertices in $R$. Importantly, the ball is restricted to the vertices in the subtree $T'$ (rather than the whole tree). 
Let \EMPH{$B$} be the set of vertices clustered in this recursive call. The algorithm proceeds recursively with the set of subtrees obtained by removing from $T'$ all the bags containing vertices in $B$. 

As mentioned, the preliminary clusters in $\mathcal{C}$ created in the $k+1$ rounds of the aforementioned recursive procedure might not be vertex-disjoint. To achieve disjoint clusters, the algorithm creates a dummy graph \EMPH{$G'$} by adding a vertex $s$ and connecting it via weight-0 edge to all the centers in $\mathcal{C}'$. The clusters are then obtained by computing an SSSP tree $\tau_s$ of $G'$ rooted at $s$. Specifically, for every child $u$ of $s$, we let $B'_u$ be the set of vertices in the subtree of $\tau_s$ rooted at $u$ and add it to $\mathcal{C}'$. The output is $\mathcal{C}'$.
Our algorithm is described in more detail in \Cref{fig:clustering-tw}.

\begin{figure*}[h!]
\centering\small
\begin{algorithm}
\textul{$\textsc{Clustering}(T, G)$}: \+ 
\\  \Comment{Preprocess $G$ and $T$}
\\  for each bag in $T$, connect its vertices by a clique
\\  for each $v \in V(G)$:\+
\\      replace every occurrence of $v$ with a fresh copy
\\      for every pair parent-child in $T_v$ add an edge between the copies of $v$\-
\\  $\mathcal{C} \gets \varnothing$
\\  for $i \gets k+1$ down to $1$: \Comment{Proceed in $k+1$ rounds, where $k$ is the treewidth of $G$}\+
\\      $\textsc{ClusteringRound}(T, G, i, \mathcal{C})$\-
\\
\\  \Comment{$\mathcal{C}$ may contain overlapping preliminary clusters}
\\ $\mathcal{C}' \leftarrow \varnothing$ 
\\ $G' \gets$ add to $G$ a dummy vertex $s$ and connect $s$ to every cluster center in $\mathcal{C}'$
\\ \Comment{all new edges have weight $0$}
\\ $\tau_s \gets$ an SSSP tree of $G'$ rooted at $s$ 
\\  for each child $u$ of $s$ in $\tau_s$: \+
\\ $B'_u \gets$ the set of vertices in the subtree of $\tau_s$ rooted at $u$
\\ $\mathcal{C}' \leftarrow \mathcal{C}' \cup{B'_u}$ \-
\\  return $\mathcal{C}'$\-
\end{algorithm}

\vspace{6pt}
\centering\small
\begin{algorithm}
\textul{$\textsc{ClusteringRound}(T, G, i, \mathcal{C})$}: \+
\\  $R \gets$ root bag of $T$ 
\\  \Comment{If $R$ contains unclustered vertices, $R$ is called a round-$i$ root bag}
\\  \Comment{Otherwise, we continue recursively with each one of its children}
\\  $B \gets \varnothing$
\\  for each unclustered $u$ in $R$:\+
\\      $B_u \gets \{v \in T \mid \delta_G(u,v) \le \eps\Delta\}$ \Comment{Cluster downwards in the tree}
\\ $\mathcal{C} \leftarrow \mathcal{C} \cup B_u$ \Comment{Add $B_u$ as a preliminary-cluster}
\\  $B \gets B \cup B_u$\-
\\  $S(B) \gets$ the set of bags in $T$ containing a vertex from $B$ 
\\ if $R$ does not contain unclustered vertices: \+ 
\\      $S(B) \gets \{R\}$ \-
\\ \Comment{The subgraph of $T$ induced by $S(B)$ is connected}
\\  $T_1, \ldots, T_g \gets$ subtrees of $T \setminus S(B)$
\\  for $j \gets 1,\ldots, g$: \+
\\		$\textsc{ClusteringRound}(T_j, G, i, \mathcal{C})$
\end{algorithm}
\caption{Algorithms $\textsc{ClusteringRound}(T, G, i, \mathcal{C})$ and $\textsc{Clustering}(T, G, \mathcal{C})$}
\label{fig:clustering-tw}
\end{figure*}

\subsection{Algorithm analysis}

\begin{observation}\label{obs:shareRoot}
Consider any call of \textsc{ClusteringRound}$(T,G,i,\mathcal{C})$.
Let $S(B)$ 
be the set of bags in $T$ containing a vertex from $B$. Then the subgraph of $T$ induced by $S(B)$ is connected.
\end{observation}
\begin{proof}
We shall assume that the root bag $R$ contains at least one unclustered vertex, otherwise the subgraph of $T$ induced by $S(B)$ is a single node $R$.
Recall that $B$ is the union of balls $B_u$, taken over all centers $u$ from the root bag $R$.

Consider any vertex $u$ in $R$ and the corresponding ball $B_u$. Since $B_u$ is connected in $G$, we note that the bags containing vertices in $B_u$ form a connected subtree of $T$. 
(This follows from the facts that (1) the bags containing any vertex $v$  form a connected subtree of $T$, and (2) the bags containing vertices along a path between any pair $(u,v)$ form a connected subtree of $T$.)
 
Taking union over all centers in $R$ gives us a union of connected subtrees in $T$ that all contain bag $R$ in common. It follows that the union is connected.
\end{proof}

\begin{lemma}\label{lem:tw-cluster-diameter}
The set $\mathcal{C}'$ of clusters returned by $\textsc{Clustering}(T, G)$ form a partition of $V$.
Moreover, every cluster is connected and has a strong diameter at most $O(\eps \Delta)$.
\end{lemma}
\begin{proof}
In every round $i$, the algorithm visits every bag $X$ of $T$ and either clusters all the unclustered vertices in it (when $X$ is a round-$i$ root bag) or clusters at least one vertex in it. It follows that after $k+1$ rounds every vertex becomes clustered.
The set of resulting clusters in $\mathcal{C}$ (each of which is a ball $B_u$) is a cover of $V$ rather than a partition.
After the last round in \textsc{Clustering}, the algorithm transforms $\mathcal{C}$ to a collection $\mathcal{C}'$ of {\em disjoint} clusters, where every vertex that is clustered under $\mathcal{C}$ remains clustered in $\mathcal{C}'$.

Every cluster is formed by taking all vertices in a subtree of the SSSP tree rooted at the dummy vertex~$s$, and is thus connected. By definition, every cluster of $\mathcal{C}$, which is a ball of radius $\eps\Delta$, has a strong diameter of at most $O(\eps\Delta)$. After running the SSSP algorithm, every cluster may reduce its size, but all the vertices remaining in that cluster are at the same distance from the respective center as before. 
\end{proof}

\begin{definition}\label{def:decompWalk}
Let $\pi = \pi_G(u,v) = \langle u = x_0, x_1, \ldots, x_\ell = v\rangle$ be a path between $u$ and $v$ in $G$ and let $T$ be a tree decomposition of $G$. We say that a walk $w$ in $T$ \emph{corresponds} to $\pi$ if it is the shortest walk in $T$ that visits bags containing $(x_i, x_{i+1})$ for  $0 \le i \le \ell-1$ in that order.
\end{definition}

\begin{observation}\label{obs:path}
    For every pair of points $u$ and $v$ in $G$, there is a shortest path $\pi(u,v)$ in $G$ such that the corresponding walk $\pi_T$ in $T$ is a simple path. This follows from the fact that every bag of $T$ induces a clique in $G$.
\end{observation}

\begin{observation}\label{obs:hierarchy}
Let $X$ and $Y$ be two round-$i$ root bags such that $X$ is an ancestor of $Y$. Then, the distance between every center of $X$ clustered at the $i$th round and every center of $Y$ clustered at the $i$th round is greater than $\eps\Delta$.
\end{observation}
\begin{proof}
Suppose towards contradiction that there were two centers of round-$i$ clusters $x \in X$ and $y \in Y$ at distance at most $\eps\Delta$. 
Consider the recursive call of \textsc{ClusteringRound} when the ball centered at $x$ is created. 
Then, $y$ is in $B_x$ by definition, as
it is distance at most $\eps\Delta$ from $x$, $y \in Y$ and $Y$ is a descendant of $X$. Hence, the bag $Y$ is in $B$, and by the algorithm description $Y$ will not be visited in any recursive call at the $i$th round. Hence, $Y$ cannot be a round-$i$ bag, yielding a contradiction. \end{proof}

\begin{observation}\label{obs:hierachy2}
Consider an arbitrary round $i \in \set{k+1, \ldots, 1}$ during a call to $\textsc{Clustering}$.
Let $\pi$ be a path in $G$ that corresponds to a path $\pi_T$ in $T$. Suppose that $\pi_T$ contains no round-$i'$ root bag, for all of the rounds $i' \ge i$. Then, the vertices of $\pi$ are contained in at most $k+1$ clusters created at the $i$th round.
\end{observation}
\begin{proof}

Since no bag in the path $\pi_T$ is a round-$i$ bag, the algorithm implies that every bag on $\pi_T$ intersects a cluster created at some ancestor root bag of the $i$th round. (This is true because the algorithm creates clusters only \emph{down} the tree, and every occurrence of a vertex in the tree is replaced by a fresh copy.)
By \Cref{obs:shareRoot}, there is a single such ancestor root bag, which gives rise to at most $k+1$ clusters. Thus, vertices in the bags of $\pi_T$ are contained in at most $k+1$ clusters in the $i$-th round. The observation follows from the fact that every vertex in $\pi$ lies in some bag in $\pi_T$.
\end{proof}

\begin{lemma}\label{obs:simplePath}
Let $i \in \set{k+1, \ldots, 1}$ be an arbitrary round in a call to $\textsc{Clustering}$.
Let $u$ and $v$ be any two vertices in $G$,  let $\pi$ be a shortest path between $u$ and $v$ in $G$, and let $\pi_T$ be a simple path  corresponding to $\pi$. Suppose that there is no round-$i'$ root bag on $w$, for any earlier round $i' > i$. Then, the union of the vertices in bags in $\pi_T$ is contained in at most $J_{\textit{i}}=O((i/\eps)^i)$ clusters created at any round $i, i-1, \ldots, 1$.
\end{lemma}

\begin{proof}
The proof proceeds by induction on the round number $i$.
For the basis $i=1$, there at most $k+1$ clusters on $\pi$
by \Cref{obs:hierarchy}.

For the induction step, we assume the statement holds for all values smaller than $i$ (i.e., for all later rounds) and proceed to prove it for $i$, with $i > 1$. We first consider clusters created at the $i$th round.  Every time the algorithm creates new clusters, it creates up to $k+1$ clusters from the vertices of some root bag. 
Suppose first that $\pi_T$ does not pass through any root bag. By \Cref{obs:hierachy2}, vertices in bags of $\pi_T$ are contained in at most $k+1$ clusters, created by some root bag higher than the bags visited by $\pi_T$.
Otherwise, let $X_1, X_2, \ldots, X_\ell$ be the root bags which $\pi_T$ passes through. The number of different round-$i$ clusters visited is at most $(k+1)(\ell+1)$, since the path touches at most $(k+1)\ell$ clusters created by the root bags it visits and at most $k+1$ additional clusters created by some root bag higher than the highest bag of the path. 
We now find an upper bound for $\ell$.
Consider a pair $X_i, X_{i+1}$ of root bags along $\pi_T$ such that $X_i$ is an ancestor of $X_{i+1}$. 
It follows from \Cref{obs:hierarchy} that the distance between the centers of $X_i$ and $X_{i+1}$ is at least $\eps\Delta$. The same conclusion can be reached when $X_{i}$ is a descendant of $X_{i+1}$.
Among the root bags in $X_1, X_2, \ldots, X_\ell$, there is at most one pair of consecutive root bags which are not in ancestor/descendant relation. The path between them is contained in at most $k+1$ clusters by \Cref{obs:hierachy2}.
We conclude that if length of $\pi$ is $\Delta$, then the corresponding walk $\pi_T$ passes through at most $1/\eps + 2$ different root bags, i.e., $\ell \le 1/\eps+2$. Hence, the number of clusters visited by $\pi$ is at most $(k+1)(1/\eps + 3)$. 

The above argument bounds the number of intersections with the clusters built in the $i$-th round by $(k+1)(1/\eps + 3)$. This splits the path into at most $(k+1)(1/\eps + 3)+1$ connected parts that were unclustered by clusters from the $i$th round. Inductively, each such part intersect $J_{i-1}$ clusters. Hence, $J_i$ satisfies recurrence 
\[
J_i = ((k+1)(1/\eps + 3)+1)\cdot J_{i-1} + (k+1)(1/\eps + 3)
\]
with $J_{1} = 1$. The solution is given by $J_{i}=O((i/\eps)^i)$. The lemma follows since the number of rounds is $k+1$ and $J_{k+1}=O((k/\eps)^k)$.
\end{proof}

By combining this lemma with \Cref{obs:path}, we conclude:
\begin{corollary}\label{lem:countClusters} 
For every pair of vertices $u$ and $v$ within distance at most $\Delta$ in $G$, there is an $O((k/\eps)^{k+1})$-hop path in the cluster graph induced by the subset of clusters that has nontrivial intersections with some (approximate) shortest path $\pi(u,v)$ between $u$ and $v$ in $G$.  
\end{corollary}

From \Cref{lem:countClusters} and \Cref{lem:tw-cluster-diameter}, we conclude:
\begin{theorem}
    For any $\e \in (0,1)$, any graph $G$ with treewidth $k$ embeds exactly into a graph $G'$ with treewidth $O(k)$, such that $G'$ has an $(\e, O(k/\e)^{k+1})$-shortcut partition.
\end{theorem}

Combining this theorem with our general framework for constructing tree cover (\Cref{lm:reduction} and \Cref{thm:partition-to-cover}) proves \Cref{thm:treewidth}.


\section{Embedding planar graphs into low-treewidth graphs}
\label{S:treewidth-embedding}

We show that every planar graph embeds, with distortion $+O(\e \Delta)$, into a graph with $O_\e(1)$ treewidth.
The proof is simple and uses two facts of planar graphs: it can be covered by $O_\e(1)$ forests, and it has (a weak version of) a shortcut partition.

\begin{claim}\label{claim:cover-partition}
    Let $G$ be a planar graph with diameter $\Delta$, and let $\e > 0$. Then there is a set $\mathcal{F}$ of $O(\e^{-3})$ forests of rooted trees 
    and a partition $\mathfrak{P}$ of $G$ into components with diameter $\e \Delta$, such that:
    \begin{itemize}
        \item \textnormal{[Low-hop.]} For every pair of vertices $u$ and $v$ in $G$, there is a path in $G$ between $u$ and $v$ that intersects at most $O(\e^{-1})$ clusters.
        \item \textnormal{[Root preservation.]} For every $u,v\in V(G)$, there is some tree $T$ in a forest of $\mathcal{F}$ such that (i) $\dist_{G}(u,v) \leq \dist_{T}(u,v) \leq \dist_G(u,v) + O(\e\Delta)$, and (ii) the shortest path between $u$ and $v$ in $T$ passes through the root of $T$.
        \item \textnormal{[Cluster disjointness.]} Every forest in $\mathcal{F}$ is a spanning forest (i.e., it contains no Steiner vertices or edges). Further, no two trees in any forest in $\mathcal{F}$ intersect the same cluster of $\mathfrak{P}$.
    \end{itemize}
\end{claim}

\noindent This claim is satisfied by choosing $\mathfrak{P}$ to be an $(\e, O(\e^{-1}))$-shortcut partition with $O(1)$ distortion
(which exists by \Cref{thm:gridtree-implies-shortcut}),
and choosing $\mathcal{F}$ to be the associated spanning forest cover of size $O(\e^{-3})$ described in Section~\ref{S:treecover-gridlike} (cf. \Cref{thm:planar-cover}).
In particular, the low-hop property in the claim follows from the low-hop property of shortcut partitions; the root preservation property follows from the proof of \Cref{lem:hierarchy-distortion}; and the cluster disjointness property follows from \Cref{lem:hierarchy-forest}.

\subsection{Adding edges to a low-treewidth graph}
Before describing the construction of the low-treewidth embedding, we prove a general lemma to bound the treewidth of a graph produced by adding certain edges to a low-treewidth graph. We begin by stating a simple property of tree decompositions.
(For definition and basic properties of treewidth and tree decompositions, we refer the readers to the textbook by Diestel~\cite{Diestel} for an introduction.)

\begin{claim}\label{claim:connected-treedecomp}
    Let $G$ be a graph, and let $\mathcal{T}$ be a tree decomposition for $G$. 
    If $H$ is a connected subgraph of $G$, then the set of all bags intersecting nontrivially with $H$ forms a connected subtree in $\mathcal{T}$.
\end{claim}

\begin{proof}
    Let $v$ be an arbitrary vertex in $H$. 
    Let $V_i$ denote the set of all vertices that are within hop-distance $i$ of $v$ in $H$, and let \EMPH{$B_i$} be the set of bags in $\mathcal{T}$ that intersect nontrivially with $V_i$. Because $H$ is connected, there is some $i$ where all bags in $B_i$ intersect nontrivially with $H$. We prove, by induction on $i$, that $B_i$ is a connected subtree in $\mathcal{T}$. In the base case where $i = 0$, $V_0 = \set{v}$ and so the claim follows from the definition of tree decomposition. For the inductive step, assume that the claim holds for every $B_j$ with $j < i$. 
    For every $u \in V_i \setminus V_{i-1}$, let \EMPH{$B^+_u$} denote the set of bags in $\mathcal{T}$ that contain $u$. Notice that every $u \in V_i \setminus V_{i-1}$ is adjacent to some vertex in $V_{i-1}$, so there is some bag (containing both $u$ and a vertex in $V_{i-1}$ adjacent to $u$) that is in both $B^+_u$ and $B_{i-1}$. As $B^+_u$ and $B_{i-1}$ are both connected subsets of $\mathcal{T}$, we conclude that $B_i = \bigcup_{u \in V_i \setminus V_{i-1}} B^+_u \cup B_{i-1}$ is a connected subset of $\mathcal{T}$. 
\end{proof}

\begin{lemma}\label{lem:extend-forest}
    Let $G$ be a graph, and let $\mathcal{T}$ be a tree decomposition for $G$ with treewidth $w$. Let $\mathcal{F} = \set{F_1, \ldots, F_k}$ be a set of spanning forests (that is, each forest is a subgraph of $G$ containing rooted trees). Then there is a tree decomposition $\mathcal{T}'$ for $G$ such that (i) $\mathcal{T}'$ has width $O(wk)$, and (ii) for every tree $T$ in some forest of $\mathcal{F}$ and for every vertex $v$ in $T$, there is some bag in $\mathcal{T}'$ that contains both $v$ and $\operatorname{root}(T)$.
\end{lemma}
\begin{proof}
    Let $\mathcal{T}' \gets \mathcal{T}$. 
    For every $v \in V(G)$, let $\EMPH{$R_v$} \coloneqq \set{\big. \operatorname{root}(T) \mid \text{$T$ is a tree in $\mathcal{F}$ where $v \in V(T)$}}$, and add the vertices $R_v$ to each bag in $\mathcal{T}'$ that contains $v$.

    \smallskip \noindent \textbf{Width of $\mathcal{T}'$.} Notice that for every $v$, at most one tree per forest of $\mathcal{F}$ contains $v$, so $|R_v| \leq k$. Each bag in $\mathcal{T}$ has at most $w + 1$ vertices. The corresponding bag in $\mathcal{T}'$ contains at most $(k+1)(w+1)$ vertices: it contains each of the $w+1$ original vertices $v$, plus the $w+1$ sets $R_v$ that are each of size at most $k$.

    \smallskip \noindent \textbf{Connectivity in $\mathcal{T}'$.} 
    We need to show that for every vertex $v$ in $G$, the bags in $\mathcal{T}'$ containing $v$ form a connected subtree. 
    For an arbitrary $v$ in $G$, let \EMPH{$H_v$} denote the union of all trees in (some forest of) $\mathcal{F}$ that contain $v$. Notice that $H_v$ is a connected subgraph of $G$, as it is the union of connected subgraphs that share a common vertex. \Cref{claim:connected-treedecomp} implies that the set of bags \EMPH{$B_v$} of $\mathcal{T}$ that intersects nontrivially with $V(H_v)$ form a connected subset in $\mathcal{T}$. 
    However, the set of bags in $\mathcal{T}'$ containing $v$ is precisely $B_v$.  
\end{proof}

\subsection{Embedding a planar graph}
Let $G$ be a planar graph with diameter $\Delta$, and let $\e > 0$. 
Let $\mathcal{F}$ be the forest cover and let $\mathfrak{P}$ be the partition guaranteed by \Cref{claim:cover-partition}. 
We construct a graph \EMPH{$G'$} by first contracting each cluster in $\mathfrak{P}$. For each cluster, choose an arbitrary vertex in that cluster to be the \EMPH{center vertex} of the cluster. Replace each supernode in $G'$ with a star: the center vertex of the cluster is the center of the star, and the non-center vertices in the cluster are the other vertices. Edges between supernodes are replaced with edges between the corresponding center vertices. Assign each edge $(u,v) \in E(G')$ weight equal to $\delta_G(u,v)$.

\begin{lemma}\label{lem:contractg-treewidth}
    The graph $G'$ has treewidth $O(\e^{-1})$.
\end{lemma}
\begin{proof}
    The low-hop property of \Cref{claim:cover-partition} guarantees that there is a path of hop-length $O(\e^{-1})$ between every pair of vertices in $V(G')$. Further, $G'$ is planar: Planar graphs are minor-closed, and replacing supernodes with stars does not affect planarity. Planar graphs have treewidth asymptotically upper-bounded by hop diameter, which proves the claim.
\end{proof}

We now augment $G'$ with extra edges into \EMPH{$\hat{G}$} to reduce the distortion. 
Let $\hat{G} \gets G'$. For every tree $T$ in $\mathcal{F}$, and for every vertex $v$ in $T$, add an edge to $\hat{G}$ between $\operatorname{root}(T)$ and $v$. Assign the edge weight to be $\dist_G(\operatorname{root}(T), v)$.

\begin{lemma}\label{lem:hatg-distortion}
    For every pair of vertices $u$ and $v$ in $\hat{G}$, we have $\dist_G(u,v) \leq \dist_{\hat{G}}(u,v) \leq \dist_{G}(u,v) + O(\e\Delta)$.
\end{lemma}
\begin{proof}
    Let $u$ and $v$ be vertices in $\hat{G}$. By the root preservation property in \Cref{claim:cover-partition}, there is some tree $T$ in $\mathcal{F}$ such that $\dist_T(u, \operatorname{root}(T)) + \dist_T(\operatorname{root}(T), v) \leq \dist_G(u, v) + O(\e\Delta)$. The construction of $\hat{G}$ guarantees that there is an edge in $\hat{G}$ from $u$ to $\operatorname{root}(T)$ and an edge from $\operatorname{root}(T)$ to $v$, weighted according to their distances in $G$. This proves the upper-bound for $\dist_{\hat{G}}(u,v)$. The lower bound follow from the fact that every edge $(u,v)$ in $\hat{G}$ has weight $\dist_G(u,v)$.
\end{proof}

To prove that $\hat{G}$ has low treewidth, we want to use \Cref{lem:extend-forest} applied to $G'$ and the forest cover for~$\mathcal{F}$. 
To do this, we need first to state a lemma that lets us translate between spanning forests in $G$ and spanning forests in $G'$.
\begin{lemma}\label{lem:tree-translation}
    For every $F$ in $\mathcal{F}$, there is a corresponding spanning forest $F'$ of rooted trees (each a subgraph of $G'$), such that the following holds: For every tree $T$ in $F$, there is a tree $T'$ in $F'$ such that $V(T) \subseteq V(T')$ and $\operatorname{root}(T) = \operatorname{root}(T')$.
\end{lemma}
\begin{proof}
    For each tree $T$ in $F$, 
    Let $V_T \subseteq V(G)$ be the union of cluster vertices in $\mathfrak{P}$ that intersect nontrivially with $T$.
    Let \EMPH{$V'_T$} is the corresponding subset of $V(G')$ by identifying vertices in $V_T$ that came from the same cluster in $\mathfrak{P}$ into a single vertex.
    The sets  $\{V'_T\}_{T\in F}$ are pairwise disjoint, because of the cluster disjointness property of $\mathfrak{P}$ in \Cref{claim:cover-partition}. Further, because each cluster in $\mathfrak{P}$ is connected and because $T$ is a connected subgraph of $G$, each set $V'_T$ induces a connected subgraph of $G'$. Let $T'$ be a spanning tree in $G'$ rooted at $\operatorname{root}(T)$ with vertex set $V'_T$. The collection of all trees $T'$ is pairwise vertex-disjoint, and thus forms a forest on $G'$.
\end{proof}

\begin{lemma}\label{lem:hatg-treewidth}
    The graph $\hat{G}$ has treewidth $O(\e^{-4})$.
\end{lemma}
\begin{proof}
    For every forest $F$ in $\mathcal{F}$, apply \Cref{lem:tree-translation} to find a corresponding forest $F'$ in $G'$. The resulting set $\mathcal{F}'$ contains $O(\e^{-3})$ forests. By construction of $\hat{G}$, every edge $(u,v) \in E(\hat{G}) \setminus E(G')$ is induced by some tree $T$ where (without loss of generality) $u = \operatorname{root}(T)$ and $v \in V(T)$. By \Cref{lem:tree-translation}, there is a corresponding tree $T'$ in $\mathcal{F}'$ such that $u = \operatorname{root}(T')$ and $v \in V(T')$. Thus, applying \Cref{lem:extend-forest} to $G'$ (which, by \Cref{lem:contractg-treewidth}, has treewidth $O(\e^{-1})$) and $\mathcal{F}'$ yields a valid tree decomposition for $\hat{G}$ with width $O(|\mathcal{F}'|\cdot \tw(G')) =  O(\e^{-4})$.
\end{proof}

 We now observe that \Cref{thm:FKS-improved} follows from \Cref{lem:hatg-distortion} and \Cref{lem:hatg-treewidth}.


\section{Reduction from additive to multiplicative tree cover}
\label{S:add-to-mul}

In this section, we prove \Cref{lm:reduction}, which we restate below.
Recall that a tree cover $\mathcal{T}$ is \EMPH{$\Delta$-bounded} if every tree of $\mathcal{T}$ has diameter at most $\Delta$.

\MultToAdd*

We introduce the notion of a \emph{family of pairwise hierarchical partitions} and base our reduction on this family. To formally define this notion, we first define a hierarchical partition.

\begin{definition}[Hierarchical Partition]\label{label:hier-part}Let $\mu > 1$ be a parameter. 
Let $(X,\delta_X)$ be a metric space with minimum distance $1$ and maximum distance at most $\Phi$. A \EMPH{$\mu$-hierarchical partition} for $(X,\delta_X)$, denoted by $\mathbb{P} = \{\mathcal{P}_0,\mathcal{P}_1,\ldots,\mathcal{P}_{i_{\max}}\}$ where $i_{\max} = O(\log_{\mu}(\Phi)) $, is a set of partitions  such that:
	\begin{enumerate}
		\item  $\mathcal{P}_0$ contains singletons only and $\mathcal{P}_{i_{\max}}$ contains a single set $X$.
		\item \label{it:partition}Each $\mathcal{P}_i$ is a partition of $X$ into clusters of diameter at most $\mu^{i}$ for every integer $i \in [0,i_{\max}]$. We call \EMPH{$\mathcal{P}_i$} a \EMPH{partition at scale $i$} of $\mathbb{P}$. 		
  \item \label{it:nested} Each set $S$ in $\mathcal{P}_i$ for $i\geq 1$ is the union of some sets in $\mathcal{P}_{i-1}$.
	\end{enumerate}
\end{definition}
\Cref{it:nested} implies that the partitions are nested and hence form a hierarchy.

\begin{definition}[Hierarchical Pairwise Partition Family (HPPF)]
\label{def:fam-hier-part}
Let $\sigma,\mu,\rho\geq 1$ be parameters. Let $(X,\delta_X)$ be a metric space with minimum distance $1$ and maximum distance $\Phi$. 
A \EMPH{$(\sigma,\mu,\rho)$-hierarchical pairwise partition family}, which we abbreviate as \EMPH{$(\sigma,\mu,\rho)$-HPPF}, is a family of $\mu$-hierarchical partitions of size $\sigma$, denoted by $\mathfrak{P} = \{\mathbb{P}_1,\ldots, \mathbb{P}_{\sigma}\}$, of $(X,\delta_X)$ such that:	
\begin{enumerate}
    \item Each $\mathbb{P}_j$, $j\in [\sigma]$, is a  $\mu$-hierarchical partition of $(X,\delta_X)$. 
    \item \label{it:dif-point-hierachy} 
    For every two different points $x$ and $y$ in $X$, there exists a  hierarchy  $\mathbb{P}_j$ in $\mathfrak{P}$ and a partition $\mathcal{P}_i$ at scale $i$ of $\mathbb{P}_j$ such that both $x$ and $y$ belong to the same cluster and $\delta_X(x,y)\geq \mu^i/\rho$. 
\end{enumerate}
\end{definition}

\Cref{it:dif-point-hierachy} in \Cref{def:fam-hier-part} is called the \EMPH{pairwise property}. Recall that each cluster in $\mathcal{P}_i$ has diameter at most $\mu^i$. The pairwise property posits that for every pair of points, there is a cluster in some partition at scale $i$ containing the pair, whose diameter is roughly the same as the distance of the pair up to a factor of $\rho$. The pairwise property is crucial in our reduction.  

The HPPF is a variant of the \emph{hierarchical partition family} (HPF) introduced earlier by~\cite{BFN19Ramsey,KLMN04} without the pairwise property. Instead, the HPF in these works has a \emph{padded property}: every ball $B_X(x,\mu^i/\rho)$ is wholly contained in some cluster of some partition at scale $i$.

\begin{lemma}\label{lm:HPPF-minor-free} Let $\eps\in (0,1)$ be a parameter. Any $K_r$-minor-free metric $(X,\delta_X)$ admits  $(\sigma,\mu,\rho)$-HPPF for $\sigma = O(3^r\log(1/\eps)/\log{r}), \mu \geq 1/\eps, \rho = O(r^2)$. 
\end{lemma} 

We defer the proof of \Cref{lm:HPPF-minor-free} to \Cref{subsec:HPPF}. Equipped with a HPPF as in \Cref{lm:HPPF-minor-free}, we now show the reduction as claimed in \Cref{lm:reduction}.

\subsection{The reduction: proof of \Cref{lm:reduction}}

In this section, we prove \Cref{lm:reduction}, assuming that \Cref{lm:HPPF-minor-free} holds. Let $\mathfrak{P}$ be a $(\sigma,\mu,\rho)$-HPPF for $\sigma = O(3^r\log(1/\eps)/\log{r}), \mu = 1/\eps, \rho = O(r^2)$ as in \Cref{lm:HPPF-minor-free}. For each hierarchy of partitions $\mathbb{P}\in \mathfrak{P}$, we will construct a tree cover $\mathcal{T}_{\mathbb{P}}$. The final tree cover will be $\mathcal{T} \coloneqq \cup_{\mathbb{P}\in \mathfrak{P}}\mathcal{T}_{\mathbb{P}}$.
	
For each scale $i$, we construct a \EMPH{net $N_i$} as a subset of $X$ inductively as follows. 
$N_{i_{\max}}$ 	contains a single point chosen arbitrarily from $X$. Suppose that we are given $N_{i+1}$. For every set $S$ in $\mathcal{P}_i$ (the partition at scale $i$), if $S\cap N_{i+1} = \varnothing$ then we chose a point in $S$ arbitrarily and add it to $N_i$. 
(Initially, $N_i$ is empty.) Otherwise, we add $S\cap N_{i+1}$ to $N_i$. The sets $N_i$'s satisfy the following claim.

\begin{claim}
\label{clm:Ni-prop}
One has
$N_{i_{\max}}\subseteq \ldots\subseteq N_{0} = X$. Furthermore for every scale $i$ and every set $S\in \mathcal{P}_i$, $|S\cap N_i| = 1$.
\end{claim} 
\begin{proof}
	The first claim follows directly from the fact that $\mathcal{P}_i$ is a partition of $X$ and whenever  $S\cap N_{i+1}\not=\varnothing$ we add all points of $S\cap N_{i+1}$ to $N_i$.  We prove the second claim by induction: the base case is $N_{i_{\max}}$ that has a single point and hence the claim holds. For the inductive case, if $S\cap N_{i+1} \not= \varnothing$, then $|S\cap N_{i+1}| = 1$. This is because $\mathcal{P}$ is a hierarchy and hence there exists  a   superset of $S$ in $\mathcal{P}_{i+1}$, say $S'$, such that $S\subseteq S'$ and thus has $|S'\cap N_{i+1}| = 1$ by induction. It follows that $|S\cap N_i| = 1$ by construction. Otherwise,  $S\cap N_{i+1} = \varnothing$, we choose a single point of $S$ to add to $N_i$ by construction and hence $|S\cap N_i| = 1$. 
\end{proof}

Let $\EMPH{$r_i$} \coloneqq \mu^i$. \Cref{clm:Ni-prop} shows that $N_{i}$ is a \emph{$r_{i}$-cover} of $N_{i-1}$ as each point $x\in N_{i-1}$ has a point $y\in N_{i}$ such that $\delta_X(x,y) \leq r_i$ since the diameter of sets in $\mathcal{P}_{i}$ is at most $r_i$. For every point $x\in S\in \mathbb{P}_i$, we call the point $x_i\in N_i\cap S$ the \EMPH{ancestor at scale $i$} of $x$. Observe that:

\begin{observation}\label{obs:cover-dist} For any point $x\in X$, $\delta_X(x,x_i)\leq r_i$ where $x_i$ is the ancestor at scale $i$ of $x$.
\end{observation}
\begin{proof}
This follows from the fact that $x, x_i \in S$ for some $S \in \mathbb{P}_i$, and $S$ has diameter at most $r_i$.
\end{proof}

In the following construction, we regard points in different sets $N_i$ as \emph{different} and write \EMPH{$(x,i)$} as a copy of $x$ if $x\in N_i$, and the tree cover we construct is for all points in $N_0\cup N_1\cup \ldots\cup N_{i_{\max}}$. 
For each tree in the final tree cover $\mathcal{T}_{\mathbb{P}}$, we keep only one copy per point, which is the copy in $N_0$ of the point.

For each scale $i\in [i_{\max}]$, we construct a set of $\kappa$ forests $\EMPH{$\mathcal{F}^i$} = \{F^i_{1},\ldots , F^{i}_{\kappa}\}$ rooted at points in $N_i$ as follows. (The value of $\kappa$ will be set later.) 
For each set $S$ in $\mathcal{P}_i$, we construct a tree cover, denoted by \EMPH{$\mathcal{T}^i_S$}, with additive distortion $+\eps r_i/\rho$ for $S\cap N_{i-1}$. 
Note that $r_i$ is the diameter of $S$. 
By \Cref{lm:reduction}, $|\mathcal{T}^i_S| \leq \tau(|S\cap N_{i-1}|,\eps/\rho) \leq \tau(n,O(\eps/r^2))$. 
(Recall that $\rho = O(r^2)$ in \Cref{lm:HPPF-minor-free}.) 
We choose $\kappa = \tau(n,O(\eps/r^2))$ so that $|\mathcal{T}^i_S|\leq \kappa$ for every $S$ in $\mathcal{P}_i$. By duplication, we assume that $\mathcal{T}^i_S$ contains exactly $\kappa$ trees, denoted by \EMPH{$T_{1,S},\ldots, T_{\kappa,S}$}.  
We root each tree $T_{t,S}$ in $\mathcal{T}^i_S$ at the copy $(x,i)$ of the (single) point $x$ in $S\cap N_{i}$; by \Cref{clm:Ni-prop}, there is only one such point.
Next, for each $t\in [\kappa]$, let $\EMPH{$F^{i}_t$} \coloneqq \{T_{t,S}: S \in \mathcal{P}_i\}$ to be the $t$-th forest of $\mathcal{F}^i$. When $i = 0$, we define $F^i_t$ to be the forest of singletons.

Finally, we construct the $t$-th tree for $t\in [\kappa]$, denoted by \EMPH{$T_t$}, in the cover $\mathcal{T}_{\mathbb{P}}$ by connecting the forests $F^{i}_{t}$ at different scales $i$.
Specifically, for every scale $i$ in $[i_{\max}]$, for each point $x$ in $N_{i-1}$, add an edge of length 0 between $(x, i-1)$ and $(x, i)$. This operation effectively connects the connected components of $F^{i}_t$ and the components of $F^{i-1}_t$ for every $i\in [i_{\max}]$.

\begin{figure}[ht!]
    \centering
    \includegraphics[width=\textwidth]{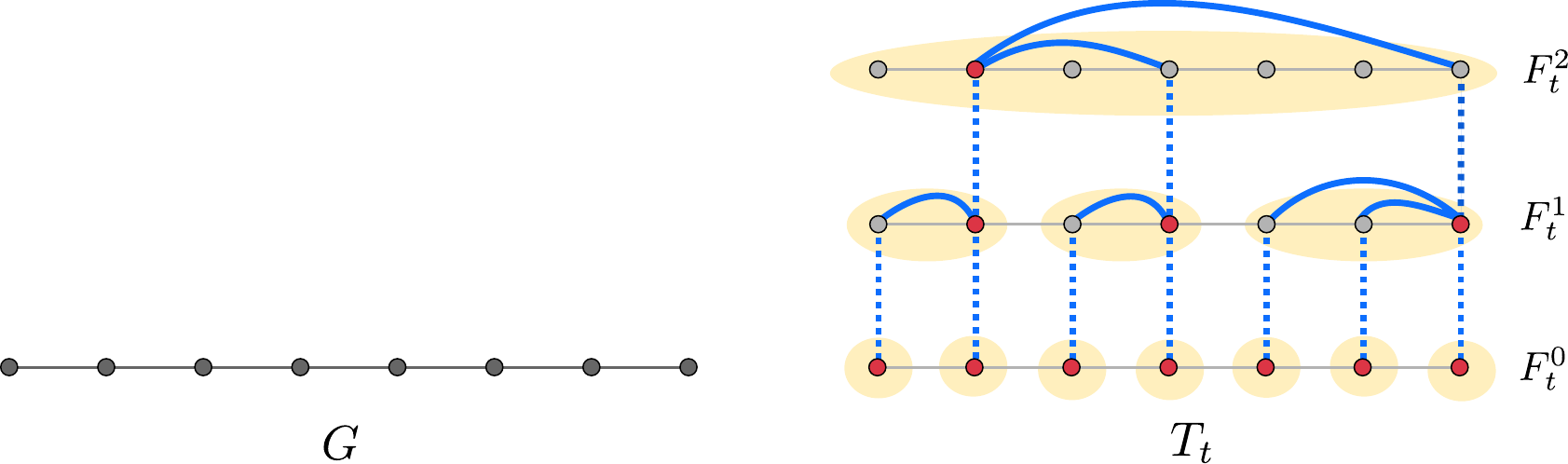}
    \caption{A graph $G$ and a tree $T_t$. The tree $T_t$ consists of forests $F_t^2$, $F_t^1$, and $F_t^0$, connected (by blue dashed lines) to form $T_t$. At each scale $i$:  there is a partition of  $\mathcal{P}_i$ (highlighted in yellow), a set of net points $N_i$ (drawn in red), and a forest $F_t^i$ (drawn with blue lines).}
    \label{fig:nettree}
\end{figure}

\begin{claim}
\label{clm:Xt-tree} 
$T_t$ is a tree.
\end{claim}
\begin{proof} Let $\hat{F}^0_{t} \coloneqq F^0_t$ which is a forest of singletons. Let $\hat{F}^{i}_t$ be the forest obtained by connecting $F_t^{i}$ and $F_t^{i-1}$ in the construction algorithm via edges of length 0 between $(x,i)$ and $(x,i-1)$ for every $(x,i)\in N_i$. Then $\hat{F}^{i}_t$
is a forest containing trees rooted at points of $N_{i}$.
As $N_{i_{\max}}$ consists of a single point, the forest $\hat{F}^{i_{\max}}_t$ (which is also $T_t$) is a tree.
\end{proof}

The following claim is crucial in bounding the distortion of the tree cover.

\begin{claim}\label{clm:Xt-dist} Let $x_0$ be any point in $N_0$ and $x_i\in N_i$ be the ancestor at scale $i$ of $x_0$ for any $i\ge 1$. Then $\delta_G(x_i,x_0)\leq d_{T_t}(x_i,x_0) = O(r_i)$ when $\eps < 1$.
\end{claim}
\begin{proof} The lower bound follows from the fact that each tree in the forest $F^i_t$ is a dominating tree for $N_{i-1}\cap S$. We now focus on proving the upper bound. Let $c_0$ be the constant such that every tree cover of additive distortion $+\eps \Delta$ is $(c_0\cdot D)$-bounded following the assumption of \Cref{lm:reduction}. Since every $S\in \mathcal{P}_i$ has diameter at most $r_i$, it follows that $d_{F^i_t}(x_{i-1},x_i)\leq c_0 r_i$. Thus, $d_{T_t}(x_i,x_0)\leq c_0\sum_{j=0}^i r_i \leq c_0r_i(1 + \eps + \eps^2 + \ldots)\leq O(r_i)$ when $\eps <1$.
\end{proof}

The following claim concludes the proof of \Cref{lm:reduction}.

\begin{claim}\label{clm:reduction} Let $\mathcal{T} = \cup_{\mathbb{P}\in \mathfrak{P}}\mathcal{T}_{\mathbb{P}}$. Then $\mathcal{T}$ is a tree cover with multiplicative distortion $(1+O(r^2\eps))$ of size $O(\tau(n,O(\eps/r^2)) \cdot 3^r\log(1/\eps)/\log{r})$.
\end{claim}

\begin{proof} Since $|\mathfrak{P}| = \sigma = O(3^r\log(1/\eps)/\log{r})$ and $|\mathcal{T}_{\mathbb{P}}|\leq \kappa =  \tau(n,O(\eps/r^2))$, $|\mathcal{T}| \leq O(\tau(n,O(\eps/r^2)) \cdot 3^r\log(1/\eps)/\log{r})$ as claimed. It remains to bound the distortion. We observe by the construction that every tree in $\mathcal{T}$ is a dominating tree since every tree in $\mathcal{T}_{\mathbb{P}}$ is dominating. 
	
Let $x,y$ be any two points in $X$. Let $\mathbb{P}\in \mathfrak{P}$ be the hierarchy and $\mathcal{P}_i$ be a partition of $X$ at scale $i$ of $\mathbb{P}$ such that $\{x,y\}\subseteq S$ for some set $S\in \mathcal{P}_i$ and  $\delta_X(x,y)\geq r_i/\rho$. $\mathbb{P}$ and $i$ exist by the definition of HPPF. Let $\hat{x}, \hat{y}$ be the ancestors of $x$ and $y$, respectively, at scale $i-1$. As $\delta_X(x,y)\geq r_i/\rho$, we have:
\begin{equation}\label{eq:dxy-ri}
	r_{i-1}\leq \eps r_i \leq \eps \rho \delta_X(x,y).
\end{equation}
Since $x,y$ are both in $S$, $\hat{x}$ and $\hat{y}$ are both in $S$ by the definition of $N_i$'s.  Since the tree cover $\mathcal{T}^i_S$ for $N_{i-1}\cap S$ in the construction of $\mathcal{F}^i$ has additive distortion $\eps r_i/\rho$, there is a forest $F^i_t\in \mathcal{F}^i$ for some $t\in [\kappa]$ such that $d_{F^i_t}(\hat{x},\hat{y})\leq \delta_X(\hat{x},\hat{y}) + \eps r_i/\rho$. It follows that:
\begin{equation}\label{eq:dhatxy-dist}
	\begin{split}
			d_{T_t}(\hat{x},\hat{y})&\leq \delta_X(\hat{x},\hat{y}) + \eps r_i/\rho\\
			&\leq \delta_X(x,y) + O(r_{i-1}) + \eps r_i/\rho\qquad \mbox{(by triangle inequality and \Cref{obs:cover-dist})}\\
			&= \delta_X(x,y)  + O(\eps\rho \delta_X(x,y)) + \eps \delta_X(x,y) \qquad \mbox{(by \Cref{eq:dxy-ri})}\\
			&= \delta_X(x,y) + O(\rho)\eps \delta_X(x,y)
	\end{split}
\end{equation}
Furthermore, by \Cref{clm:Xt-dist} and the triangle inequality, we have:
\begin{equation}
	\begin{split}
		\delta_{T_t}(x,y) &\leq 	d_{T_t}(\hat{x},\hat{y}) + O(r_{i-1}) \\
		&\leq \delta_X(x,y) + O(\rho)\eps \delta_X(x,y) + O(r_{i-1})  \qquad\mbox{(by \Cref{eq:dhatxy-dist})}\\
		&\leq \delta_X(x,y) + O(\rho)\eps \delta_X(x,y) + O(\rho)\eps \delta_X(x,y)
		\qquad\mbox{(by \Cref{eq:dxy-ri})}\\
		&= (1+O(r^2 \eps))\delta_X(x,y) 
	\end{split}
\end{equation}
since $\rho = O(r^2)$. That is, there exists a tree in $\mathcal{T}_{\mathbb{P}}$ (which is $T_t$), and hence a tree in $\mathcal{T}$, such that the distance between $x$ and $y$ is preserved in the tree up to a factor of $1+O(r^2\eps)$. The claim now follows.
\end{proof}

\subsection{HPPF construction: proof of \Cref{lm:HPPF-minor-free}}\label{subsec:HPPF}

We base our construction of a HPPF  on the  \emph{hierarchical partition family} (HPF) formally define below. The FPF was formally introduced in~\cite{BFN19Ramsey} though its ideas appeared earlier~\cite{KLMN04}.

\begin{definition}[Hierarchical Partition Family (HPF)]\label{def:fam-HPF} Let $\sigma,\mu,\rho\geq 1$ be parameters. Let $(X,\delta_X)$ be a metric space with minimum distance $1$ and maximum distance $\Phi$.  A \EMPH{$(\sigma,\mu,\rho)$-hierarchical partition family} (\EMPH{$(\sigma,\mu,\rho)$-HPF} for short) is a family of $\mu$-hierarchical partitions of size $\sigma$, denoted by $\mathfrak{P} = \{\mathbb{P}_1,\ldots, \mathbb{P}_{\sigma}\}$ of $(X,\delta_X)$, such that:
	
	\begin{enumerate}
		\item Each $\mathbb{P}_j$, $j\in [\sigma]$, is a  $\mu$-hierarchical partitions of $(X,\delta_X)$. 
		\item For every point $x\in X$, there exists a  hierarchy  $\mathbb{P}_j\in \mathfrak{P}$ and a partition $\mathcal{P}_i$ at scale $i$ of $\mathbb{P}_j$ such $B_X(x,\mu^i/\rho)\subseteq S$ for some cluster $S\in \mathcal{P}_i$. 
	\end{enumerate}
\end{definition}

The following lemma was shown in~\cite{KLMN04} as noted by \cite{BFN19Ramsey}.

\begin{lemma}[\cite{KLMN04}]\label{lm:KLMN} Any $K_r$-minor-free metric admits a  $(\sigma,\mu,\rho)$-HPF with $\sigma = O(3^r)$, $\mu = O(r^2)$ and $\rho = O(r^2)$.
\end{lemma}

Next, we show how to construct a HPPF from a HPF. The following lemma and \Cref{lm:KLMN} implies \Cref{lm:HPPF-minor-free} as all parameters $\hat{\sigma},\hat{\mu},\hat{\rho}$ are $O(1)$ for planar graphs.

\begin{lemma} Let $(X,\delta_X)$ be a metric space admitting a $(\hat{\sigma},\hat{\mu},\hat{\rho})$-HPF\/. Then for any $\eps\in (0,1)$, $(X,\delta_X)$ admits a $(\sigma,\mu,\rho)$-HPPF with $\sigma = O(\hat{\sigma}\log(1/\eps)/\log(\hat{\mu}))$, $\mu \geq 1/\eps$, and $\rho = \hat{\rho}$. 
\end{lemma}
\begin{proof}
	Let $\hat{\mathfrak{P}} = \{\hat{\mathbb{P}}_1,\ldots, \hat{\mathbb{P}}_{\hat{\sigma}}\}$ be a $(\hat{\sigma},\hat{\mu},\hat{\rho})$-HPF by the assumption of the lemma. 
    For every hierarchy $\hat{\mathbb{P}}_a \in \hat{\mathfrak{P}}$ for $a\in [\hat{\sigma}]$, we ``partition'' it into $\kappa \coloneqq \ceil{\log_{\hat{\mu}}(1/\eps)}$ hierarchies $\set{\mathbb{P}_{a,0},\mathbb{P}_{a,1},\ldots, \mathbb{P}_{a,\kappa-1}}$ as follows: 
    each hierarchy $\mathbb{P}_{a,t}$ for $t\in [\kappa]$ includes all partitions $\mathcal{P}_i$ in all scales $i$'s of $\hat{\mathbb{P}}_a$ such that $i\equiv t \pmod \kappa$. We then define:
	\begin{equation}\label{eq:HPPF}
		\mathfrak{P} \coloneqq \{\mathbb{P}_{a,t}: a\in [\hat{\sigma}], t\in [\kappa]\}.
	\end{equation}
Clearly, $\sigma = |\mathfrak{P}| =\hat{ \sigma}\cdot\kappa  = O(\hat{\sigma}\log(1/\eps)/\log(\hat{\mu}))$ as claimed. Furthermore, for every hierarchy $\mathbb{P}\in \mathfrak{P}$, the ratio of the diameter of scale $i+1$ to the diameter of scale $i$ is $\hat{\mu}^{\kappa} = \hat{\mu}^{\ceil{\log_{\hat{\mu}}(1/\eps)}}\geq 1/\eps$; thus, $\mu\geq 1/\eps$. 

Finally, let $(x,y)$ be any two different points in $X$. 
Let $i\geq 0$ be such that $\hat{\mu}^{i-1}/\rho \leq \delta_X(x,y) < \hat{\mu}^{i}/\rho$; such $i$ exists since $\rho \leq 1$, $\delta_X(x,y)\geq 1, \mu\geq 1$. 
By Item 2 in \Cref{def:fam-HPF}, there exist a hierarchy $\hat{\mathbb{P}}_a \in \hat{\mathfrak{P}}$, a partition $\mathcal{P}_i \in \hat{\mathbb{P}}_a$, and a cluster $S\in \mathcal{P}_i$ such that $B_X(x,\mu^{i}/\rho)\subseteq S$.  Since hierarchies $\set{\mathbb{P}_{a,0},\mathbb{P}_{a,1},\ldots, \mathbb{P}_{a,\kappa-1}}$ partition $\hat{\mathbb{P}}_a$, there exists a hierarchy $\mathbb{P}_{a,t}$ for some $t\in [\kappa]$ such that $S$ belongs to the partition at some scale of $\mathbb{P}_{a,t}$. Furthermore, both $x$ and $y$ are in $S$ since $B_X(x,\mu^{i}/\rho)\subseteq S$, as desired.
\end{proof}


\section{Algorithmic applications}
\label{S:applications}

In this section, we discuss the algorithmic applications of our tree cover theorem (\Cref{thm:main}).

\paragraph{Distance oracles.~}  Let $\mathcal{T}$ be a tree cover of size $O(\eps^{-3}\log(1/\eps))$. In the preprocessing stage, we construct a distance oracle for each tree $T$ in $\mathcal{T}$ using the LCA data structure, following a standard construction~\cite{LW21,FGNW17}, as follows. Root $T$ at a vertex $r$. For each vertex $v$ in $T$, we store the distance $d_T(v,r)$ at $v$. We then construct the LCA data structure for $T$, denoted by $\lca_T$. To query the distance between two vertices $u$ and $v$ in $T$, first we query $\lca_T$ to get the LCA of $u$ and $v$, denoted by $w$. We then return $d_T(u,r) + d_T(v,r) - 2d_T(w,r)$ as the distance between $u$ and $v$ in $G$. 

 The distance oracle for $G$ will contain the distance oracle for each tree $T$ in the tree cover. The total space over all trees is then $\sum_{T\in \mathcal{T}}(S_{LCA}(|V(T)|) + O(|V(T)|)) = O(S_{LCA}(n)\cdot \eps^{-3}\log(1/\eps))$ as $|V(T)| = O(n)$ for every $T$ in $\mathcal{T}$. To query the distance between $u$ and $v$, we simply go over each tree $T$ in $\mathcal{T}$, query the distance $d_T(u,v)$ and finally return $\min_{T\in \mathcal{T}}d_T(u,v)$ as the distance. The query time is $O(Q_{LCA}(n)\cdot \eps^{-3}\log(1/\eps))$. The returned distance is a $(1+\eps)$-approximation of $\delta_G(u,v)$ since the distortion of $\mathcal{T}$ is $(1+\eps)$. Our \Cref{thm:app-distance-oracle} now follows.

\paragraph{Emulator of linear size.~} Given a planar graph $G$ and a set $S$ of $k$ terminals in $G$. We say that a graph $H$ is a \emph{$(1+\eps)$-emulator} if $S\subseteq V(H)$ and for every two terminals $t_1,t_2 \in S$, 
\[
\delta_G(t_1,t_2)\leq d_H(t_1,t_2) \leq (1+\eps)\cdot \delta_G(t_1,t_2). 
\]
The size of the emulator is the number of edges. 
Cheung, Goranci, and Henzinger constructed a \emph{planar} $(1+\eps)$-emulator of almost quadratic size $\tilde{O}(k^2/\eps^2)$. 
Recently, Chang, Krauthgamer, and Tan~\cite{CKT22,CKT22b} obtained a planar $(1+\eps)$-emulator of \emph{almost linear} size $k \poly(\log k, 1/\eps)$,
which breaks below the $\Omega(k^2)$ lower bound when no distortion is allowed.
 
 Using our tree cover in \Cref{thm:main}, we can construct a $(1+\eps)$-emulator of \emph{linear size} as follows. 
 For each tree $T$ in $\mathcal{T}$, we prune $T$ so that it only has $O(k)$ vertices by (i) repeatedly removing non-terminal leaves until every leaf is a terminal in $S$ and (ii) contracting non-terminal vertices of degree 2 and reweighting the new edge appropriately. This pruning does not change the distances in $T$ between terminals. Finally, we simply union all the tree $T$ in the cover $\mathcal{T}$. The same copies of a terminal in different trees will become a single terminal in the emulator. The size of our emulator is $O(k\cdot \eps^{-3}\log(1/\eps))$ since each tree has $O(k)$ vertices. We note that, unlike prior work, our emulator may not be planar.

 \paragraph{Low-hop emulators of planar metrics.} 
 We have just shown that any point set in a planar metric has a linear size $(1+\eps)$-emulator. 
 In some applications, we would want our emulators to have the low-hop property: every vertex $u$ can reach a vertex $v$ by a shortest path having few edges in the emulator.   
 Our tree cover theorem provides a simple way to construct such an emulator, following the line of~\cite{KLMS22}. Specifically, for each tree in the tree cover, construct a low-hop spanner for the tree using the result by Solomon~\cite{Solomon2013}, and take the union of all the spanners. The end result is an emulator with $k$ hops for any integer $k\geq 1$ and size $O(n\alpha_k(n)\cdot \eps^{-3}\log(1/\eps))$ where $\alpha_k(n)$ is the function $\alpha_{k}(n)$ is the inverse of an Ackermann-like function at the $\lfloor k/2 \rfloor$-th level of the primitive recursive hierarchy, where 
$\alpha_0(n) = \lceil n/2\rceil$,
$\alpha_1(n) = \lceil \sqrt{n} \rceil$,
$\alpha_2(n)= \lceil\log{n}\rceil$,
$\alpha_3(n)= \lceil\log\log{n}\rceil$,
$\alpha_4(n)= \log^*{n}$,
$\alpha_5(n)= \lfloor \frac{1}{2}\log^*{n} \rfloor$, etc.

\paragraph{Distance labeling schemes.}
A distance labeling scheme is an assignment of labels (binary strings) to all vertices of a graph,
so that the distance between any two nodes can be computed solely from their labels and the size of labels is as small as possible. 
A major open problem is to determine the complexity of distance labeling in an $n$-vertex unweighted and undirected planar graphs. 
An upper bound of $O(\sqrt{n} \log n)$ bits, complemented with a lower bound of $\Omega(n^{1/3})$, were achieved in the pioneering work of Gavoille et al.\ \cite{GPPR04}. 
The upper bound was later improved to $O(\sqrt{n})$~\cite{GU16}.
For unweighted planar graphs of diameter bounded by $\Delta$, we obtain a distance labeling scheme of $O(2^\Delta \log^2 n)$ bits, as a direct corollary of our tree cover result of \Cref{thm:planar-spanning} in conjunction with
known distance labeling schemes (e.g., \cite{Peleg00}) for trees. This result generalizes for unweighted minor-free graphs of bounded diameter.

For $(1+\eps)$-approximate distance labeling in unweighted planar graphs, the state-of-the-art is a labeling scheme of $O(\log n \cdot \log \log n \cdot \eps^{-1})$ bits \cite{Thorup04}. 
We obtain a $(1+\eps)$-approximate labeling scheme of $O(\log n \cdot \eps^{-3} \cdot \log^2(1/\eps))$ bits, as a direct corollary of our tree cover theorem of \Cref{thm:main} in conjunction with
the approximate distance labeling scheme for trees from \cite{FGNW17}.
This bound is optimal up to the $\eps$-dependence, even for trees \cite{FGNW17}.

\paragraph{Routing in planar metrics.~}
In compact routing schemes, the basic goal is to achieve efficient tradeoffs between the space usage --- the maximal number of bits per node --- and the stretch, which is the maximum ratio between distances of the
route used and the shortest path over all source-destination pairs.
There are two basic variants. In the {\em labeled} model,
the designer is allowed to assign nodes with short labels that can be used for routing. In the
{\em name-independent} model, the node labels are determined by an adversary. 
There is a large body of work on routing in metric spaces, including in restricted families of metrics, most notably in Euclidean and doubling metrics~\cite{DBLP:conf/podc/HassinP00, DBLP:conf/podc/AbrahamM04,Tal04, DBLP:conf/podc/Slivkins05,AGGM06,GR08Soda, CGMZ16}.
Many known routing protocols are rather complex, both conceptually and in implementation details. 

Our tree cover theorems (\Cref{thm:main,thm:treewidth,thm:planar-spanning})
essentially provide a reduction from the problem of routing in planar metrics to that of routing in tree metrics.
Our approach has several advantages over previous work. 
(1) {\em Simple}: Routing along trees is as basic as it gets. 
(2) {\em General}:
Our approach applies beyond planar metrics, to bounded treewidth metrics, to unweighted minor-free graphs of bounded diameter, and in principle to any graph family for which an efficient tree cover theorem exists. (3) {\em Low-hop}: By ``shortcutting'' the trees in our tree cover using the 2-hop sparse 1-spanners 
by Solomon~\cite{Solomon2013},
one can carry out the routing protocol on these spanners;
they are not as basic as tree metrics, but they 
achieve hop-diameter $2$ with size $O(n \log n)$.
As discussed and demonstrated in \cite{KLMS22}, a routing protocol in which the hop-distances are guaranteed to be bounded by 2 is advantageous.

We summarize our result on 2-hop routing in planar metrics.
\begin{theorem}\label{thm:routing}
For any $n$-point planar metric, one can construct a $(1+\eps)$-stretch $2$-hop routing scheme in the labeled, fixed-port model with headers of $\lceil\log{n}\rceil$ bits, labels of $O_\eps(\log^2 n)$ bits, local routing tables of $O_\eps(\log^2 n)$ bits, and local decision time $O(1)$.
\end{theorem}

\paragraph{Acknowledgement.~} Hung Le and Cuong Than are supported by the NSF CAREER Award No.\ CCF-2237288 and an NSF Grant No.\ CCF-2121952. Shay Solomon is funded by the European Union (ERC, DynOpt, 101043159).
Views and opinions expressed are however those of the author(s) only and do not necessarily reflect those of the European Union or the European Research Council.
Neither the European Union nor the granting authority can be held responsible for them.
Shay Solomon is also supported by the Israel Science Foundation (ISF) grant No.1991/1.
Shay Solomon and Lazar Milenković are supported by a grant from the United States-Israel Binational Science Foundation (BSF), Jerusalem, Israel, and the United States National Science Foundation (NSF).


\small
\bibliographystyle{alphaurl}
\bibliography{main}



\end{document}